\documentclass[journal,onecolumn]{IEEEtran}

\ifCLASSINFOpdf
\else
\fi

\usepackage{amsthm}

\usepackage{amssymb}
\usepackage{mathtools}
\usepackage{cases}
\usepackage{autobreak}

\usepackage{lipsum}
\usepackage{caption}

\usepackage{booktabs}
\usepackage{multirow}

\usepackage{color}
\usepackage{hyperref}
\usepackage[noadjust]{cite}

\usepackage{amsmath}
\usepackage{extarrows}

\usepackage{comment}

\usepackage{algorithmic}
\usepackage{algorithm}

\newcommand{\E}{\mathbb{E}}
\newcommand{\Prob}{\mathbb{P}}

\newcommand{\cov}{\mathrm{Cov}}
\newcommand{\Var}{\mathrm{Var}}

\newcommand{\FR}{{\bold{R}}}
\newcommand{\FT}{{\bold{T}}}
\newcommand{\FS}{{\bold{S}}}
\newcommand{\BQ}{{\bold{Q}}}
\newcommand{\BH}{{\bold{H}}}
\newcommand{\BI}{{\bold{I}}}
\newcommand{\BX}{{\bold{X}}}

\newcommand{\BZ}{{\bold{Z}}}
\newcommand{\BY}{{\bold{Y}}}

\newcommand{\BG}{{\bold{G}}}

\newcommand{\BA}{{\bold{A}}}
\newcommand{\BB}{{\bold{B}}}
\newcommand{\BC}{{\bold{C}}}
\newcommand{\BW}{{\bold{W}}}

\newcommand{\BO}{{\mathcal{O}}}

\allowdisplaybreaks

\DeclareMathOperator{\Tr}{Tr}

\newcommand{\RNum}[1]{\uppercase\expandafter{\romannumeral #1\relax}}
\newcommand\numberthis{\addtocounter{equation}{1}\tag{\theequation}}

\newtheorem{remark}{Remark}
\newtheorem{theorem}{Theorem}
\newtheorem{lemma}{Lemma}
\newtheorem{proposition}{Proposition}

\begin{document}
%
\title{Mutual Information Density of Massive MIMO Systems over Rayleigh-Product Channels}
%
%
%

\author{Xin~Zhang,~{\textit{Graduate Student Member,~IEEE}} and 
        Shenghui~Song,~\IEEEmembership{Senior Member,~IEEE}
\thanks{
This work was partially supported by a grant from the NSFC/RGC Joint Research Scheme sponsored by the Research Grants Council of the Hong Kong Special Administrative Region, China and National Natural Science Foundation of China (Project No. N\_HKUST656/22), and the Shenzhen Science and Technology Innovation Committee under Grant SGDX20210823103201006.

The authors are with the Department of Electronic and Computer Engineering, The Hong Kong University of Science and Technology, Hong Kong
(e-mail: xzhangfe@connect.ust.hk; eeshsong@ust.hk).

Copyright (c) 2017 IEEE. Personal use of this material is permitted.  However, permission to use this material for any other purposes must be obtained from the IEEE by sending a request to pubs-permissions@ieee.org.
}}

\maketitle

\begin{abstract} 
The Rayleigh-product channel model is utilized to characterize the rank deficiency caused by keyhole effects. However, the finite blocklength analysis for Rayleigh-product channels is not available in the literature. In this paper, we will characterize the mutual information density (MID) and perform the FBL analysis to reveal the impact of rank-deficiency in Rayleigh-product channels. To this end, we first set up a central limit theorem for the MID over Rayleigh-product MIMO channels in the asymptotic regime where the number of scatterers, number of antennas, and blocklength go to infinity at the same pace. Then, we utilize the CLT to obtain the upper and lower bounds for the packet error probability, whose approximations in the high and low signal to noise ratio regimes are then derived to illustrate the impact of rank-deficiency. One interesting observation is that rank-deficiency degrades the performance of MIMO systems with FBL and the fundamental limits of Rayleigh-product channels degenerate to those of the Rayleigh case when the number of scatterers approaches infinity.
\end{abstract}


%

\maketitle

%
%
%
%

\section{Introduction}
\label{introduction}
Multiple-input multiple-output (MIMO) becomes an essential technique in wireless communications to enhance the system throughput and reliability. With infinite blocklength (IBL), the mutual information (MI) of a MIMO system with channel matrix $\BH$ is defined as $C(\sigma^{2})=\log \det(\frac{\bold{H}\bold{H}^{H}}{\sigma^2}+\bold{I})$ with the signal-to-noise ratio (SNR) $\frac{1}{\sigma^2}$~\cite{telatar1999capacity}. The expectation of the MI (ergodic mutual information) and the probability that the MI is less than a given rate (outage probability) are utilized to measure the system throughput and reliability of MIMO systems over fading channels, respectively~\cite{chiani2003capacity,mckay2008mutual,zhong2011ergodic,li2013distribution}. Unfortunately, the IBL assumption is not valid in practice, and the corresponding performance metrics in a more practical setting with finite blocklength (FBL) are the second-order coding rate and packet error probability, respectively~\cite{hayashi2009information,polyanskiy2010channel}. 

In~\cite{polyanskiy2010channel}, the conventional Shannon's coding rate was refined to show that the maximal channel coding rate can be represented by a normal approximation
\begin{equation}
\log M(n,\varepsilon)=nC-\sqrt{nV}Q^{-1}(\varepsilon)+\BO(\log(n)),
\end{equation}
where $M(n,\varepsilon)$ is the cardinality of a codebook that has blocklength $n$ and can be decoded
with packet error probability less or equal to $\varepsilon$. In this formulation, $C$ is the channel capacity, $V$ is the channel dispersion, and $Q^{-1}(\cdot)$ denotes the inverse $Q$-function. The second-order asymptotics were first studied by Strassen~\cite{strassen1962asymptotische} and then revived by Hayashi using information spectrum methods~\cite{hayashi2009information,koga2002information}. 
In~\cite{hayashi2009information,polyanskiy2010channel}, the authors derived the second-order coding rate for the additive white Gaussian noise channel and the third-order term in the normal approximation was given in~\cite{tan2015third}. The closed-form expression for the coding rate were derived based on the asymptotic distribution of the mutual information density (MID)~\cite{polyanskiy2011scalar,yang2013quasi,tan2014asymptotic}. In fact, MID takes MI as its degenerated case when the blocklength has a higher order than the system dimensions, and its first-order and second-order statistics play key roles in the FBL analysis~\cite{polyanskiy2010channel,koga2002information,polyanskiy2011scalar,yang2013quasi,hayashi2009information,hoydis2015second,zhou2019lossy,collins2018coherent}. However, compared to the MI analysis, the evaluation of MID is more challenging due to the FBL constraint. In~\cite{collins2018coherent}, considering the multiple antenna block-fading channel, the authors utilized the Berry-Esseen theorem~\cite{berry1941accuracy,esseen1942limit} to show the asymptotic Gaussianity of the MID and derived the channel dispersion, which is an important quantity dominating the latency required to achieve the channel capacity and is highly related to the eigenvalue distribution of the gram channel matrix $\BH\BH^{H}$.

The above MI and MID analyses were performed for small-scale MIMO systems. With the increasing demands for high throughput and the development of antenna technology, the size of MIMO systems has been increasing constantly, e.g., massive MIMO. The large number of antennas brings further challenges to the performance analysis and system design with both IBL and FBL. Furthermore, the full-rank condition for the channel matrix may no longer hold when the numbers of antennas at the transceivers are very large. Fortunately, random matrix theory (RMT) has been proven effective in characterizing the asymptotic distribution of the MI and MID for large scale MIMO systems by the central limit theorem (CLT), and the results have been shown accurate even for small-scale systems~\cite{hachem2008new,hachem2008clt,hachem2012clt,bao2015asymptotic,hu2019central}.

\subsection{Characterizing the MI of MIMO Systems by RMT}

The MI of the full-rank MIMO channels has been characterized by setting up its CLT using RMT. In~\cite{kamath2005asymptotic}, Kamath~\textit{et al.} derived the closed-form expressions for the mean and variance of the MI over the i.i.d. MIMO fading channel. In~\cite{hachem2008new}, Hachem~\textit{et al.} derived the CLT for the MI of correlated Gaussian MIMO channels and gave the closed-form mean and variance. Hachem~\textit{et al.} extended the CLT to the non-Gaussian MIMO channel with a given variance profile and the non-centered MIMO channel in~\cite{hachem2008clt} and~\cite{hachem2012clt}, respectively, which shows that the pseudo-variance and non-zero fourth order cumulant of the random fading affects the asymptotic variance. In~\cite{bao2015asymptotic}, Bao~\textit{et al.} derived the CLT for the MI of independent and identically distributed (i.i.d) MIMO channels with non-zero pseudo-variance and fourth-order cumulant. In~\cite{hu2019central}, Hu~\textit{et al.} set up the CLT for the MI of elliptically correlated (EC) MIMO channels and validated the effect of the non-linear correlation. Considering the non-centered MIMO with non-separable correlation structure, the authors of~\cite{zhang2023fundamental} set up the CLT for the MI of holographic MIMO channels.

However, the full-rank MIMO channel is incapable of characterizing the reduced-rank behavior of MIMO systems due to the lack of scatterers around the transceivers~\cite{almers2006keyhole}. In particular, when the number of antennas exceeds the number the scatterers, the full-rank condition may not hold. To this end, RMT has also been utilized for the MI analysis of the rank-deficient channels, e.g. the Rayleigh-product channel and the double-scattering channel~\cite{gesbert2002outdoor,zhang2022large}. To this end, RMT has also been utilized for the MI analysis of the rank-deficient channels, e.g., the Rayleigh-product channel and the double-scattering channel~\cite{gesbert2002outdoor,zhang2022large}. In~\cite{zheng2016asymptotic}, the authors established the CLT for the MI of Rayleigh-product channels and gave a closed-form expression for the outage probability with equal number of antennas at the transceivers. In~\cite{zhang2022outage,zhang2022asymptotic}, the authors set up the CLT for the MI considering the channel correlation and unequal number of antennas at the transceivers. It has been shown that rank-deficiency will cause a larger fluctuation of the MI and a higher outage probability than the full-rank Rayleigh channel~\cite{zheng2016asymptotic,zhang2022outage}.

\subsection{Characterizing the MID of MIMO Systems by RMT}
RMT has also been utilized to characterize the MID for large scale MIMO systems. In~\cite{hoydis2015second}, the authors derived the CLT for the MID of i.i.d. MIMO Rayleigh channels by RMT and gave the closed-form expressions for the mean, variance of the MID, and the packet error probability, which explicitly reveals the impact of the number of antennas. In~\cite{you2022closed}, by assuming that the numbers of transmit and receive antennas go to infinity with the same pace, the close-form expressions of the mean and variance for the channel dispersion of i.i.d. MIMO Rayleigh channels were given by utilizing RMT. It is worthy noticing that different from the CLT for the MID (Berry-Esseen theorem) derived in~\cite{collins2018coherent}, the CLT established by RMT in~\cite{hoydis2015second} utilizes the concentration property of the random matrix (CLT for the linear statistics of eigenvalues) and can be used for the closed-form evaluation of the packet error rate. 

However, the MID analysis and the impact of rank-deficiency on the fundamental limits of Rayleigh-product channels with FBL has not been investigated. In this paper, we will characterize the MID of Rayleigh-product channels and analyze the impact of rank-deficiency on the packet error probability in the FBL regime.

\subsection{Challenges}
 Characterizing the MID for Rayleigh-product channels is very challenging due to several reasons. On the one hand, compared with the MID analysis for the (single) Rayleigh channel, setting up a CLT for the MID of Rayleigh-product channels needs to handle the fluctuations induced by two independent random matrices instead of only one. The difficulty has been shown in~\cite{gotze2017distribution} when extending the CLT for the linear spectral statistics of a single random matrix~\cite{lytova2009central} to that of the product of random matrices. A classical approach of setting up a CLT is to show the convergence of the characteristic function for the concerned statistic to that of the Gaussian distribution~\cite{lytova2009central}. However, as demonstrated in~\cite{gotze2017distribution,zhang2022asymptotic}, the evaluation of the characteristic function and the asymptotic variance of the MID are much involved and rely on the complex computation of the trace of the resolvent. Furthermore, the high SNR analysis is also more challenging for Rayleigh-product channels than that for the Rayleigh case since the key parameter is determined by a cubic equation instead of a quadratic equation~\cite{zheng2016asymptotic,zhang2022asymptotic}. On the other hand, compared with the MI analysis for two-hop channels~\cite{zhang2022asymptotic}, the MID expression has an additional term related to the coding scheme such that the covariance between the MI and the additional term, and the variance of the additional term need to be evaluated.

\subsection{Contributions}
To the best of the authors' knowledge, the characterization of the MID for massive MIMO systems over Rayleigh-product channels and the associated FBL analysis are not available in the literature. This paper is the first attempt to fill this research gap. The main contributions of this paper are summarized as follows:

\begin{itemize}

\item[(i)] We give the closed-form approximation for the cumulative distribution function (CDF) of the MID over Rayleigh-product MIMO channels with equal energy constraint. To this end, we show the asymptotic Gaussianity of the MID by setting up a CLT when the number of the antennas, the number of the scatterers, and the blocklength go to infinity with the same pace, and give the closed-form expressions for the asymptotic mean and variance. The CLT is proved by showing that the characteristic function of the MID converges to that of the Gaussian distribution by utilizing the Gaussian tools (the integration by parts formula and Nash-Poincar{\'e} inequality~\cite{hachem2008new,pastur2005simple}). Furthermore, we prove that the approximation error is $\BO(n^{-\frac{1}{4}})$. The CLT derived in this paper can degenerate to the CLT for MID over Rayleigh channels in~\cite[Theorem 2]{hoydis2015second} when the number of scatterers has a higher order than the number of antennas and blocklength.

\item[(ii)] Based on the CLT, we derive the closed-form approximation for the upper and lower bounds of the optimal average error probability over Rayleigh-product MIMO channels. Meanwhile, we show that the approximation error for the upper and lower bounds is $\BO(n^{-\frac{1}{2}})$ and the derived results can degenerate to existing results as summarized in Table~\ref{ilu_intro}. In particular, when the number of scatterers has a higher order than the number of antennas and blockelngth, the results in this paper resort to those of the Rayleigh channel~(\cite{hoydis2015second} in Table~\ref{ilu_intro}). On the other hand, when the blocklength has a higher order than the number of antennas and scatterers, the packet error probability converges to the outage probability~(\cite{zheng2016asymptotic,zhang2022outage,zhang2022asymptotic} in Table~\ref{ilu_intro}). This phenomenon also happens in single-hop MIMO channels~\cite{yang2014quasi}. Furthermore, the derived result with FBL over Rayleigh-product channels can also degenerate to that with IBL over Rayleigh channels~\cite{hachem2008new,kamath2005asymptotic} when both the number of scatterers and blocklength have higher order than the number of antennas. To evaluate the impact of the number of scatterers, we give the high and low SNR approximations for the derived upper and lower bounds, which indicate that the Rayleigh-product MIMO channel has a larger error probability than the Rayleigh MIMO channel.

\item[(iii)] Numerical results validate the accuracy of the derived upper and lower bounds. It is shown that the gap between the upper and lower bounds for the packet error rate is small in the practical SNR regime and the slope of the bounds for the optimal average packet error probability matches well with that for the packet error probability of the LDPC codes. Additionally, given the same rate, the packet error probability of Rayleigh-product channels is higher than that of Rayleigh channels, and the gap decreases as the rank of the Rayleigh-product channel increases, which agrees with the theoretical results.

\end{itemize}

\textit{Paper Organization:}~Section~\ref{sec_mod} introduces the MIMO system over Rayleigh-product channels and the definition of MID. Section~\ref{sec_main} presents the CLT to characterize the asymptotic distribution of the MID. Section~\ref{sec_FBL_ana} gives the the closed-form expressions for the upper and lower bounds of the average error probability. Section~\ref{sec_simu} shows the numerical results and Section~\ref{sec_con} concludes the paper.

\textit{Notations:} Bold, upper case letters and bold, lower case letters represent matrices and vectors, respectively. The probability and expectation operator are denoted by $\Prob(\cdot)$ and $\E[\cdot]$, respectively. The $N$-dimensional vector space and $M$-by-$N$ matrix space are represented by $\mathbb{C}^{N}$ and $\mathbb{C}^{M\times N}$. The conjugate transpose, spectral norm, and trace of $\bold{A}$ are denoted by $\bold{A}^{H}$, $\|\BA \|$, and $\Tr(\BA)$, respectively. The $(i,j)$-th entry of $\bold{A}$ is denoted by $[\BA]_{i,j}$ or $A_{ij}$. The conjugate of a complex number is represented by $(\cdot)^{*}$ and the $N$ by $N$ identity matrix is denoted by $\bold{I}_{N}$. The CDF of the standard normal distribution is denoted by $\Phi(x)$.  The centered random variable is represented as $\underline{x}=x-\E [x]$ and the covariance between $x$ and $y$ is denoted by $\cov(x,y)=\E[\underline{x}\underline{y}] $. The convergence in distribution is represented by $\xrightarrow[N \rightarrow \infty]{\mathcal{D}}$ and the limit as $a$ approaches $b$ from the right is denoted by $a  \downarrow b$. The polynomial with positive coefficients and the support operator  are represented by $\mathcal{P}(x)$ and $\mathrm{supp}(\cdot)$, respectively. Given a set $\mathcal{S}$, $\mathrm{P}(\mathcal{S})$ denotes the set of probability measures with support of a subset of $\mathcal{S}$. The big-O, little-O, and big-theta notations are represented by $\BO(\cdot)$, $o(1)$, and $\Theta(\cdot)$, respectively.

\begin{table}
\caption{Summary of Related Works.}
\label{ilu_intro}
\centering
\begin{tabular}{|c|c|c|}
\hline
& IBL Analysis & FBL Analysis \\
\hline 
Rayleigh Channel &~\cite{kamath2005asymptotic,hachem2008new} &~\cite{hoydis2015second} \\
\hline
Rayleigh-product Channel&~\cite{zheng2016asymptotic,zhang2022outage,zhang2022asymptotic}& This work
\\
\hline
\end{tabular}
\vspace{-0.3cm}
\end{table}

\section{System Model and MID}
\label{sec_mod}
In this paper, we consider a MIMO system with $N$ receive antennas and $M$ transmit antennas. The received signal $\bold{r}_t\in \mathbb{C}^{N}$ (output of the channel) is given by
\begin{equation}
\label{sig_mod}
\bold{r}_{t}={\BH}\bold{x}_{t}+\sigma\bold{w}_{t}, ~~t=1,2,...,n,
\end{equation}
where $\bold{x}_{t}\in\mathbb{C}^{M}$ represents the transmit signal (input of the channel), ${\BH}\in\mathbb{C}^{N\times M}$ denotes the channel matrix, and $\bold{w}_{t} \in\mathbb{C}^{N}$, whose entries follow $\mathcal{CN}(0,1)$, represents the normalized~additive white Gaussian noise at time $t$. Thus, $\sigma^2$ denotes the noise power. The blocklength, i.e., the number of the channel uses required to transmit a codeword, is defined as $n = L_c/M$, where $L_c$ represents the code-length.

Due to channel fading, $\BH$ is generally modeled as a random matrix. For example, for the Rayleigh channel, $\BH_{r}$ is an i.i.d. random matrix whose entries follow the complex Gaussian distribution $\mathcal{CN}(0,\frac{1}{M})$. The Rayleigh-product channel, which can be regarded as the product of two Rayleigh channels, is modeled by~\cite{jin2008transmit}
\begin{equation}
\label{cha_mod}
{\BH}={\BZ}{\BY},
\end{equation}
where the entries of $\BZ\in\mathbb{C}^{N\times L}$ and $\BY\in\mathbb{C}^{L\times M}$ follow $\mathcal{CN}(0,\frac{1}{L})$ and $\mathcal{CN}(0,\frac{1}{M})$, respectively, with $L$ denoting the number of scatterers\footnote{In this work, we consider the i.i.d. channel to investigate the impact the number of scatterers on the error probability in the FBL regime. Such an i.i.d model happens when antennas and scatterers are sufficiently separated so that there is no spatial correlation among antennas or scatterers~\cite{zheng2016asymptotic}.}. The rank-deficiency of the channel is reflected by $L < \min \{ N,M\}$. When the number of scatterers goes to infinity, the Rayleigh-product channel approaches Rayleigh channel, i.e., when $L\rightarrow \infty$, the probability density function of $\BH$ approaches that of $\BH_{r}$~\cite{gesbert2002outdoor}. In this paper, we consider the quasi-statistic channel where $\BH$ does not change in $n$ channel uses. We assume the transmitter has the statistical knowledge of $\BH$ while the receiver has perfect channel state information (CSI). For ease of illustration, we introduce the following notations:
$\BX^{(n)}=(\bold{x}_1,\bold{x}_2,...,\bold{x}_n)$, $\FR^{(n)}=(\bold{r}_1,\bold{r}_2,...,\bold{r}_n)$, and $\BW^{(n)}=(\bold{w}_1,\bold{w}_2,...,\bold{w}_n)$. 

\subsection{Mutual Information Density (MID)}
MID, also called~\textit{information density}, plays an important role in the FBL analysis as the error probability can be bounded by the CDF of the MID~\cite{polyanskiy2010channel,collins2018coherent,zhou2018dispersion}. The per-antenna MID of the MIMO system considered in~(\ref{sig_mod}) can be adapted from the single antenna case as
\begin{equation}
\label{mid_def}
I_{N,L,M}^{(n)}=\frac{1}{Mn}\log\left( \frac{\mathrm{d}  \mathbb{P}_{\FR^{(n)}|\BX^{(n)},\BH}(\FR^{(n)} |\BX^{(n)},\BH)}{ \mathrm{d}  \mathbb{P}_{\FR^{(n),+}|\BH}(\FR^{(n)}|\BH)} \right),
\end{equation}
where $\frac{\mathrm{d} P}{\mathrm{d}  Q}$ denotes the Radon-Nikodym derivative of measure $P$ with respect to measure $Q$~\cite{polyanskiy2010channel}. Here $\bold{\FR}^{(n),+}|\BH$ represents the channel output following capacity achieving output distribution (the distribution induced by a capacity achieving input distribution, e.g., Gaussian distribution~\cite[Section III.A]{collins2018coherent}) with its column $\bold{r}^{+}_{i}|\BH \sim \mathcal{CN}(\bold{0}_{N},\BI_{N}+\sigma^{-2}\BH\BH^{H},\bold{0}_{N\times N}) $, $i=1,2,...,n$. The output distribution of $\bold{r}_i$ conditioning on $\BH$ and $\bold{x}_i$, $\bold{r}_i |(\BH,\bold{x}_i) $, follows $\mathcal{CN}(\BH\bold{x}_{i}, \sigma^2\BI_N,\bold{0}_{N\times N})$~\cite[Eq. (87)]{hoydis2015second},~\cite[Eq. (29)]{collins2018coherent}. In this case, the MID in~(\ref{mid_def}) can be written as
\begin{align*}
\label{mid_exp}
& I_{N,L,M}^{(n)}\overset{\bigtriangleup}{=}\frac{1}{M}\log\det(\bold{I}_{N}+\frac{1}{\sigma^2}\BH\BH^{H}) 
+\frac{1}{Mn} \times
\\
&
\Tr((\BH\BH^{H}+\sigma^2\bold{I}_{N})^{-1}
(\BH\BX^{(n)}\!+\!\sigma \BW^{(n)})(\BH\BX^{(n)}\!+\!\sigma \BW^{(n)})^{H} )
\\
&
-\frac{1}{Mn}\Tr(\BW^{(n)}(\BW^{(n)})^{H}).\numberthis
\end{align*}
The distribution of both MI and MID can be used to investigate the bounds for the error probability. Specifically, the distribution of the MI (capacity without CSI at the transmitter) can be utilized to investigate the outage probability, which is the bound for the error probability in the IBL regime. It can be observed from~(\ref{mid_exp}) that MID is the sum of three terms, composed of three random matrices $\BZ$, $\BY$, and $\BW^{(n)}$. The first term is the per-antenna MI (per-antenna capacity), the second is a resolvent related term, and the third is the noise term. In the following, we will set up a CLT to investigate the distribution of MID in~(\ref{mid_exp}) and derive the mean and variance of MID. The main results in this paper are based on the following assumption.

 \textbf{Assumption A.} (Asymptotic Regime) $0<\lim\inf\limits_{M \ge 1} \frac{M}{L} \le \frac{M}{L}  \le \lim \sup\limits_{M \ge 1} \frac{M}{L} <\infty$, $0<\lim \inf\limits_{M \ge 1}  \frac{M}{N} \le \frac{M}{N}  \le \lim \sup\limits_{M \ge 1}  \frac{M}{N} <\infty$, $0<\lim\inf\limits_{M \ge 1}  \frac{M}{n} \le \frac{M}{n}  \le \lim \sup\limits_{M \ge 1} \frac{M}{n} <\infty$.
\textbf{Assumption A} assumes that $M$, $N$, $L$, and $n$ go to infinity and keeps the relative relation of these parameters through the fixed ratios $\rho$, $\eta$, and $\kappa$. The derived results in the regime can be regarded as a large system approximation, which has been widely used in the performance analysis over MIMO systems and validated to be numerically accurate for moderate-scale systems~\cite{zhang2022outage,hoydis2011asymptotic,hachem2008new,bao2015asymptotic}. In fact, the asymptotic regime in~\textbf{Assumption A} is only required for the asymptotic analysis and not required for the practical operation. Different from existing works, we will also analyze the error term for the approximation instead of just showing the convergence. Denote the ratios $\eta=\frac{N}{M}$, $\rho=\frac{n}{M}$, and $\kappa=\frac{M}{L}$ and let $n  \xrightarrow[]{\rho,\eta, \kappa}\infty$ represent the asymptotic regime where $n$, $N$, $M$, and $L$ grow to infinity with the fixed ratios $\rho$, $\eta$, and $\kappa$. In the traditional IBL analysis for massive MIMO systems over Rayleigh-product channels~\cite{zhang2022asymptotic,zhang2022outage}, it assumes $\rho  \xrightarrow[]{\eta, \kappa}\infty $. It is worth noticing that when $\kappa\rightarrow 0 $ with fixed $\eta$ and $\rho$ (denoted as $\kappa \xrightarrow[]{\eta,\rho} 0$, which means that $L$ has a higher order than $M$, $N$, and $n$), the Rayleigh-product channel degenerates to the Rayleigh channel.

\section{MID Characterization}
\label{sec_main}
In this section, we will characterize the distribution of the MID. To this end, we first introduce some existing results regarding the first-order analysis for the MI of Rayleigh-product channels, i.e., the closed-form expression for the ergodic MI (EMI) in~\cite{zhang2022asymptotic,hoydis2011asymptotic}.
\subsection{Preliminary Results}
\begin{theorem} 
\label{the_erc}
(EMI of Rayleigh-product channels) Given \textbf{Assumption~A} and the channel matrix $\BH$ defined in~(\ref{cha_mod}), the per antenna mutual information is given by 
\begin{equation}
\label{MI_exp}
C(\sigma^2)=\frac{1}{M}\log\det\left(\bold{I}_N +\frac{1}{\sigma^2}\BH\BH^{H}\right)
\end{equation}
and there holds true that~\cite[Theorem 2]{hoydis2011asymptotic}
\begin{equation}
C(\sigma^2)\xrightarrow[N \xrightarrow{\eta, \kappa} \infty]{a.s.} \overline{C}(\sigma^2),
\end{equation}
and~\cite[Proposition 1]{zhang2022asymptotic}
\begin{equation} 
\label{C_appro}
\E [C(\sigma^2)] \xrightarrow[]{N \xrightarrow[]{\eta, \kappa} \infty}\overline{C}(\sigma^2)+\BO(\frac{1}{M^2}),
\end{equation}
where $\overline{C}(\sigma^2)$ is given by
\begin{equation}
\label{cap_app}
\begin{aligned}
 \overline{C}(\sigma^2) &=-\frac{\log(\sigma^2)}{\kappa}-\left(\eta-\frac{1}{\kappa}\right)\log\left(1-\frac{\omega}{\eta(1+\omega)}\right)
\\
&
+\log(1+\omega)-\frac{\log(\omega)}{\kappa}
-\frac{2\omega}{1+\omega}
+\frac{\log(\eta)}{\kappa}.
\end{aligned}
\end{equation}
Here $\omega$ is the root of the following cubic equation
\begin{equation}
\label{cubic_eq}
\begin{aligned}
P(\sigma^2)&=\omega^3+\left(2\sigma^2+\eta\kappa-\kappa-\eta+1 \right)\frac{\omega^2}{\sigma^2}
\\
&
+\left(1+\frac{\eta\kappa}{\sigma^2}-\frac{2\eta}{ \sigma^2}+\frac{1}{\sigma^2}\right)\omega-\frac{\eta}{\sigma^2}=0,
\end{aligned}
\end{equation}
which satisfies $\omega>0$ and $\eta+(\eta-1)\omega>0$. 
\end{theorem}
Theorem~\ref{the_erc} presents a closed-form approximation for $\E [C(\sigma^2)]$ with error $\BO(\frac{1}{M^2})$, which will be utilized to evaluate the asymptotic mean of the MID in Theorem~\ref{clt_the}. It is worthy of noticing that different from the EMI of Rayleigh channels~\cite[Eq. (9)]{verdu1999spectral}, the key parameter for characterizing the EMI, i.e., $\omega$, is the root of a cubic equation instead of a quadratic equation. For ease of illustration, we introduce the following notations:
\begin{equation}
\label{delta_def}
\begin{aligned}
\delta=\frac{1}{\sigma^2}\left(\frac{N}{L}-\frac{M\omega}{L(1+\omega)} \right),
~~\overline{\omega}=\frac{1}{1+\omega}.
\end{aligned}
\end{equation}
The derivatives of $\delta$, $\omega$, $\overline{\omega}$, and $\overline{C}(\sigma^2)$ with respect to $\sigma^2$ are given by
\begin{subequations}
\label{delta_deri}
\begin{align}
\delta' &=\frac{\mathrm{d} \delta}{\mathrm{d}\sigma^2} =-\frac{\delta }{\Delta_{\sigma^2}},\label{delta_deri1}\\
\omega' &=\frac{\omega\delta'}{\delta(1+\delta\overline{\omega}^2)},\label{delta_deri2}\\
\overline{\omega}'& =-\frac{\omega\overline{\omega}^2\delta'}{\delta(1+\delta\overline{\omega}^2)},\label{delta_deri3}\\
\overline{C}'(\sigma^2)& =\frac{\delta}{\kappa}-\frac{\eta}{\sigma^2},\label{delta_deri4}
\end{align}
\end{subequations}
where
\begin{equation}
\label{def_Delta}
\Delta_{\sigma^2}=\sigma^2+ \frac{M\omega\overline{\omega}^2}{L\delta(1+\delta \overline{\omega}^2)}.
\end{equation}
The above results can be obtained by taking derivative on both sides of~(\ref{cubic_eq}) and the proof is given in the extended version of this paper~\cite[Appendix A]{zhang2022second}. The asymptotic distribution of the MID induced by the code with the equal energy constraint (sphere constraint)  
\begin{equation}
\label{sph_cons}
\begin{aligned}
\mathcal{S}^{n}_{=}=\{\BX^{(n)} \in\mathbb{C}^{M\times n }| \frac{\Tr(\BX^{(n)}(\BX^{(n)})^{H})}{Mn} = 1  \}
\end{aligned}
\end{equation}
is given by the following theorem.

\subsection{CLT for the MID}
\begin{theorem}
\label{clt_the}
 (CLT of the MID)
Given~\textbf{Assumption~{A}} and any sequence of $\BX^{(n)}\in \mathcal{S}^{(n)}_{=} $, the distribution of the MID converges to a Gaussian distribution, i.e.,
\begin{equation}
\sqrt{\frac{{Mn}}{V_n}}(I_{N,L,M}^{(n)}-\overline{C}(\sigma^2))  \xrightarrow[{n  \xrightarrow[]{\rho,\eta, \kappa}\infty}]{\mathcal{D}}  \mathcal{N}(0,1),
\end{equation}
where $V_n$ is given by
\begin{equation}
\begin{aligned}
\label{clt_var}
V_n &=-\rho\log(\Xi)+\eta+\frac{\sigma^4\delta'}{\kappa}
+\frac{\rho\kappa \overline{\omega}^4\Tr(\BA_{n}^2)}{M}
 \\
& \times
\left(\frac{\omega^2(1+\delta\overline{\omega})}{1+\delta\overline{\omega}^2}-\frac{\omega\omega'}{\delta(1+\delta \overline{\omega}^2)}\right),
\end{aligned}
\end{equation}
with 
\begin{equation}
\label{def_A}
 \BA_{n}=\bold{I}_M-\frac{1}{n}\BX^{(n)}(\BX^{(n)})^H,
 \end{equation}
 \begin{equation}
\Xi=\frac{(1+\delta\overline{\omega}^2)\delta(\sigma^2+ \frac{\kappa\omega\overline{\omega}^2}{\delta(1+\delta \overline{\omega}^2)})}{\eta\kappa(1+\delta\overline{\omega})},
\end{equation}
where $\omega$ is same as that in Theorem~\ref{the_erc} and $\delta$, $\overline{\omega}$, $\omega'$, $\delta'$, and $\Delta_{\sigma^2}$ can be obtained by~(\ref{delta_def})-(\ref{def_Delta}). Furthermore, there holds true that
 \begin{equation}
 \label{prob_con_rate}
\Prob\left( \! \sqrt{\frac{{Mn}}{V_n}}(I_{N,L,M}^{(n)}-\overline{C}(\sigma^2))\! \le\! x \!  \right) \!=\! \Phi(x)+\BO(n^{-\frac{1}{4}}).
\end{equation}
\end{theorem}
\begin{proof}
The proof of Theorem~\ref{clt_the} is given in Appendix~\ref{prof_the2}.
\end{proof}
Note that the CLT in Theorem~\ref{clt_the} is derived conditioning on $\BX^{(n)}$. In fact, it holds true for any sequence $\BA_{n}$ with $\BX^{(n)}\in \mathcal{S}^{(n)}_{=}$, which indicates the CLT is valid for any probability measure so that we can establish the upper and lower bounds for the error probability by constructing $\BA_n$ ($\BX_n$) later in Theorem~\ref{the_oaep}. It can be observed that the per-antenna MI in~(\ref{MI_exp}) is the first term of the MID in~(\ref{mid_exp}) and the variance of MI is $-\log(\Xi)$~\cite{zhang2022asymptotic,zhang2022outage}. According to Theorem~\ref{clt_the} and~\cite[Theorem 2]{zhang2022asymptotic}, the asymptotic distribution of both MI and MID are Gaussian distribution but the asymptotic variances of MI and MID are different since MID takes the impact of the noise in the FBL regime into account, i.e.,  $\BW^{(n)}$.

Theorem~\ref{clt_the} is the first result regarding the distribution of MID over Rayleigh-product MIMO channels and the expressions for the mean $\overline{C}(\sigma^2)$ and variance $V_{n}$ are given by closed forms. With Theorem~\ref{clt_the}, we can approximate the CDF of the MID using that of the Gaussian distribution and perform the FBL analysis over Rayleigh-product MIMO channels later in Section~\ref{sec_FBL_ana}. Theorem~\ref{clt_the} tackles the product structure of random matrices in Rayleigh-product MIMO channels and takes the impact of number of antennas, scatterers, and blocklength into account. The CLT for MID over Rayleigh MIMO channels in~\cite[Theorem 2]{hoydis2015second} is a special case of Theorem~\ref{clt_the} when $\kappa \xrightarrow[]{\eta,\rho} 0$. Furthermore, Theorem~\ref{clt_the} gives the quantitative evaluation on the approximation error $\BO(n^{-\frac{1}{4}})$ in~(\ref{prob_con_rate}), which is based on the general condition $M^{-1}\Tr(\BA_{n}^2)=\BO(M)$ resulting from the fourth-order moment $\Tr((n^{-1}\BX^{(n)}\BX^{(n),H})^2)=\BO(M^2)$. Existing works on the CLT for MID and MI over MIMO channels~\cite{zheng2016asymptotic,zhang2022outage,zhang2022asymptotic,hachem2008new} did not consider the approximation error of the CDF. Next we show how Theorem~\ref{clt_the} degenerates to~\cite[Theorem 2]{hoydis2015second} by the following remark.

\begin{remark} (Degeneration to the Rayleigh case) When $\kappa \xrightarrow[]{\eta,\rho} 0$, equation~(\ref{cubic_eq}) becomes 
\begin{equation}
\label{qua_eq}
\omega^2+\left(1-\frac{\eta}{\sigma^2}+\frac{1}{\sigma^2}\right)\omega -\frac{\eta}{\sigma^2}=0,
\end{equation}
\noindent whose solution $\omega_{\infty}$ is $\omega_{\infty}=\delta_0=\frac{\eta-1-2\sigma^2+\sqrt{(1-\eta+\sigma^2)^2+4\eta\sigma^2}}{2\sigma^2}$ such that $\omega\xrightarrow[]{\kappa \xrightarrow[]{\eta,\rho} 0}\omega_{\infty}=\delta_0$ and $\omega' \xrightarrow[]{\kappa \xrightarrow[]{\eta,\rho} 0}\omega'_{\infty}=\delta'_0$. We can further obtain from~(\ref{delta_def}) that
\begin{equation}
\frac{\delta}{\kappa}\xrightarrow[]{\kappa \xrightarrow[]{\eta,\rho} 0}\frac{1}{\sigma^2}(\eta-\frac{\omega_{\infty}}{1+\omega_{\infty}})\overset{(a)}{=}\omega_{\infty},
\end{equation}
where $(a)$ follows from~(\ref{qua_eq}). We can rewrite the mean $\overline{C}(\sigma^2)$ as
\begin{align*}
\label{degene_mean}
\overline{C}(\sigma^2)& =\eta\log(1+\frac{\kappa \omega\overline{\omega}}{\sigma^2\delta})
+\frac{1}{\kappa}\log(1+\delta\overline{\omega})
\\
&
+\log(1+\omega)-2\omega\overline{\omega}
\xrightarrow[]{\kappa \xrightarrow[]{\eta,\rho} 0} \eta\log(1+\frac{1}{\sigma^2(1+\delta_{0})})
\\
&
+\log(1+\delta_{0})-\frac{\delta_{0}}{1+\delta_{0}}.\numberthis
\end{align*}
For the variance, we have 
\begin{equation}
\Xi \xrightarrow[]{\kappa \xrightarrow[]{\eta,\rho} 0} \frac{\delta_{0}(\sigma^2+\frac{1}{(1+\delta_{0})^2})}{\eta}=1-\frac{\delta_{0}^2}{\eta(1+\delta_{0})^2},
\end{equation}
and
\begin{equation}
\frac{\kappa\omega^2(1+\delta\overline{\omega})}{1+\delta\overline{\omega}^2}-\frac{\kappa\omega\omega'}{\delta(1+\delta \overline{\omega}^2)}
 \xrightarrow[]{\kappa \xrightarrow[]{\eta,\rho} 0} 0-\delta_{0}',
\end{equation}
such that 
\begin{equation}
\label{degenevar}
V_n \xrightarrow[]{\kappa \xrightarrow[]{\eta,\rho} 0} -\rho\log(1-\frac{\delta_{0}^2}{\eta(1+\delta_{0})^2})+\eta+\sigma^4\delta'_0-\frac{\rho\delta_{0}' \Tr(\BA_{n}^2)}{M(1+\delta_{0})^4}.
\end{equation}
The right hand side of~(\ref{degene_mean}) and~(\ref{degenevar}) are identical to~\cite[Eqs. (12) and (20)]{hoydis2015second}, respectively, which indicates that Theorem~\ref{clt_the} is equivalent to that for the Rayleigh channel when $\kappa \xrightarrow[]{\eta,\rho} 0$.
\end{remark}

\section{Bounds for the Optimal Average Error Probability}
\label{sec_FBL_ana}

In this section, we will utilize the CLT in Theorem~\ref{clt_the} to derive the upper and lower bounds for the optimal average error probability with the maximal energy constraint. In the following, we first give the definitions of the metrics and then perform the FBL analysis. 

\subsection{Performance Metrics}

\textit{The encoder mapping}:
A $(\mathrm{P}_{\mathrm{e}}^{(n)},G)$-code for the system in~(\ref{sig_mod}) can be represented by the following mapping $f$,
\begin{equation}
f:\mathcal{G} \rightarrow \mathbb{C}^{M\times n}.
\end{equation}
Here the transmitted symbols are denoted by $\BX_{m}^{(n)}=f(m)\in \mathcal{S}^{n}$, where 
\begin{equation}
\label{max_cons}
\begin{aligned}
\mathcal{S}^{n}=\{\BX^{(n)} \in\mathbb{C}^{M\times n }| \frac{\Tr(\BX^{(n)}(\BX^{(n)})^{H})}{Mn} \le 1  \}
\end{aligned}
\end{equation}
denotes the maximal energy constraint and $m$ is uniformly distributed in $\mathcal{G}=\{1,2,..,G \}$. Here $\mathcal{C}_{n}$ is the codebook, i.e., $\{f(1),f(2),...,f(G)\}$. There are  two types of constraints on the channel inputs, i.e., the maximal energy constraint in~(\ref{max_cons}) and the equal energy constraint (sphere constraint) in~(\ref{sph_cons}). Obviously, the inputs satisfying the equal energy constraint is a subset of those following the maximal energy constraint. Our goal is to obtain the bounds for the error probability with the maximal energy constraint in~(\ref{max_cons}).

\textit{The decoder mapping}: The decoder mapping from the channel output $\FR^{(n)}=\BH f(m)+\sigma \BW^{(n)}$ to the message can be represented by
\begin{equation}
g:\mathbb{C}^{N\times n} \rightarrow \mathcal{G} \cup \{  \mathrm{e}   \}.
\end{equation}
The mapping $g$ gives the decision of $\hat{m}=g(\FR^{(n)})$, i.e., the decoder picks the transmitted message if it is correctly decoded otherwise an error $e$ occurs. 

\noindent\textit{A.1. Average Error Probability}

Since $m$ is assumed to be uniformly distributed, the \textit{average error probability} for a code $\mathcal{C}_n$ with blocklength $n$, encoder $f$, decoder $g$, and input $G$ is given by
\begin{equation}
\label{pe_def}
\mathrm{P}_{\mathrm{e}}^{(n)}(\mathcal{C}_n)=\frac{1}{G}\sum_{i=1}^{G}\Prob( \hat{m}\neq m | m=i),
\end{equation}
where the evaluation involves the randomness of $\BH$, $\BW^{(n)}$, and $m\in \mathcal{G}$. The optimal average error probability is given by
\begin{equation}
\label{pro_e_ori}
\begin{aligned} 
\mathrm{P}_{\mathrm{e}}^{(n)}(R)=\inf_{\mathrm{supp}(\mathcal{C}_n)\subseteq \mathcal{S}^{n}}
\mathrm{ P}_{\mathrm{e}}^{(n)}(\mathcal{C}_n) ,
\end{aligned}
\end{equation}
where $R$ denotes the per-antenna rate of each transmitted symbol and $\frac{1}{nM}\log(|\mathcal{C}_n|) \ge R $. 

Unfortunately, it is very difficult to obtain the exact expression of the optimal average error probability for any $M$, $L$, $N$, and $R$. In particular, we consider the rate $R$ within $\BO(\frac{1}{\sqrt{n}})$ of the ergodic capacity in the regime that $n$, $M$, $N$, $L$ go to infinity with the same pace, which is referred to as the~\textit{second-order coding rate}~\cite{hayashi2009information,polyanskiy2010channel}. Thus, we will consider the optimal average error probability with respect to a perturbation $r$ around the ergodic capacity in the FBL regime (Here we induce the ergodicity in the inherently non-ergodic quasi-static fading channel by increasing the channel matrix dimensions such that the per-antenna rate converges to the ergodic-capacity rate, i.e., the ergodicity is over the space instead of the time dimension). To handle the difficulty caused by the randomness of Rayleigh-product MIMO channels, we will back off from the infinity by assuming that $n$, $M$, $N$, $L$ go to infinity with the same pace to obtain the closed-form evaluation for the optimal average error probability. This asymptotic regime has been widely used in evaluating the performance of large MIMO systems~\cite{hachem2008new,zheng2016asymptotic,zhang2022asymptotic} and the strikingly simple expressions for the asymptotic performance have been validated to be accurate even for the small-scale systems. 

\noindent\textit{A.2. The Optimal Average Error Probability with Respect to the Second-Order Coding Rate}

The analysis for single antenna system considered the rate within $\BO(\frac{1}{\sqrt{n}})$ of the capacity~\cite{hayashi2009information}, where the code length is equal to the blocklength. In MIMO systems, the code length is $Mn$ so we need to consider the rate within $\BO(\frac{1}{\sqrt{Mn}})$ of per-antenna capacity $C(\sigma^2)$~\cite{hoydis2015second,collins2018coherent}. Given a second-order coding rate $r$ for the system with $M$ transmit antennas and blocklength $n$, i.e.~\cite{polyanskiy2010channel,yang2013quasi,tomamichel2014second,le2015case,zhou2018dispersion},
\begin{align}
\label{se_code_rate}
\liminf_{ n  \xrightarrow[]{\rho,\eta, \kappa}\infty}\frac{1}{\sqrt{M n}} \{\log(|\mathcal{C}_n|)- Mn \E [C(\sigma^2)]\}  \ge r,
\end{align}
the optimal average error probability is given by~\cite{hayashi2009information,hoydis2015second}
\begin{align}
   \label{def_oaep}
& \Prob_{\mathrm{e}}(r| \rho,\eta, \kappa )
{=}\liminf\limits_{\mathrm{supp}(\mathcal{C}_n)\subseteq \mathcal{S}^{n} }\limsup_{n  \xrightarrow[]{\rho,\eta, \kappa}\infty} \mathrm{P}_{\mathrm{e}}^{(n)}(\mathcal{C}_n),
\end{align}
where $C(\sigma^2)=\frac{1}{M}\log\det(\bold{I}_{N}+\frac{1}{\sigma^2}\BH\BH^{H})$ denotes the per antenna capacity and $\mathrm{P}_{\mathrm{e}}^{(n)}(\mathcal{C}_n)$ in~(\ref{pe_def}) represents the average error probability of code $\mathcal{C}_n$. From~(\ref{pro_e_ori}) and~(\ref{def_oaep}), we can observe that for $r=\BO(1)$, the per antenna rate $R=\frac{\log(|\mathcal{C}_n|)}{Mn}=\E [C(\sigma^2)] +\frac{r}{\sqrt{Mn}}$ is a $\BO(\frac{1}{\sqrt{Mn}})$ perturbation around $\E [C(\sigma^2)]$.

According to the convergence $\E[ C(\sigma^2)]-\overline{C}(\sigma^2)=\BO(\frac{1}{M^2})$ in~(\ref{C_appro}) of Theorem~\ref{the_erc}, we have $\sqrt{Mn}(\E [C(\sigma^2)]-\overline{C}(\sigma^2))=\BO(\frac{1}{M})$ such that we can replace $\E [C(\sigma^2)]$ by $\overline{C}(\sigma^2)$ in~(\ref{se_code_rate}). The convergence rate $\BO(\frac{1}{M^2})$ is proved in~\cite{zhang2022asymptotic} for the cases where $\BZ$ and $\BY$ are both Gaussian matrices, and the convergence rate may not be valid for non-Gaussian matrices. The analysis of non-Gaussian matrices indicates that  $\sqrt{Mn}(\E[ C(\sigma^2)]-\overline{C}(\sigma^2))=\BO(1)$ for single-hop channels~\cite{zhang2021bias,bao2015asymptotic} when the entry of the channel matrix has a non-zero pseudo-variance or fourth-order cumulant, which is referred to as the bias. The bias also exists for the two-hop channel if the entries of $\BZ$ and $\BY$ have non-zero pseudo-variance or fourth-order cumulant. 

However, Theorem~\ref{clt_the} is based on the equal energy constraint in~(\ref{sph_cons}) while our goal is to obtain the result with the maximal energy constraint. Fortunately, the error probability with the maximal energy constraint is shown to be bounded by that with the equal energy constraint through the following lemma.

\begin{lemma} 
\label{bnd_err}
(Bounds for the optimal average error probability)\cite[Eq. (77) and Eq. (89)]{hoydis2015second} The optimal average error rate can be bounded by
\begin{equation}
\mathbb{F}(r|\rho,\eta, \kappa )\le \Prob_{\mathrm{e}}(r|\rho,\eta, \kappa  ) \le \mathbb{G}(r|\rho,\eta, \kappa ),
\end{equation}
where 
\begin{subequations}
\begin{align}
&\mathbb{G}(r|\rho,\eta, \kappa  )=\lim_{\zeta  \downarrow 0} \limsup\limits_{ N \xrightarrow[]{\rho,\eta, \kappa} \infty}\nonumber
\\
&
 \Prob[\sqrt{Mn}(I^{(n)}_{N,L,M}-\overline{C}(\sigma^2))\le r+\zeta  ]\label{upp_bound},
\\
&\mathbb{F}(r|\rho,\eta, \kappa  )=\inf\limits_{\{\Prob({\BX^{(n+1)})}\in \mathrm{P}(\mathcal{S}_{=}^{n+1})\}_{n=1}^{\infty}}
\lim\limits_{\zeta  \downarrow 0} \nonumber
\\
&
 \limsup\limits_{N \xrightarrow[]{\rho, \eta, \kappa} \infty} \Prob[\sqrt{Mn}(I^{(n+1)}_{N,L,M}-\overline{C}(\sigma^2))\le r-\zeta  ],
\end{align}
\end{subequations}
with $I^{(n)}_{N,L,M}$ representing the MID given in~(\ref{mid_exp}). Here~(\ref{upp_bound}) is induced by the input $\BX^{(n)}\in \mathbb{C}^{M\times n}=\widetilde{\bold{X}}^{(n)}\left(\frac{1}{Mn}\Tr (\widetilde{\bold{X}}^{(n)}\widetilde{\bold{X}}^{(n),H})\right)^{-\frac{1}{2}}$, where $\widetilde{\bold{X}}^{(n)}\in \mathbb{C}^{M\times n}$ is an i.i.d. Gaussian matrix.
\end{lemma}
\begin{remark} 
\label{ana_sphc}
Note that the $\inf$ operation is taken over $\{\Prob({\BX^{(n+1)})}\in \mathrm{P}(\mathcal{S}_{=}^{n+1})\}$ in $\mathbb{F}(r|\rho,\eta, \kappa  )$, where $\mathrm{P}(\mathcal{S}_{=}^{n+1})$ denotes the set of probability measures on codes with support a subset of $\mathcal{S}_{=}^{n+1}$. This results from the adaptation from the maximal energy constraint to the equal energy constraint by introducing an auxiliary symbol~\cite[Lemma 39]{polyanskiy2010channel}. Lemma~\ref{bnd_err} is important due to two reasons. First, it converts the evaluation with constraint $\mathcal{S}$ to that on the sphere coding (equal energy constraint) $\mathcal{S}_{=}$ in~(\ref{sph_cons}). Second, it indicates that, to characterize $\Prob_{\mathrm{e}}(r|\rho,\eta, \kappa) $, we need to investigate the distribution of the MID. The parameter $\zeta$ will be used when analyzing the cases with $r>0$ and $r\le 0$ in the proof of Theorem~\ref{the_oaep}.
\end{remark}  

\noindent\textit{B. FBL analysis}

With Lemma~\ref{bnd_err}, the upper and lower bounds of the optimal average error probability can be obtained by investigating the distribution of the MIDs $I_{N,L,M}^{(n)}$ and $I_{N,L,M}^{(n+1)}$, respectively, with the equal energy constraint $\mathcal{S}_{=}$. Theorem~\ref{clt_the} gives the asymptotic distribution of the MID given the sphere channel input $\BX^{(n)}$. According to Lemma~\ref{bnd_err} and Theorem~\ref{clt_the}, we can give the approximations for the upper and lower bounds of the optimal average error probability by the following theorem.

\begin{theorem} (Bounds for the optimal average error probability)
\label{the_oaep}
 The optimal average error probability $\Prob_{\mathrm{e}}(r|\rho,\eta, \kappa)$ for the second-order coding rate over Rayleigh-product channels is bounded by
\begin{equation}\label{lower_exp}
\Prob_{\mathrm{e}}(r|\rho,\eta, \kappa) \ge 
\begin{cases}
\Phi(\frac{r}{\sqrt{V_{-}}})+\BO(n^{-\frac{1}{2}}),~~ r\le 0,
\\
\frac{1}{2},~~ r> 0,
\end{cases}
\end{equation}
\begin{equation}\label{upper_exp}
\Prob_{\mathrm{e}}(r|\rho,\eta, \kappa)  \le \Phi(\frac{r}{\sqrt{V_{+}}})+\BO(n^{-\frac{1}{2}}),
\end{equation}
where
\begin{equation}
\begin{aligned}
\label{var_upp_low}
V_{-}&=-\rho\log(\Xi)+\eta+\frac{\sigma^4\delta'}{\kappa},~~
\\
V_{+}&=V_{-}+\kappa\overline{\omega}^4
\left[\frac{\omega^2(1+\delta\overline{\omega})}{1+\delta\overline{\omega}^2}-\frac{\omega\omega'}{\delta(1+\delta \overline{\omega}^2)}\right].
\end{aligned}
\end{equation}
\end{theorem}
\begin{proof} The proof of Theorem~\ref{the_oaep} is omitted due to the page limitation and can be found in the extended version of this paper~\cite[Appendix D]{zhang2022second}.
\end{proof}
Theorem~\ref{the_oaep}, which gives the closed-form approximation for the genuine bounds in Lemma~\ref{bnd_err}, is obtained by utilizing Theorem~\ref{clt_the}. Theorem~\ref{the_oaep} depicts the optimal average error probability for a region of the coding rate close to the ergodic capacity, i.e., $\frac{1}{Mn}\log(|\mathcal{C}_{n}|)=\overline{C}(\sigma^2)+\frac{r}{\sqrt{nM}}$. The corresponding results over~additive white Gaussian noise channels depend only on the noise level $\sigma^2$~\cite{hayashi2009information,polyanskiy2010channel,hoydis2015second} while the results for the Rayleigh channel depend on $\sigma^2$, $\eta$, and $\rho$. Theorem~\ref{the_oaep} shows the impact of the number of scatterers on the optimal average error probability by introducing $\kappa$.

\begin{remark}\label{tightness_bound} (Tightness of bounds) The evaluation in (\ref{lower_exp}) and~(\ref{upper_exp}) give the approximation for the bounds in~Lemma~\ref{bnd_err}. The difference between the upper and lower bounds can be evaluated by
\begin{equation}
\begin{aligned}
0 &< \Phi\Bigl(\frac{r}{\sqrt{V_{+}}}\Bigr)-\Phi\Bigl(\frac{r}{\sqrt{V_{-}}}\Bigr)\le \frac{-r (V_{+}-V_{-} ) }{\sqrt{2\pi V_{+}V_{-}}(\sqrt{V_{+}}+\sqrt{V_{-}})}
\\
&=\frac{\rho^{-\frac{3}{2}}  rD}{\sqrt{(\rho^{-1}V_{+})(\rho^{-1}V_{-})}(\sqrt{\rho^{-1}V_{+}}+\sqrt{\rho^{-1} V_{-}})},
\end{aligned}
\end{equation}
where $D=\kappa\overline{\omega}^4
[\frac{\omega^2(1+\delta\overline{\omega})}{1+\delta\overline{\omega}^2}-\frac{\omega\omega'}{\delta(1+\delta \overline{\omega}^2)}]$. For $r<0$, the lower and upper bounds are not equal and the gap is related to the fourth order term $M^{-1}\E\{\Tr((n^{-1}\BX^{(n)}\BX^{(n),H})^2 )\}$ of $\mathbb{P}(\BX^{(n)})$ in $M^{-1}\Tr(\BA^2)$. This gap orginates from the techniques in~\cite[Eq. (82)]{hoydis2015second}, and vanishes as $\rho  \xrightarrow[]{\eta, \kappa}\infty $. Specifically, when  $\rho  \xrightarrow[]{\eta, \kappa}\infty $, there holds true that $D=\BO(1)$, $\rho^{-1}V_{-}=\BO(1)$ and $\rho^{-1}V_{+}=\BO(1)$ such that $|\Phi(\frac{r}{\sqrt{V_{+}}})-\Phi(\frac{r}{\sqrt{V_{-}}})|=\BO(\rho^{-\frac{3}{2}}) \xrightarrow[]{\rho \xrightarrow[]{\eta, \kappa}\infty} 0$. Therefore, by noticing that $r=\sqrt{Mn}(R- \overline{C}(\sigma^2))$, we have
\begin{equation}
\label{gap_outage}
\begin{aligned}
\Phi\Bigl(\frac{r}{\sqrt{V_{-}}}\Bigr) &=\Phi\Bigl(\frac{M(R-\overline{ C}(\sigma^2))}{\sqrt{\rho^{-1}V_{-}}}\Bigr) \xrightarrow[]{\rho \xrightarrow[]{\eta, \kappa}\infty}
\\
&
=\Phi\Biggl(\frac{M(R-\overline{ C}(\sigma^2))}{\sqrt{-\log(\Xi)}}\Biggr) \approx \mathrm{P}_{\mathrm{out}}(\eta,\kappa,R),
\end{aligned}
\end{equation}
where the last step follows from the outage probability evaluation in~\cite{zhang2022asymptotic,zhang2022outage}. We can obtain the same result for the upper bound by replacing $V_{-}$ with $V_{+}$ in~(\ref{gap_outage}). This indicates that both lower and upper bounds converge to the outage probability. However, for finite $\rho$, the gap between the bound and outage probability is not negligible, indicating that the outage probability~\cite{zhang2022asymptotic,zhang2022outage} is optimistic in the asymptotic regime given in~\textbf{Assumption A}. We can conclude from~(\ref{lower_exp}) and~(\ref{upper_exp}) that the approximation error is $\BO(n^{-\frac{1}{2}})$, which is better than that derived with general moment condition $M^{-1}\Tr(\BA_{n}^2)=\BO(M)$ in Theorem~\ref{the_oaep}. This is because the proof of lower and upper bounds is based on $M^{-1}\Tr(\BA_{n}^2)=\BO(1)$.
\end{remark}
\begin{remark} \label{degenerate_rem} Now we compare Theorem~\ref{the_oaep} with existing works.

\textbf{1. Degeneration to the outage probability for Rayleigh-product channels~\cite{zhang2022asymptotic,zhang2022outage,zheng2016asymptotic}.}
The outage probability with IBL can be obtained by letting $\rho \rightarrow \infty$ ($n$ has a higher order than $M$, $N$, and $L$) while keeping $\eta$ and $\kappa$, i.e., $\Prob_{\mathrm{e}}(r|\rho,\eta, \kappa)\xlongrightarrow[]{\rho \xrightarrow[]{\eta,\kappa} \infty} \Phi\Bigl(\frac{M(R-\overline{C}(\sigma^2))}{\sqrt{-\log(\Xi)}}\Bigr)$, which is equivalent to the result in~\cite[Eq. (19)]{zhang2022outage} and~\cite[Eq. (16)]{zhang2022asymptotic} when the correlation matrices $\FR=\BI_{N}$, $\FS=\BI_{L}$, and $\FT=\BI_{M}$. In~\cite{zhang2022asymptotic,zhang2022outage,zheng2016asymptotic}, the blocklength is assumed to be infinitely large and the limit is taken with respect to $N$, $L$, and $M$. In this work, we change the asymptotic regime such that the blocklength increases at the same pace as $N$, $L$, and $M$. In this case, the impact of the FBL is reflected by the terms after $-\log(\Xi)$ in~(\ref{var_upp_low}), i.e., $\frac{1}{\rho}(\eta+\frac{\sigma^4\delta'}{\kappa})$ in $\frac{V_{-}}{\rho}$ and $\frac{1}{\rho}(\eta+\frac{\sigma^4\delta'}{\kappa}+[\frac{\omega^2(1+\delta\overline{\omega})}{1+\delta\overline{\omega}^2}-\frac{\omega\omega'}{\delta(1+\delta \overline{\omega}^2)}])$ in $\frac{V_{+}}{\rho}$.  When $ \rho \xrightarrow[]{\eta,\kappa} \infty $ and $\eta=1$, the mean $\overline{C}(\sigma^2)$ and variance term $-\log(\Xi)$ degenerate to~\cite[Proposition 1 and Eq. (23)]{zheng2016asymptotic}, respectively.

\textbf{2. Degeneration to the bounds for Rayleigh channels~\cite{hoydis2015second}.}
By letting $\kappa \xrightarrow[]{\eta,\rho} 0$, the quantities $V_{-}$ and $V_{+}$ will degenerate to $\theta^2_{-}$ and $\theta^2_{+}$ in~\cite[Eqs. (24)-(25)]{hoydis2015second} for the Rayleigh channel.

\textbf{3. Degeneration to the outage probability for Rayleigh channels~\cite{kamath2005asymptotic,hachem2008new}.} By letting $\kappa  \xrightarrow[]{\eta} 0$ and $\rho \xrightarrow[]{\eta} \infty$,  we have $\Prob_{\mathrm{e}}(r|\rho,\eta, \kappa)\xlongrightarrow[]{\eta} \Phi\Bigl(\frac{M(R-\overline{C}_{\mathrm{Rayleigh}}(\sigma^2))}{\sqrt{-\log(\Xi_{\mathrm{Rayleigh}})}}\Bigr)$, where $\overline{C}_{\mathrm{Rayleigh}}$ and $\Xi_{\mathrm{Rayleigh}}$ are given in~\cite[Theorems 1 and 2 of uncorrelated case]{hachem2008new} and~\cite[Eq. (11)]{kamath2005asymptotic}. The above degeneration relationships are illustrated in Fig~\ref{fig_degenerate}.
\end{remark}

\begin{figure}
\centering
\includegraphics[width=0.45\textwidth]{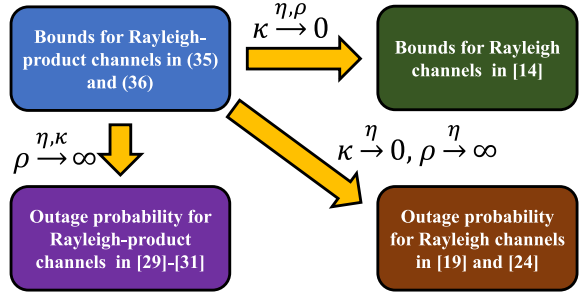}
\captionof{figure}{Comparison with existing works.}
\label{fig_degenerate}
\vspace{-0.3cm}
\end{figure}

\begin{proposition} 
\label{equ_pro}
(Bounds with equal number of transceiver antennas) When the transceivers have the same number of antennas, i.e., $M=N$ ($\eta=1$), the optimal average error probability is bounded with
\begin{equation}
\begin{aligned}
\label{equal_variance}
V_{-}=\rho W(\omega) +1+X(\omega),~~
V_{+}=V_{-}+Y(\omega),
\end{aligned}
\end{equation}
where $\omega$ is the positive solution of the equation
\begin{equation}
  \label{cubic_g}
  \omega^{3}+2\omega^{2}+(1+\frac{\kappa}{\sigma^2}  - \frac{1}{\sigma^2} )\omega-\frac{1}{\sigma^2}=0,
  \end{equation}
  satisfying $\omega>0$ and $1+(1-\kappa)\omega>0$. $W(\omega)$, $X(\omega)$ and $Y(\omega)$ are given by
\begin{equation}
\begin{aligned}
W(\omega)&= \log\left( \frac{(1+\omega)^2}{1+2\sigma^2\omega^3+2\sigma^2\omega^2}\right),~
\\
X(\omega)
&=-\frac{(1+z\omega^2(1+\omega))}{(1+\omega)(1+2z\omega^3+2z\omega^2)},
\\
Y(\omega)&=\frac{\kappa}{(1+\omega)^4}\Biggl[\frac{ \omega^2(\sigma^2(1+\omega)^2+\omega+1+\kappa ) }{\sigma^2(1+\omega)^2+(1-\kappa)\omega+1+\kappa}
\\
&
+\frac{[1+(1-\kappa)\omega]^4X(\omega)}{\kappa^2[(1-\kappa)\omega^2+2\omega+1]^2}\Biggr].
\end{aligned}
 \end{equation}  
\end{proposition}
\begin{proof}
The proof of Proposition~\ref{equ_pro} is given in the extended version of this paper~\cite[Appendix E]{zhang2022second}.
\end{proof}
\begin{remark} By letting $\rho\rightarrow \infty$ while keeping $\eta$ and $\kappa$, the upper and lower bounds in~(\ref{equal_variance}) converge to the limiting outage probability, which can be given as $\Phi(\frac{r}{\sqrt{W(\omega)}})$. The limiting outage probability is identical to the outage probability with IBL in~\cite[Proposition 2]{zhang2022outage} and~\cite[Proposition 3]{zheng2016asymptotic} with the rate threshold $R=M\overline{C}(\sigma^2)+\sqrt{\rho}r$.
\end{remark}
\begin{proposition} 
\label{high_snr_pro}
(High SNR approximation) The high SNR approximations for $V_{+}$ and $V_{-}$ are given as
\begin{equation}
\label{v__approx-}
V_{-}=
\begin{cases}
\begin{aligned}
&-\rho\log((1-\kappa)(1-\eta^{-1}))+1+\BO(\sigma^2),
\\
&~ \text{when}~\kappa<1\wedge \eta>1,
\\
&-\rho\log((1-\eta)(1-{\eta\kappa}))+\eta+\BO(\sigma^2),
\\
&~ \text{when}~
\eta<1 \wedge \eta\kappa<1,
\\
&-\rho\log((1-\kappa^{-1})(1-(\eta\kappa)^{-1}))+\frac{1}{\kappa}+\BO(\sigma^2),
\\
&~ \text{when}~(\kappa>1\wedge \eta>1)\vee (\eta\kappa>1\wedge \eta<1).
\end{aligned}
\end{cases}
\end{equation}
\begin{equation}
\label{v__approx+}
\begin{aligned}
V_{+}=
\begin{cases}
\begin{aligned}
&-\rho\log((1-\kappa)(1-\eta^{-1}))\!+1\!+\BO(\sigma^2),
\\
&~ \text{when}~ \kappa<1\wedge \eta>1,
\\
&-\rho\log((1-\eta)(1-{\eta\kappa}))\!+\!\eta(2-\eta)\!+\!\BO(\sigma^2),
\\
&~ \text{when}~\eta<1 \wedge \eta\kappa<1,
\\
&-\rho\log((1-\frac{1}{\kappa})(1-(\eta\kappa)^{-1}))\!+\!\frac{1}{\kappa}\!+\!\frac{(\kappa-1)}{\kappa^2}\!+\!\BO(\sigma^2),
\\
&~ \text{when}~
(\kappa>1\wedge \eta>1)\vee (\eta\kappa>1\wedge \eta<1),
\end{aligned}
\end{cases}
\end{aligned}
\end{equation}
where $\wedge$ and $\vee$ represent ``AND'' and ``OR'' operator, respectively.
\end{proposition}
\begin{proof} The proof of Proposition~\ref{high_snr_pro} is given in the extended version of this paper~\cite[Appendix F]{zhang2022second}.
\end{proof}
\begin{remark}
In the first two cases, when $\kappa \xrightarrow[]{\eta,\rho} 0$, the approximations of the variances will converge to those for the Rayleigh channel in~\cite[Remark 4]{hoydis2015second}. We can conclude that the variances of the MID in the Rayleigh-product channel are larger than those of the Rayleigh channel. In the first case, the gap between the upper bound and lower bound vanishes. The third case can not degenerate to the Rayleigh case, and illustrates the impact of rank-deficiency in the high SNR regime. Here the edge cases ($M=L$ or $M=N$ or $L=N$) are not discussed as the variance will increase with $\frac{1}{\sigma^2}$ to infinity. The analysis for $N=M$ can be derived from~\cite[Propositions 1]{zhang2022asymptotic}.
\end{remark}

\begin{proposition}
\label{low_snr_app}
(Low SNR approximation) In the low SNR region when $\sigma^2 \rightarrow \infty$, the following parameters will converge to zero
\begin{equation}
\label{v_low}
\begin{aligned}
V_{-}=2(1+\eta\kappa)\eta \sigma^{-2} +\BO(\sigma^{-4}),~
\\
V_{+}=2(1+\eta\kappa)\eta \sigma^{-2} +\BO(\sigma^{-4}).
\end{aligned}
\end{equation}
\end{proposition}
\begin{proof} The proof of Proposition~\ref{low_snr_app} is given in the extended version of this paper~\cite[Appendix G]{zhang2022second}.
\end{proof}
By Proposition~\ref{low_snr_app}, we can obtain $\frac{V_{+}-V_{-}}{V_{+}}=\BO(\sigma^{-2})$, which indicates the gap between the upper and lower bounds is small in the low SNR region. Compared with the Rayleigh channel, the term $2\eta^2\kappa\sigma^{-2}$ is introduced by the two-hop structure and shows the impact of the number of scatterers $L$. The term will vanish when $\kappa \xrightarrow[]{\eta,\rho} 0$ such that
$V_{+}$ and $V_{-}$ become $2\eta \sigma^{-2}$, which is the same as that of the Rayleigh channel.
\begin{remark} 
\label{comparison_rayleigh_remark}
By comparing the high SNR approximations for $V_{-}$ and $V_{+}$ in~(\ref{v__approx-}) and~(\ref{v__approx+}) with those for Rayleigh channels in~\cite[Eqs. (27) and (28)]{hoydis2015second}, we have $V_{-}>\theta_{-}^{2}$ and $V_{+}>\theta_{+}^{2}$. This also holds true for the low SNR case if we compare the low SNR approximations in~(\ref{v_low}) with those for Rayleigh channels in~\cite[Eqs. (30) and (31)]{hoydis2015second}. Meanwhile, by comparing the high SNR and low SNR approximations for $\overline{C}(\sigma^2)$ and those for Rayleigh channels, we can obtain that $\overline{C}(\sigma^2) \le \overline{C}_{Rayleigh}(\sigma^2)$. Under such circumstances, for a given rate $R<\overline{C}(\sigma^2)$, we have $0>r=\sqrt{Mn}(R-\overline{C}(\sigma^2))>r'=\sqrt{Mn}(R-\overline{C}_{Rayleigh}(\sigma^2)$, $V_{-}>\theta_{-}^{2}$, and $V_{+}>\theta_{+}^{2}$ such that $\Phi(\frac{r}{\sqrt{V_{+}}})>\Phi(\frac{r'}{\sqrt{\theta^2_{+}}})$ and $\Phi(\frac{r}{\sqrt{V_{-}}})>\Phi(\frac{r'}{\sqrt{\theta^2_{-}}})$, which implies that the optimal average error probability of the Rayleigh-product channel is worse than that of the Rayleigh channel.
\end{remark}

\section{Numerical Results}
\label{sec_simu}
In this section, the analytical results derived in this paper will be verified by numerical simulations.
\vspace{-0.6cm}

\subsection{Approximation Accuracy of Upper and Lower Bounds}
Figs.~\ref{rev_lower} and~\ref{rev_upper} compare the analytical expressions in (\ref{lower_exp}) for the lower bound and~(\ref{upper_exp}) for the upper bound with numerical results. Expressions in (\ref{lower_exp}) and (\ref{upper_exp}) are large-system approximations for the distribution of the MID bounds in Lemma~\ref{bnd_err}, which can be evaluated by Monte-Carlo simulation. The simulation settings are: $\sigma^{-2}= 5$ dB, $R=1.55$ nat/s/Hz, $M=16$, $N=32$, and $L=\{48,64,80 \}$. The number of Monte-Carlo realizations is $5\times 10^{7}$. It can be observed that the analytical expressions in Theorem~\ref{the_oaep} are accurate. 

\vspace{-0.3cm}
\begin{figure}[t!]
\centering
\includegraphics[width=0.45\textwidth]{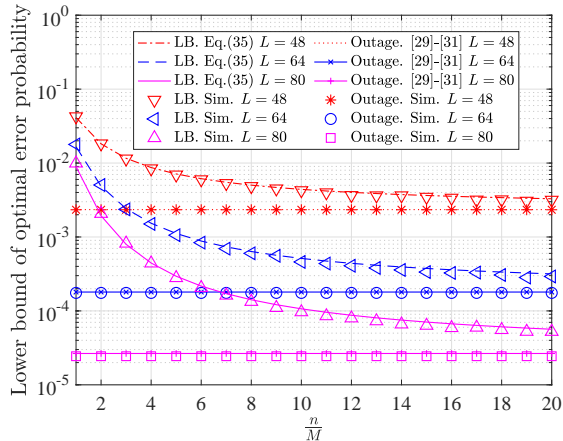}
\captionof{figure}{Lower bound.}\label{rev_lower}
\end{figure}
\vspace{-0.3cm}
\begin{figure}
\centering
\includegraphics[width=0.45\textwidth]{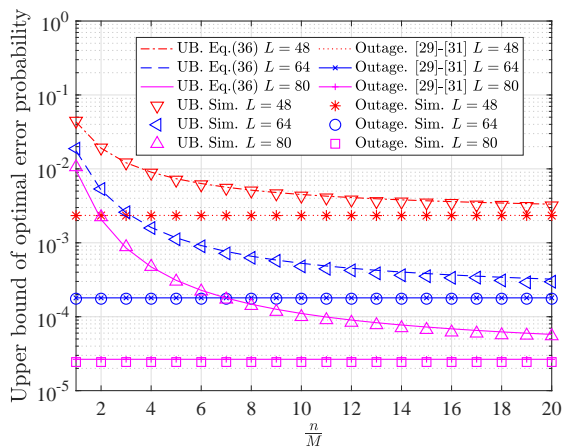}
\captionof{figure}{Upper bound.}\label{rev_upper}
\vspace{-0.5cm}
\end{figure}


\subsection{Impact of $n$ on the Optimal Average Error Probability}
It can be observed from Figs.~\ref{rev_lower} and~\ref{rev_upper} that when $\rho $ becomes larger, both the lower and upper bounds in Eqs.~(\ref{lower_exp}) and~(\ref{upper_exp}) approach the outage probability with IBL~\cite{zhang2022asymptotic,zhang2022outage,zheng2016asymptotic}, which agrees with the analysis in Remark~\ref{tightness_bound} and Remark~\ref{degenerate_rem}.1. We can also observe that for small $\frac{n}{M}$, the outage probability is overly optimistic. In this case, the FBL effect should be considered and the proposed bounds offer better performance evaluation.

\begin{figure}
\centering
\includegraphics[width=0.45 \textwidth]{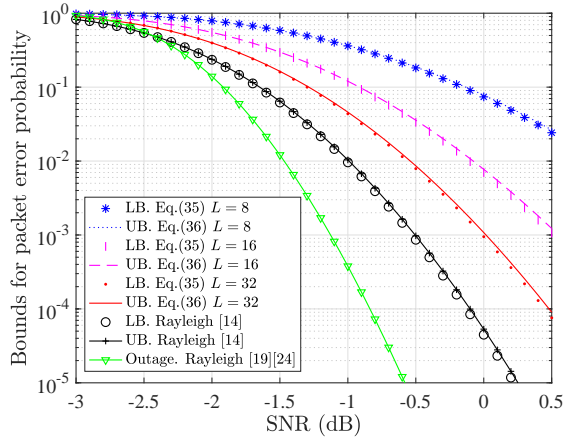}
\captionof{figure}{Bounds with different numbers of the scatterers.}
\label{simu_rank}
\end{figure}

\subsection{Impact of the Number of Scatterers on the Packet Error Probability}

In Fig.~\ref{simu_rank}, we compare the error bounds with different $L$. The parameters are set as $N=16$, $M=8$, $L=\{8, 16, 32 \}$\footnote{Although the theoretical results are derived by RMT with infinite $M$, $N$, $L$, and $n$, it has been demonstrated that the RMT results are accurate for small-scale systems~\cite{hoydis2011asymptotic,zhang2022large,zhang2022asymptotic}.}, $n=36$, and $R=\log(2) $ nat/s/Hz. We compare the theoretical bounds of the Rayleigh-product channel in Eqs. (\ref{lower_exp})-(\ref{upper_exp}) with those of the Rayleigh channel in~\cite{hoydis2015second} and the outage probability in~\cite{kamath2005asymptotic,hachem2008new}. There are gaps between the theoretical bounds for Rayleigh-product channels and those for Rayleigh channels since the latter fail to characterize the rank deficiency of Rayleigh-product channels. It can be observed that the upper and lower bounds are close and the optimal average error probability of the Rayleigh-product channel is worse than that of the Rayleigh channel, which agrees with Remark~\ref{comparison_rayleigh_remark}. Furthermore, as $L$ increases, the bounds for the Rayleigh-product channel approach those of the Rayleigh channel, which coincides with the analysis in Remark~\ref{degenerate_rem}.2.

\vspace{-0.2cm}

\subsection{Comparison Between Theoretical Bounds and Performance of LDPC}

Next, we compare the theoretical bounds with the performance of specific coding schemes, where the WiMAX standard with a LDPC code~\cite{hoydis2015second} is adopted. Specifically, the system parameters are set as $M=8$, $N=16$, $L=24$, and the inputs are generated uniformly. The coding scheme we adopt is the $1 \slash 2$ LDPC codes, which were used in the WiMAX standard~\cite{8303870}. A bit interleaved coded modulation scheme with a random interleaver is used before modulation. The QPSK modulation is employed at the transmitter and the coding rate is set as $R=\log(2)$. At the receiver, the received signal is demodulated by the maximal likelihood (ML) demodulator~\cite{muller2002coding,mckay2005capacity},
\begin{equation}
L(s_{i}|\bold{r},\BH)=\log\frac{\sum_{\bold{c}\in \mathcal{C}^{(i)}_{1} } p(\bold{r}|\bold{c},\BH) }{\sum_{\bold{c}\in \mathcal{C}^{(i)}_{0} } p(\bold{r}|\bold{c},\BH) },
\end{equation}
where $L(s_{i}|\bold{r},\BH)$ represents the log likelihood ratio of the $i$-th bit $s_{i}\in \{0,1\}$ and $\mathcal{C}^{(i)}_{1}=\{\bold{c}| c_{i}=1,\bold{c}\in \mathcal{C}   \}$ denotes the set of the codewords whose $i$-th digit is $1$. Here $\mathcal{C}^{(i)}_{0}$ represents the set of the codewords whose $i$-th digit is $0$. $p(\bold{r}|\bold{c},\BH)$ is the conditional probability density function of the received signal $\bold{r}$. The output of the demodulator is then decoded by the soft-decision LDPC decoder. The length of the LDPC codes are $l \in \{576,2304\}$ bits, which correspond to the blocklengths $n=\frac{l}{2M}\in\{36, 144 \} $. The packet error is compared with the theoretical bounds in Theorem~\ref{the_oaep} for different SNRs ($\frac{1}{\sigma^2}$). The second-order coding rate is $r=\sqrt{Mn}(R-\overline{C}(\sigma^2))$. From Fig.~\ref{simu_ldpc}, it can be observed that the upper bound and lower bound are nearly overlapped, which validates the tightness of the bounds. Similar to the single Rayleigh case~\cite{hoydis2015second}, the theoretical bounds are nearly parallel with the error probability of the LDPC codes and the gap is around $2$ dB in the considered SNR range. Similar phenomenon can also be observed from the low rank case ($L=4$) but with worse performance, as shown in Fig.~\ref{low_ldpc}.
\begin{figure}
\centering
\includegraphics[width=0.45\textwidth]{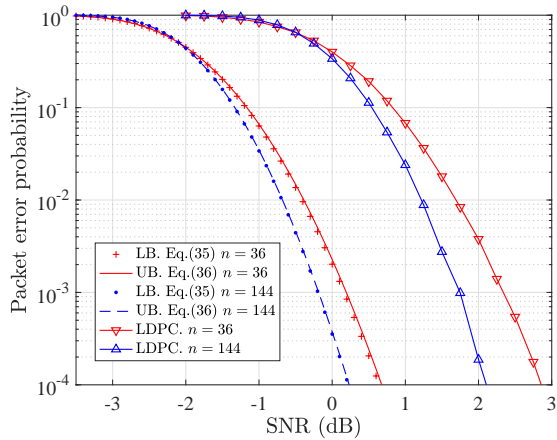}
\captionof{figure}{Comparison of the bounds with LDPC coding ($L=24$).}
\label{simu_ldpc}
\vspace{-0.3cm}
\end{figure}

\begin{figure}
\centering
\includegraphics[width=0.45\textwidth]{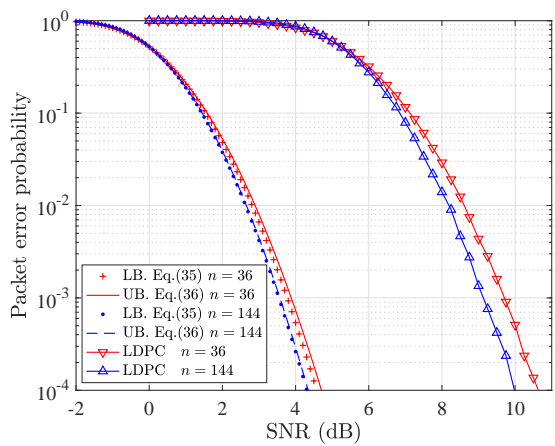}
\captionof{figure}{Comparison of the bounds with LDPC coding ($L=4$).}
\label{low_ldpc}
\vspace{-0.4cm}
\end{figure}

\section{Conclusion and Future Works}
\label{sec_con}
This paper characterized the MID of massive MIMO systems over Rayleigh-product channels by setting up a CLT utilizing RMT. Based on the CLT, we performed the FBL analysis and analyzed the impact of rank-deficiency. Specifically, we first derived the upper and lower bounds for the packet error probability, and then obtained their high and low SNR approximations, which explicitly show the impact of rank-deficiency. The results in this work degenerate to the FBL results for Rayleigh channels when the number of scatterers approaches infinity, the IBL results for Rayleigh-product channels when the blocklength approaches infinity, and the IBL results for Rayleigh channels when both the number of scatterers and the blocklength go to infinity. The derived results can be extended to double-scattering channels and IRS-aided MIMO channels with Rayleigh fading~\cite{moustakas2022reconfigurable} by the approach in Appendix~\ref{prof_the2} with the aid of the computation results in~\cite{zhang2022asymptotic}. Considering challenges in CSI acquisition for large-scale MIMO systems, it is also meaningful to extend the results to the case with imperfect CSI~\cite{potter2013achievable}.

\appendices

\section{Proof of~(\ref{delta_deri})}
\label{proof_delta_deri}
\begin{proof} By taking the derivative with respect to $\sigma^2$ on both sides of~(\ref{cubic_eq}), we have
\begin{equation}
\label{omegap_eq}
3\omega^2\omega'+4\omega\omega'+(\eta-1)(\kappa-1)(\frac{2\omega\omega'-\frac{\omega^2}{\sigma^2}}{\sigma^2})+\omega'
+(\eta\kappa-2\eta+1)(\frac{\omega'-\frac{\omega}{\sigma^2}}{\sigma^2})+\frac{\eta}{\sigma^4}=0.
\end{equation}
By solving $\omega'$ in~(\ref{omegap_eq}), we have
\begin{equation}
\omega'=\frac{(\eta-1)(\kappa-1)\frac{\omega^2}{\sigma^2}+(\eta\kappa-2\eta+1)\frac{\omega}{\sigma^2}-\frac{\eta}{\sigma^2}}{\sigma^2+3\sigma^2\omega^2+4\sigma^2\omega+2(\eta-1)(\kappa-1)\omega+(\eta\kappa-2\eta+1)}:=\frac{A}{B}.
\end{equation}
By~(\ref{cubic_eq}), we can rewrite $A$ and $B$ as
\begin{equation}
A=-\omega(1+\omega)^2, 
\end{equation}
and
\begin{equation}
\label{B_compute}
\begin{aligned}
B&=\sigma^2(1+\omega)^2+2\sigma^2\omega(\omega+1)+2(\eta-1)(\kappa-1)\omega+\eta\kappa-2\eta+1
\\
&
=\sigma^2(1+\omega)^2+2(\frac{\eta+(\eta-\kappa)\omega}{1+\omega})
+\eta\kappa-2\eta+1
\\
&
=\sigma^2(1+\omega)^2+(\eta\kappa-\frac{\kappa \omega}{1+\omega})-\frac{\kappa\omega}{1+\omega}+1
\\
&\overset{(a)}{=}\sigma^2(1+\omega)^2+\sigma^2\delta+\kappa\omega( \frac{1}{\delta}+\overline{\omega}-\overline{\omega}  )
=\sigma^2(1+\omega)^2+\sigma^2\delta+\frac{\kappa\omega}{\delta}
=\frac{\Delta_{\sigma^2}(1+\delta\overline{\omega}^2)}{\overline{\omega}^2},
\end{aligned}
\end{equation}
where step $(a)$ in~(\ref{B_compute}) follows by applying $\eta\kappa=\sigma^2\delta+\frac{\kappa\omega}{1+\omega}$ to the second bracket and utilizing $\kappa\omega(\frac{1}{\delta}+\overline{\omega})$ to replace $1$ in the last step. Therefore, we can obtain
\begin{equation}
\label{temp_omegad}
\omega'=\frac{A\overline{\omega}^2}{B\overline{\omega}^2}=-\frac{\omega}{(1+\delta\overline{\omega}^2)\Delta_{\sigma^2}}.
\end{equation}
By~(\ref{delta_def}) and~(\ref{temp_omegad}), we have
\begin{equation}
\delta'=-\frac{\delta}{\sigma^2}-\frac{\kappa\overline{\omega}^2\omega'}{\sigma^2}=-\frac{\delta}{\Delta_{\sigma^2}},
\end{equation}
which concludes~(\ref{delta_deri1}) and~(\ref{delta_deri2}) follows immediately.~(\ref{delta_deri3}) can be obtained by
\begin{equation}
\overline{\omega}'=-\overline{\omega}^2\omega'=-\frac{\omega\overline{\omega}^2\delta'}{\delta(1+\delta\overline{\omega}^2)}.
\end{equation}
By~\cite[Eq. (90)]{hoydis2011iterative}, we have $\frac{\partial \overline{C}(\sigma^2)+\eta\log(\sigma^2)}{\partial \sigma^2}=\frac{\delta}{\kappa}$ to conclude~(\ref{delta_deri4}).
\end{proof}

\section{Mathematical Tools and Useful Results}
\label{app_mathtool}
In this paper, we mainly use the Gaussian tools for the proof, which consists of the~\textit{Nash-Poincar{\'e} Inequality} and the \textit{Integration by Parts Formula}. 

\textit{1. Nash-Poincar{\'e} Inequality}~\cite[Eq. (18)]{hachem2008new},~\cite[Proposition 2.5]{pastur2005simple}.
Denote $\bold{x}=[x_1,...,x_N]^{T}$ as a complex Gaussian random vector satisfying $\E [\bold{x}]=\bold{0}$, $\E[\bold{x}\bold{x}^{T}]=\bold{0}$, and $\E[\bold{x}\bold{x}^{H}]=\bold{\Omega}$. $f=f(\bold{x},\bold{x}^{*})$ is a $\mathbb{C}^{1}$ complex function such that both itself and its derivatives are polynomially bounded. Then, the variance of $f$ satisfies the following inequality,
\begin{equation}
\begin{aligned}
\label{nash_p}
\Var[f(\bold{x},\bold{x}^{*})] \le \E[\nabla_{\bold{x}}f(\bold{x},\bold{x}^{*})^{T}\bold{\Omega}\nabla_{\bold{x}}f(\bold{x},\bold{x}^{*})^{*}]+ \E[\nabla_{\bold{x}^{*}}f(\bold{x},\bold{x}^{*})^{H}\bold{\Omega}\nabla_{\bold{x}^{*}}f(\bold{x},\bold{x}^{*})],
\end{aligned}
\end{equation}
where $\nabla_{\bold{x}}f(\bold{x})=[\frac{\partial f}{\partial x_1},...,\frac{\partial f}{\partial x_N}]^{T}$ and $\nabla_{\bold{x}^{*}}f(\bold{x})=[\frac{\partial f}{\partial x_1^{*}},...,\frac{\partial f}{\partial x_N^{*}}]^{T}$.
(\ref{nash_p}) is referred to as the~\textit{Nash-Poincar{\'e} inequality}, which gives an upper bound for functional Gaussian random variables and is widely utilized in the error estimation for the expectation of Gaussian matrices. By utilizing this inequality, it was shown that the approximation error of the deterministic approximation for the MI of MIMO channels is $\BO(\frac{1}{N})$ for both the Rayleigh and double-scattering channels~\cite{hachem2008new,zhang2022asymptotic}.

\textit{2. Integration by Parts Formula}~\cite[Eq. (17)]{hachem2008new}. The formula is given by
\begin{equation}
\label{int_part}
\E[x_{i}f(\bold{x},\bold{x}^{*})]=\sum_{m=1}^{N}[\bold{\Omega}]_{i,m} \E \Bigl[\frac{\partial f(\bold{x},\bold{x}^{*})}{\partial x_{m}^{*}}\Bigr].
\end{equation}
If $\bold{\Omega}=\bold{I}_{N}$,~(\ref{int_part}) can be simplified as
\begin{equation}
\E[x_{i}f(\bold{x},\bold{x}^{*})]=\E \Bigl[\frac{\partial f(\bold{x},\bold{x}^{*})}{\partial x_{m}^{*}}\Bigr].
\end{equation}
By this formula, the expectation for the product of a Gaussian random variable and a functional Gaussian random variable is converted to the expectation for the derivative of the functional Gaussian random variable.

\begin{lemma} 
\label{var_con}
Given \textbf{Assumption~{A}} and defining $\BQ(z)=(z\BI_{N}+\BH\BH^{H})^{-1}$, for finite $x,z>0$, $m\ge 0$ and matrix $\BB$, we have
\begin{align}
\Var\Bigl(\frac{\Tr(\BB\BQ(x)\BZ\BZ^{H}\BQ(z)^{m})}{M}\Bigr) 
&
\le \frac{K_m\Tr(\BB\BB^{H})(1+2mz^{-2}+2x^{-2}]}{M^3 z^{2m}x^2 },\label{var_Q}\\
\Var\Bigl(\frac{\Tr(\BQ(x)\BZ\BZ^{H}\BQ(z)^{m}\BH\BB\BH^{H})}{M}\Bigr)
&\le \frac{K_m \Tr(\BB\BB^{H}) (3+2x^{-2}+2mz^{-2} ]}{M^3 z^{2m}x^2},\label{var_QZZ}\\
\Var\Bigl(\frac{\Tr(\BQ(x)\BQ(z)^{m}\BH\BB\BH^{H})}{M}\Bigr)
&
\le \frac{K_m \Tr(\BB\BB^{H}) (2+2x^{-2}+2mz^{-2} )}{M^3 z^{2m}x^2},\label{var_QHH}\\
\Var\Bigl(\frac{\Tr(\BQ(x)\BZ\BZ^{H}(\BQ(z)\BH\BB\BH^{H})^{m})}{M}\Bigr)
&
\le \frac{K_m\Tr((\BB\BB^{H})^m)(2m+1+2x^{-2}+2mz^{-2})}{M^3x^2 z^{2m}},\label{var_QHCH}\\
\Var\Bigl(\frac{\Tr((\BQ(z)\BZ\BZ^{H})^{m}\BQ(x)\BH\BB\BH^{H})}{M}\Bigr)
&
\le \frac{K_m\Tr(\BB\BB^{H})(m+2+2x^{-2}+2mz^{-2})}{M^3x^2 z^{2m}}.\label{var_QZZQHBH}
\end{align}
\end{lemma}
\begin{proof} 
By Nash-Poincar{\'e} inequality~(\ref{nash_p}), we can obtain
\begin{equation}
\begin{aligned}
&\Var\Bigl(\frac{\Tr\BB\BQ(x)\BZ\BZ^{H}\BQ(z)^{m}}{M}\Bigr)\le \frac{1}{M^2L}\E \Bigl[ \sum_{i,j}  | [\BY\BH^{H} \BQ(x)\BZ\BZ^{H}\BQ(z)^{m}\BB\BQ(x)]_{j,i} |^2 
+  | [ \BQ(x)\BZ\BZ^{H}\BQ(z)^{m}\BB\BQ(x)\BH\BY^{H}]_{i,j} |^2 
  \\
&+ \sum_{k=1}^{m} | [\BY\BH^{H} \BQ^{k}(z)\BB\BQ(x)\BZ\BZ^{H}\BQ(z)^{m-k+1}]_{j,i} |^2 + | [\BQ^{k}(z)\BB\BQ(x)\BZ\BZ^{H}\BQ(z)^{m-k+1}\BH\BY^{H} ]_{i,j} |^2 
\\
&
+ |[\BZ^{H}\BQ(z)^{m}\BB\BQ(x)]_{j,i}|^2+|[\BQ(z)^{m}\BB\BQ(x)\BZ]_{i,j}|^2
  \Bigr] 
  \\
  &
 +\frac{1}{M^3} \E\Bigl[ \sum_{i,j}  | [\BH^{H} \BQ(x)\BZ\BZ^{H}\BQ(z)^{m}\BB\BQ(x)\BZ]_{j,i} |^2 
+  | [ \BZ^{H}\BQ(x)\BZ\BZ^{H}\BQ(z)^{m}\BB\BQ(x)\BH]_{i,j} |^2 
  \\
&+ \sum_{k=1}^{m} | [\BH^{H} \BQ^{k}(z)\BB\BQ(x)\BZ\BZ^{H}\BQ(z)^{m-k+1}\BZ]_{j,i} |^2 + | [\BZ^{H} \BQ^{k}(z)\BB\BQ(x)\BZ\BZ^{H}\BQ(z)^{m-k+1}\BH]_{i,j} |^2 
  \Bigr] 
  \\
&   \le \frac{1}{M^3}  [  2K_1 z^{-2m} x^{-4} \Tr(\BB\BB^{H}) +2 K_1 m z^{-2(m+1)} x^{-2} \Tr(\BB\BB^{H}) +2  \Tr(\BB\BB^{H}) x^{-2} z^{-2m} 
\\
&
+ 2K_2 z^{-2m} x^{-4} \Tr(\BB\BB^{H}) +2 K_2 m z^{-2(m+1)} x^{-2} \Tr(\BB\BB^{H})
 ]
 \\
& \le \frac{K_m\Tr((\BB\BB^{H})^m)(2m+1+2x^{-2}+2mz^{-2})}{M^3x^2 z^{2m}},
\end{aligned}
\end{equation}
where $K_1$, $K_2$, and $K_m$ does not depend on $M,N,L$. This concludes~(\ref{var_Q}). Bounds in~(\ref{var_QZZ})-(\ref{var_QZZQHBH}) can be can be evaluated similarly and omitted here.
\end{proof}
We use $(\delta,\omega, \overline{\omega})$ and $(\delta_z,\omega_z,\overline{\omega}_z)$  to denote the parameters generated by the polynomial~$P(\sigma^2)$ in~(\ref{cubic_eq}) and $P(z)$ in~(\ref{cubic_for_z}), respectively. Specifically, $\omega_{z}$ is the positive solution of 
\begin{equation}
\label{cubic_for_z}
\begin{aligned}
&\omega^3+(2z+\eta\kappa-\kappa-\eta+1 )\frac{\omega^2}{z}
\\
&
+(1+\frac{\eta\kappa}{z}-\frac{2\eta}{ z}+\frac{1}{z})\omega-\frac{\eta}{z}=0,
\end{aligned}
\end{equation}
such that $\eta+(\eta-1)\omega_z>0$. $\delta_{z}$ and $\overline{\omega}_z$ can be obtained by~(\ref{delta_def}). In particular, when $z=\sigma^2$, we have $\omega=\omega_z$, $\delta=\delta_z$, and $\overline{\omega}=\overline{\omega}_z$.

\begin{lemma} \label{appro_lem} Given~\textbf{Assumption~{A}}, $\BQ=(\sigma^2\BI_{N}+\BH\BH^{H})^{-1}$, and $\BQ(z)=(z\BI_{N}+\BH\BH^{H})^{-1}$, the following computation results hold true
\begin{equation}
\label{appro_del}
\begin{aligned}
\frac{\E[\Tr(\BQ(z))]}{L}&=\delta_z+\BO(\frac{\mathcal{P}(\frac{1}{z})}{z^2 M^2}),~~
\\
\frac{\E[\Tr(\BQ(z)\BZ\BZ^{H})]}{M}&=\omega_z+\BO(\frac{\mathcal{P}(\frac{1}{z})}{z^2 M^2}),
\\
\frac{\E[\Tr(\BQ(z)\BH\BH^{H})]}{M}&=\omega_z\overline{\omega}_z+\BO(\frac{\mathcal{P}(\frac{1}{z})}{z M^2}),
\end{aligned}
\end{equation}
\begin{equation}
\begin{aligned}
\label{del_SS}
\frac{\omega}{\delta(1+\delta_z \overline{\omega}\overline{\omega}_z)}=\frac{\omega_z}{\delta_z(1+\delta \overline{\omega}\overline{\omega}_z)},
\end{aligned}
\end{equation}
\begin{equation}
\label{QzQ_exp}
\begin{aligned}
\frac{\E[\Tr(\BQ\BQ(z))]}{L}
=\frac{\delta_z }{\Delta_{\sigma^2}(z)}
+\BO(\frac{\mathcal{P}(\frac{1}{z})}{z M^2}),
\end{aligned}
\end{equation}
\begin{equation}
\label{QzQZZ_exp}
\begin{aligned}
\frac{\E[\Tr(\BQ\BZ\BZ^{H}\BQ(z))]}{M}
=\frac{\omega \delta_z }{\delta(1+\delta_z \overline{\omega}\overline{\omega}_z)\Delta_{\sigma^2}(z)}
+\BO\Bigl(\frac{\mathcal{P}(\frac{1}{z})}{z M^2}\Bigr),
\end{aligned}
\end{equation}
\begin{equation}
\label{exp_QQHH}
\begin{aligned}
\frac{\E[\Tr(\BQ(z)\BQ\BH\BH^{H})] }{M}=\frac{\overline{\omega}\overline{\omega}_z\omega_z\delta_z}{\delta_z(1+\delta \overline{\omega}\overline{\omega}_z)\Delta_{\sigma^2}(z)}
+\BO\Bigl(\frac{\mathcal{P}(\frac{1}{z})}{z M^2}\Bigr),
\end{aligned}
\end{equation}
\begin{equation}
\label{exp_QzQHH}
\begin{aligned}
&
\frac{\E[\Tr(\BQ \BZ  \BZ^H \BQ(z)\BH \BH^{H})]}{M}
\\
&
=\frac{\delta\overline{\omega}\omega_z\overline{\omega}_z}{1+\delta\overline{\omega}\overline{\omega}_z}
+\frac{M\overline{\omega}\omega_z^2\overline{\omega}_z }{L\delta_z(1+\delta\overline{\omega}\overline{\omega}_z )^2\Delta_{\sigma^2}(z)}
+\BO\Bigl(\frac{\mathcal{P}(\frac{1}{z})}{z M^2}\Bigr),
\end{aligned}
\end{equation}
\begin{equation}
\label{QZZQZZ_exp}
\begin{aligned}
&\frac{\E[\Tr(\BQ \BZ  \BZ^H \BQ(z)\BZ \BZ^{H})]}{M}
=\frac{M\omega_z\omega(1+\delta\overline{\omega})}{L(1+\delta\overline{\omega}\overline{\omega}_z)}
\\
&
+
\frac{M\omega_z }{L\delta_z(1+\delta\overline{\omega}\overline{\omega}_z )}
\frac{\E[\Tr(\BQ \BZ  \BZ^H \BQ(z))]}{M} ++\BO\Bigl(\frac{\mathcal{P}(\frac{1}{z})}{z M^2}\Bigr),
\end{aligned}
\end{equation}where 
\begin{equation}
\label{Delta_exp}
\Delta_{\sigma^2}(z)=\sigma^2+ \frac{M\omega_z\overline{\omega}\overline{\omega}_z}{L\delta_z(1+\delta \overline{\omega}\overline{\omega}_z)}.
\end{equation}
\end{lemma}
\begin{proof} Results in (\ref{appro_del}) can be obtained by taking $\FR=\BI_N$, $\FS=\BI_L$, and $\FT=\BI_M$ in~\cite[Lemma 4]{zhang2022asymptotic}. Now we focus on the proof of~(\ref{del_SS}) to~(\ref{QZZQZZ_exp}). By the definitions of $\delta$, $\omega$, and $\overline{\omega}$, we have $\frac{\omega}{\delta}=\frac{L}{M(1+\delta\overline{\omega})}$ and
 \begin{equation} 
\begin{aligned}
\frac{\omega}{\delta(1+\delta_z \overline{\omega}\overline{\omega}_z)}&=\frac{\omega}{\delta[1+\delta_z\overline{\omega}_z-\omega\delta_z\overline{\omega}\overline{\omega}_z ]  }
=\frac{1}{[\frac{M}{L}(1+\delta_z\overline{\omega}_z)(1+\delta\overline{\omega})-\delta\delta_z\overline{\omega}\overline{\omega}_z ]  }
\\
&
=\frac{\omega_z}{\delta_z[1+\delta\overline{\omega}-\omega_z\delta\overline{\omega}\overline{\omega}_z ]  }
=\frac{\omega_z}{\delta_z(1+\delta \overline{\omega}\overline{\omega}_z)},
\end{aligned}
\end{equation}
which concludes~(\ref{del_SS}). The proof of~(\ref{QzQ_exp})-(\ref{QZZQZZ_exp}) relies on the integration by parts formula and Nash-Poincar{\'e} inequality. Specifically, we use the integration by parts formula to set up equations with respect to the concerned evaluation and the Nash-Poincar{\'e} inequality to bound the error terms. Next, we will prove~(\ref{QzQ_exp}) by first evaluating $\frac{\E\Tr\BQ(z)\BQ\BH\BH^{H} }{M}$. With the integration by parts formula~(\ref{int_part}), we have
\begin{equation}
\label{QzQHH}
\begin{aligned}
&\frac{\E[\Tr(\BQ(z)\BQ\BH\BH^{H})] }{M}
\\
&
=\frac{1}{M}\sum_{i,j}\E [Y_{j,i}^{*}[\BZ^{H}\BQ(z)\BQ\BH]_{j,i}]
=\frac{\E[\Tr(\BQ(z)\BQ\BZ\BZ^{H})]}{M}-\E\Bigl[\frac{\Tr(\BQ(z)\BZ\BZ^{H})}{M}\frac{\Tr(\BQ(z)\BQ\BH\BH^{H})}{M} \Bigr]
\\
&
-\E\Bigl[ \frac{\Tr(\BQ(z)\BQ\BZ\BZ^{H})}{M}\frac{\Tr(\BQ\BH\BH^{H})}{M}\Bigr]
=
\frac{\E[\Tr(\BQ(z)\BQ\BZ\BZ^{H})]}{M}
-\frac{\E[\Tr(\BQ(z)\BZ\BZ^{H})]}{M}\frac{\E[\Tr(\BQ(z)\BQ\BH\BH^{H})]}{M}
\\
&
- \frac{\E[\Tr(\BQ(z)\BQ\BZ\BZ^{H})]}{M}\frac{\E[\Tr(\BQ\BH\BH^{H})]}{M}+\varepsilon_{H,1}+\varepsilon_{H,2},
\end{aligned}
\end{equation}
where $\varepsilon_{H,1}=-\frac{1}{M^2}\cov(\Tr(\BQ(z)\BZ\BZ^{H}),\Tr(\BQ(z)\BQ\BH\BH^{H}))$ and $\varepsilon_{H,2}=-\frac{1}{M^2}\cov(\Tr(\BQ(z)\BQ\BZ\BZ^{H}),\Tr(\BQ\BH\BH^{H}))$. By using the Cauchy-Schwarz inequality and the variance control in~(\ref{var_Q}) and~(\ref{var_QHH}) of Lemma~\ref{var_con}, $|\varepsilon_{H,1}|$ and $|\varepsilon_{H,2}|$ can be bounded by
\begin{equation}
\begin{aligned}
\label{var_con_H}
&|\varepsilon_{H,1}|\le \frac{1}{M^2}\Var^{\frac{1}{2}}(\Tr\BQ(z)\BZ\BZ^{H})
\Var^{\frac{1}{2}}(\Tr(\BQ(z)\BQ\BH\BH^{H}))
=\BO\Bigl(\frac{\mathcal{P}(z^{-1})}{M^2 z^2}\Bigr),
\\
&|\varepsilon_{H,2}|\le \frac{1}{M^2}\Var^{\frac{1}{2}}(\Tr(\BQ(z)\BQ\BZ\BZ^{H}))
 \Var^{\frac{1}{2}}(\Tr(\BQ\BH\BH^{H}))
=\BO\Bigl(\frac{\mathcal{P}(z^{-1})}{M^2 z}\Bigr),
\end{aligned}
\end{equation}
where $\mathcal{P}(\cdot)$ denotes a polynomial with positive coefficients. With the results of~(\ref{appro_del}) in Lemma~\ref{appro_lem},~(\ref{QzQHH}) can be further written as
\begin{equation}
\begin{aligned}
\label{QzHHQ2}
&\frac{\E[\Tr(\BQ(z)\BQ\BH\BH^{H})] }{M}=(1-\omega\overline{\omega})\frac{\E[\Tr(\BQ(z)\BQ\BZ\BZ^{H})]}{M}
-\frac{\omega_z\E [\Tr(\BQ(z)\BQ\BH\BH^{H})]}{M}+\BO\Bigl(\frac{\mathcal{P}(\frac{1}{z})}{zM^2}\Bigr)
\\
&
\overset{(b)}{=}\frac{\overline{\omega}\overline{\omega}_z\E[\Tr(\BQ(z)\BQ\BZ\BZ^{H})]}{M}+\BO\Bigl(\frac{\mathcal{P}(\frac{1}{z})}{zM^2}\Bigr),
\end{aligned}
\end{equation}
where step $(b)$ follows by moving $\frac{\omega_z\E \Tr\BQ(z)\BQ\BH\BH^{H}}{M}$ to the LHS of the first line in~(\ref{QzHHQ2}) to solve $\frac{\E \Tr\BQ(z)\BQ\BH\BH^{H}}{M}$. Then we turn to evaluate $\frac{\E\Tr\BQ(z)\BQ\BZ\BZ^{H}}{M}$. By the integration by parts formula~(\ref{int_part}), we have
\begin{equation}
\begin{aligned}
\label{QzQZZ}
&
\frac{1}{M}\E[\Tr(\BQ(z)\BQ\BZ\BZ^{H})]=\frac{1}{M}\sum_{i,j}\E[ Z_{j,i}^{*}[\BQ(z)\BQ\BZ]_{j,i}]
=\frac{\E[\Tr(\BQ\BQ(z))]}{M}-\frac{\E[\Tr(\BQ(z))\Tr(\BQ(z)\BQ\BH\BH^{H}) }{M}
\\
&
-\frac{\E[\Tr(\BQ(z)\BQ) }{L}\frac{\Tr(\BQ\BH\BH^{H})]}{M}
\\
&\overset{(a)}{=}\frac{\E[\Tr(\BQ\BQ(z))]}{M}-\frac{\delta_z\E[\Tr(\BQ(z)\BQ\BH\BH^{H}) }{M}
-\frac{\omega\overline{\omega}\E[\Tr(\BQ(z)\BQ)] }{L}
+\BO\Bigl(\frac{\mathcal{P}(\frac{1}{z})}{z M^2}\Bigr)
\\
&=(\frac{L}{M}-\omega\overline{\omega})\frac{\E[\Tr(\BQ\BQ(z))]}{L}
-\frac{\delta_z\E[\Tr(\BQ(z)\BQ\BH\BH^{H})] }{M}+\BO\Bigl(\frac{\mathcal{P}(\frac{1}{z})}{z M^2}\Bigr),
\end{aligned}
\end{equation}
where step $(a)$ in~(\ref{QzQZZ}) follows from~(\ref{appro_del}) in Lemma~\ref{appro_lem} and the variance control in~(\ref{var_con_H}). By substituting~(\ref{QzHHQ2}) into~(\ref{QzQZZ}) to replace $\frac{\E[\Tr(\BQ(z)\BQ\BH\BH^{H})]}{M}$ and solving $\frac{\E[\Tr(\BQ(z)\BQ\BZ\BZ^{H})]}{M}$, we have
\begin{equation}
\label{QzQZZ2}
\begin{aligned}
\frac{1}{M}\E[\Tr(\BQ(z)\BQ\BZ\BZ^{H})]=\frac{(\frac{L}{M}-\omega\overline{\omega})}{1+\delta_z \overline{\omega}\overline{\omega}_z}\frac{\E[\Tr(\BQ\BQ(z))]}{L}+\BO\Bigl(\frac{\mathcal{P}(\frac{1}{z})}{z M^2}\Bigr)
=\frac{\omega}{\delta(1+\delta_z \overline{\omega}\overline{\omega}_z)}\frac{\E[\Tr(\BQ\BQ(z))]}{L}+\BO\Bigl(\frac{\mathcal{P}(\frac{1}{z})}{z M^2}\Bigr).
\end{aligned}
\end{equation}

By using the resolvent identity $\sigma^2\BQ+\BQ\BH\BH^{H}=\bold{I}_{N}$ to replace $\BQ\BH\BH^{H}$ in the second last term of~(\ref{QzQZZ}) and combining~(\ref{QzQZZ}) with~(\ref{QzQZZ2}), we can obtain the approximation of $\frac{\E\Tr\BQ\BQ(z)}{L}$ as
\begin{equation}
\label{app_QQz}
\begin{aligned}
&\frac{\E[\Tr(\BQ\BQ(z))]}{L}=\frac{\frac{L\delta_z^2 }{M}+\BO(\frac{\mathcal{P}(\frac{1}{z})}{z M^2})}{(\frac{L}{M}-\omega\overline{\omega}+\frac{L\sigma^2\delta_z}{M}- \frac{\omega}{\delta(1+\delta_z \overline{\omega}\overline{\omega}_z)})}
=\frac{\delta_z }{\sigma^2+ \frac{M\omega\overline{\omega}\overline{\omega}_z}{L\delta(1+\delta_z \overline{\omega}\overline{\omega}_z)}}+\BO\Bigl(\frac{\mathcal{P}(\frac{1}{z})}{z M^2}\Bigr)
\\
&
\overset{(a)}{=}\frac{\delta_z }{\sigma^2+ \frac{M\omega_z\overline{\omega}\overline{\omega}_z}{L\delta_z(1+\delta \overline{\omega}\overline{\omega}_z)}}+\BO\Bigl(\frac{\mathcal{P}(\frac{1}{z})}{z M^2}\Bigr)
=\frac{\delta_z}{\Delta_{\sigma^2}(z)}+\BO\Bigl(\frac{\mathcal{P}(\frac{1}{z})}{z M^2}\Bigr),
\end{aligned}
\end{equation}
where $\Delta_{\sigma^2}(z)=\sigma^2+ \frac{M\omega_z\overline{\omega}\overline{\omega}_z}{L\delta_z(1+\delta \overline{\omega}\overline{\omega}_z)}$ and step $(a)$ in~(\ref{app_QQz}) follows from~(\ref{del_SS}). This concludes the proof of~(\ref{QzQ_exp}).~(\ref{QzQZZ_exp}) can be obtained by~(\ref{QzQZZ2}). With~(\ref{QzQZZ_exp}),~(\ref{exp_QQHH}) can be then obtained by~(\ref{QzHHQ2}).

Next, we turn to evaluate $\frac{\E[\Tr(\BQ \BZ  \BZ^H \BQ(z)\BH \BH^{H})]}{M}$ in~(\ref{exp_QzQHH}). It follows from the resolvent identity
\begin{equation}
\label{app_QZZQzHH}
\begin{aligned}
&\frac{\E[\Tr(\BQ \BZ  \BZ^H \BQ(z)\BH \BH^{H})]}{M}= \omega_z -\frac{\sigma^2\E[\Tr(\BQ\BZ  \BZ^H \BQ(z))]}{M}
+\BO\Bigl(\frac{\mathcal{P}(\frac{1}{z})}{M^2z }\Bigr)
=\omega_z-\frac{\sigma^2\omega}{\delta(1+\delta_z \overline{\omega}\overline{\omega}_z)}\frac{\delta_z }{\Delta_{\sigma^2}(z)}+\BO\Bigl(\frac{\mathcal{P}(\frac{1}{z})}{M^2z }\Bigr)
\\
&
\overset{(a)}{=}\omega_z- (\Delta_{\sigma^2}(z)-\frac{M\omega_z\overline{\omega}\overline{\omega}_z}{L\delta_z(1+\delta \overline{\omega}\overline{\omega}_z)})
\frac{\omega\delta_z}{\delta(1+\delta_z \overline{\omega}\overline{\omega}_z)\Delta_{\sigma^2}(z)}+\BO(\frac{\mathcal{P}(\frac{1}{z})}{M^2z })
\\
&
\overset{(b)}{=}\omega_z
-\frac{\omega_z \delta_z}{\delta_z (1+\delta \overline{\omega}\overline{\omega}_z)}+\frac{M\overline{\omega}\omega_z^2\overline{\omega}_z\delta_z }{L\delta_z^2(1+\delta\overline{\omega}\overline{\omega}_z )^2\Delta_{\sigma^2}(z)}+\BO\Bigl(\frac{\mathcal{P}(\frac{1}{z})}{M^2z }\Bigr)
\\
&=\frac{\delta\overline{\omega}\omega_z\overline{\omega}_z}{1+\delta\overline{\omega}\overline{\omega}_z}
+\frac{M\overline{\omega}\omega_z^2\overline{\omega}_z\delta_z }{L\delta_z^2(1+\delta\overline{\omega}\overline{\omega}_z )^2\Delta_{\sigma^2}(z)}
+\BO\Bigl(\frac{\mathcal{P}(\frac{1}{z})}{z M^2}\Bigr),
\end{aligned}
\end{equation}
where step $(a)$ in~(\ref{app_QZZQzHH}) follows from the definition of $\Delta_{\sigma^2}(z)$ in~(\ref{Delta_exp}) and step $(b)$ follows from~(\ref{del_SS}). This concludes the proof of~(\ref{exp_QzQHH}). Next, we will evaluate $\frac{\E[\Tr(\BQ \BZ  \BZ^H \BQ(z)\BZ\BZ^{H})]}{M}$ in~(\ref{QZZQZZ_exp}). Notice that $\frac{\E[\Tr(\BQ \BZ  \BZ^H \BQ(z)\BH \BH^{H})]}{M}$ can be written by integration by parts formula~(\ref{int_part}) and the variance control as
\begin{equation}
\begin{aligned}
\label{QZZQZHHalt}
&\frac{\E[\Tr(\BQ \BZ  \BZ^H \BQ(z)\BH \BH^{H})]}{M}=\frac{\E[\Tr(\BQ \BZ  \BZ^H \BQ(z)\BZ\BZ^{H})]}{M}
-\frac{\omega_z\overline{\omega}_z\E[\Tr(\BQ \BZ  \BZ^H \BQ(z)\BZ\BZ^{H})]}{M}
\\
&
-\frac{\omega\E[\Tr(\BQ \BZ  \BZ^H \BQ(z)\BH \BH^{H})]}{M}+\BO\Bigl(\frac{\mathcal{P}(\frac{1}{z})}{zM^2}\Bigr)
\\
&=
\frac{\overline{\omega}\overline{\omega}_z\E[\Tr(\BQ \BZ  \BZ^H \BQ(z)\BZ\BZ^{H})]}{M}+\BO\Bigl(\frac{\mathcal{P}(\frac{1}{z})}{z M^2}\Bigr).
\end{aligned}
\end{equation}
Similarly, we have
\begin{equation}
\label{QZZQZZz}
\begin{aligned}
&\frac{\E[\Tr(\BQ \BZ  \BZ^H \BQ(z)\BZ \BZ^{H})]}{M}
=\delta\omega_{z}
+(1-\frac{M\omega_z\overline{\omega}_z}{L})\frac{\E[\Tr(\BQ\BZ\BZ^{H}\BQ(z))]}{M}
-\frac{\delta\E[\Tr(\BQ\BZ\BZ^{H}\BQ(z)\BH\BH^{H})]}{M}
+\BO\Bigl(\frac{\mathcal{P}(\frac{1}{z})}{zM^2}\Bigr)
\\
&
\overset{(a)}{=}\frac{\delta\omega_z}{1+\delta\overline{\omega}\overline{\omega}_z}
+
\frac{M\omega_z }{L\delta_z(1+\delta\overline{\omega}\overline{\omega}_z )}
\frac{\E[\Tr(\BQ \BZ  \BZ^H \BQ(z))]}{M}
+\BO\Bigl(\frac{\mathcal{P}(z^{-1})}{zM^2}\Bigr),
\end{aligned}
\end{equation}
where $(a)$ is obtained by substituting $\E[\Tr(\BQ\BZ\BZ^{H}\BQ(z)\BH\BH^{H})]$ using~(\ref{QZZQZHHalt}) and then solving $\frac{\E[\Tr(\BQ \BZ  \BZ^H \BQ(z)\BZ \BZ^{H})]}{M}$. By noticing that $\frac{M\omega(1+\delta\overline{\omega})}{L}=\delta$, we have 
\begin{equation}
\begin{aligned}
&\frac{\E[\Tr(\BQ \BZ  \BZ^H \BQ(z)\BZ \BZ^{H})]}{M}
=\frac{M\omega_z\omega(1+\delta\overline{\omega})}{L(1+\delta\overline{\omega}\overline{\omega}_z)}
+
\frac{M\omega_z }{L\delta_z(1+\delta\overline{\omega}\overline{\omega}_z )}
\frac{\E[\Tr(\BQ \BZ  \BZ^H \BQ(z))]}{M} +\BO(\frac{\mathcal{P}(\frac{1}{z})}{M^2z}),
\end{aligned}
\end{equation}
which concludes~(\ref{QZZQZZ_exp}).
\end{proof}

\section{Proof of Theorem~\ref{clt_the}}
\label{prof_the2}
We first introduce notations and the proof idea. In the following, we will ignore the superscript of $\BW^{(n)}$ and $\BX^{(n)}$ for simplicity. Denote
\begin{equation}
\begin{aligned}
\label{BA_proof}
\BA=\BA_{n}&=\BI_{M}-\frac{\BX^{(n)}(\BX^{(n)})^{H}}{n},
\\
\gamma_{n}^{\BW,\BY,\BZ} &= \sqrt{nM}I_{N,L,M}^{(n)},
\\
\Phi^{\BW,\BY,\BZ}(u) & = e^{\jmath u \gamma_{n}^{\BW,\BY,\BZ}}
\end{aligned}
\end{equation}
The characteristic function of the MID $\Psi^{\BW,\BY,\BZ}_{n}(u)$ is given by
\begin{equation}
\Psi^{\BW,\BY,\BZ}_{n}(u)=\E[\Phi^{\BW,\BY,\BZ}_{n}]. 
\end{equation}
Define 
\begin{equation}
\label{M_fun_def}
\mathcal{M}(x)=\min(1,x^2), ~\mathcal{M}(A,x)=\min(A^{-1},x^2).
\end{equation}

\textbf{Proof idea:} To show that the asymptotic distribution of $\gamma_{n}^{\BW,\BY,\BZ}= \sqrt{nM}I_{N,L,M}^{(n)}$ converges to the Gaussian distribution, we first show that the characteristic function of $\gamma_n$ converges to that of the Gaussian distribution, i.e.,
\begin{equation}
\label{cha_con_eq1}
\begin{aligned}
\Psi^{\BW,\BY,\BZ}(u)=e^{\jmath u \sqrt{nM}\times \overline{C}(\sigma^2) -\frac{u^2 V_n}{2}  }+E(u,\BA),
\end{aligned}
\end{equation}
where $ E(u,\BA)  \xrightarrow[]{{N  \xrightarrow[]{\rho,\eta, \kappa}\infty}} 0$ and  $V_n$ is the asymptotic variance. This approach has been used in the second-order analysis of MIMO channels~\cite{hachem2008new,zhang2022asymptotic,zhang2022secrecy}. It worth noticing that the expectation $\E $ is taken over three random matrices $\BW^{(n)}$, $\BY$, $\BZ$ and it is challenging to evaluate $\Psi^{\BW,\BY,\BZ}(u)$ directly due to the exponential structure. Therefore, we resort to consider the derivative $\frac{\partial \Psi^{\BW,\BY,\BZ}(u)}{\partial u}$ so that the Gaussian tools, i.e. the integration by parts formula~\cite[Eq. (17)]{hachem2008new} and Nash-Poincar{\'e} inequality\cite[Eq. (18)]{hachem2008new} can be used to fully exploit the Gaussianity of the matrices $\BW^{(n)}$, $\BY$, $\BZ$. In particular, we will show
\begin{equation}
\label{der_obj}
\begin{aligned}
\frac{\partial \Psi^{\BW,\BY,\BZ}(u)}{\partial u} &= [\jmath  \sqrt{nM} \overline{C}(\sigma^2) -u V_n]  
\\
&\times
e^{\jmath u \sqrt{nM}\times \overline{C}(\sigma^2) -\frac{u^2 V_n}{2}  }+E'(u,\BA).
\end{aligned}
\end{equation}
Given the right hand side of~(\ref{der_obj}) does not contain expectation, we need to take expectation over $\BW$, $\BZ$ and $\BY$ iteratively at the left hand side of~(\ref{der_obj}). Finally, we will show the approximation error of the CDF is $\BO(n^{-\frac{1}{4}})$ in~(\ref{prob_con_rate}). Specifically, the proof can be summarized to four steps:
\begin{itemize}
\item[A.] Take expectation over $\BW$ to obtain $\frac{\partial  \Psi^{\BW,\BY,\BZ}(u)}{\partial u}=\E [(\jmath \mu^{\BY,\BZ}_{n}-u\nu^{\BY,\BZ}_{n} )\Phi_{n}^{\BY,\BZ}]+E_W(u,\BA) $, where $\Phi_{n}^{\BY,\BZ}$, $\mu^{\BY,\BZ}_{n}$, and $\nu^{\BY,\BZ}_{n}$ only depend on $\BY$ and $\BZ$.
\item[B.] Take expectation over $\BZ$ and $\BY$ for $\E [(\jmath \mu^{\BY,\BZ}_{n}-u\nu^{\BY,\BZ}_{n} )\Phi_{n}^{\BY,\BZ}]$ to obtain~(\ref{der_obj}) and compute the closed-form expression for the asymptotic variance $V_n$.
\item[C.] Take integral over~(\ref{der_obj}) to obtain~(\ref{cha_con_eq1}). Then, by L{\'e}vy’s continuity theorem~\cite{billingsley2017probability}, we could conclude that $\frac{\gamma_n-\overline{\gamma}_n}{\sqrt{V}_n} \xrightarrow[{N  \xrightarrow[]{\rho,\eta, \kappa}\infty}]{\mathcal{D}} \mathcal{N}(0,1)$.
\item[D.]  To conclude~(\ref{prob_con_rate}), we analyze the order $u$ in the error term $E(u,\BA)$ and bound the difference between the CDF of  $\frac{{\gamma}_n-\overline{C}(\sigma^2)}{\sqrt{V}_n}$ and that of Gaussian variable using Esseen inequality~\cite[p538]{feller1991introduction}.
\end{itemize}
The detailed proof is given in the following.

\subsection{Step:1 Expectation over $\BW$}
In this step, we will provide an approximation for $\frac{\partial \Psi^{ \BW,\BY,\BZ}(u)}{\partial u}$, which only relies on $\BY$ and $\BZ$ by taking the expectation over $\BW$. Since the approximation error for the resolvents over Rayleigh-product channels has the same order as that of the Rayleigh channel, the approximation for the characteristic function in~\cite[Appendix D.D, Eq. (221) to (240)]{hoydis2015second} is also applicable here. Specifically, in the asymptotic regime defined by~\textbf{Assumption~{A}}, the derivative of the characteristic function $ \Psi^{ \BW,\BY,\BZ}(u)$ can be approximated by 
\begin{equation}
\label{chara_lemma}
\begin{aligned}
\frac{\partial  \Psi^{\BW,\BY,\BZ}(u)}{\partial u} &=\E [(\jmath \mu^{\BY,\BZ}_{n}-u\nu^{\BY,\BZ}_{n} )\Phi_{n}^{\BY,\BZ}]
+ \BO\Bigl(\frac{u^2}{M}+\frac{u^3\mathcal{M}(u)}{M^2}+\frac{u^4 \mathcal{M}(u)}{M^3}\Bigr),
\end{aligned}
\end{equation}
where
\begin{equation}
\Phi_{n}^{\BY,\BZ}=e^{ \alpha^{\BY,\BZ}_{n}},
\end{equation}
\begin{equation}
\alpha^{\BY,\BZ}_{n}=\jmath u \mu^{\BY,\BZ}_{n}-\frac{u^2}{2}\nu^{\BY,\BZ}_{n}
+\frac{\jmath u^3 \beta^{\BY,\BZ}_{n}}{3},
\end{equation}
\begin{equation}
\begin{aligned}
\mu_{n}^{\BY,\BZ}&=\sqrt{\frac{n}{M}}\log\det\left(\bold{I}_N+\frac{1}{\sigma^2}\BH\BH^{H}\right)-\frac{n\Tr(\BQ\BH\BA\BH^{H})}{\sqrt{Mn}},
\\
\nu_n^{\BY,\BZ}&=\frac{n}{Mn}\Tr((\BQ\BH\BH^{H})^2)+\frac{2\sigma^2n}{Mn}\Tr\Bigl(\BQ^2\BH\frac{\BX\BX^{H}}{n}\BH^{H}\Bigr),
\\
\beta_n^{\BY,\BZ}&=\frac{n}{\sqrt{n^3M^3}}\Tr((\BQ\BH\BH^{H})^3)
\\
&
+\frac{3\sigma^2n}{\sqrt{n^3M^3}}\Tr\Bigl(\BQ^2\BH\BH^{H}\BQ\BH\frac{\BX\BX^{H}}{n}\BH^{H}\Bigr),
\end{aligned}
\end{equation}
Function $\mathcal{M}(x)$ is defined in~(\ref{M_fun_def}). Noticing that for $\alpha\ge 1$ and $A>0$, we have 
\begin{equation}
\begin{aligned}
& e^{-\frac{A u^2}{2}}\int_{0}^{u} x^{\alpha} e^{\frac{A x^2}{2}}\mathrm{d} x   \le u^{\alpha-1} e^{-\frac{A u^2}{2}}\int_{0}^{u}  x e^{\frac{A x^2}{2}} \mathrm{d} x
\\
&
=u^{\alpha-1} A^{-1}(1-e^{-\frac{A u^2}{2}})=\BO( u^{\alpha-1}\mathcal{M}(A,u) ),
\end{aligned}
\end{equation}
where $\mathcal{M}(X,u)$ is given in~(\ref{M_fun_def}). We can further approximate $ \Psi^{\BW,\BY,\BZ}_{n}(u)$ by taking integral over both sides of~(\ref{chara_lemma}) as
\begin{equation}
 \Psi^{\BW,\BY,\BZ}_{n}(u)=\E[ \Phi_{n}^{\BY,\BZ}] +\varepsilon_{w}(u),
\end{equation}
where $\varepsilon_{w}(u)=\E[ \Phi_{n}^{\BY,\BZ}  \int_{0}^{u} e^{- x\alpha^{\BY,\BZ}_{n} }\BO(\frac{x^2}{M^2}) \mathrm{d} x]$ can be bounded by
\begin{equation}
\label{epwu}
|\varepsilon_{w}(u)| \le | \E[ e^{-\frac{\nu_{n}^{\BY,\BZ} u^2}{2}} \int_{0}^{u} e^{\frac{\nu_{n}^{\BY,\BZ} u^2}{2}}\BO(\frac{x^2}{M^2}) \mathrm{d} x] |
=\BO(\frac{u  \mathcal{M}(u)}{M^2}),
\end{equation}
since $\nu_n^{\BY,\BZ}>\underline{\nu}>0$ is bounded away from zero almost surely. This indicates that
\begin{equation}
\begin{aligned}
\label{eq_cha_cha}
 \Psi^{\BW,\BY,\BZ}_{n}(u) &=\E_{\BZ,\BY}[\E_{\BW} [\Phi^{\BW,\BY,\BZ}_n (u)]]
\\
& =\E [\Phi_{n}^{\BY,\BZ}]+\BO\Bigl(\frac{u \mathcal{M}(u)}{M^2}\Bigr).
\end{aligned}
\end{equation}
With~(\ref{chara_lemma}) and~(\ref{eq_cha_cha}), we can discard the dependence on $\BW$ and turn to evaluate $\E[ (\jmath \mu^{\BY,\BZ}_{n}-u\nu^{\BY,\BZ}_{n} )\Phi_{n}^{\BY,\BZ}]$ by taking the expectation with respect to $\BY$ and $\BZ$. It is worth noticing that we have carefully collected the $u$-related terms in the error term since the order of $u$ is important in analyzing the approximation error of the CDF.

\subsection{Step 2: Expectation Over $\BZ$ and $\BY$: Evaluation of $\E  [(\jmath\mu^{\BY,\BZ}_{n}-u\nu^{\BY,\BZ}_{n} ) \Phi_{n}^{\BY,\BZ}]$  }
\label{appro_charaf}
In this step, we will utilize Gaussian tool, i.e., integration by parts formula~\cite[Eq. (17)]{hachem2008new} and the variance control in Lemma~\ref{var_con} to compute the expectation over $\BZ$ and $\BY$ to prove~(\ref{der_obj}). To this end, we first decompose $\E [ (\jmath\mu^{\BY,\BZ}_{n}-u\nu^{\BY,\BZ}_{n} ) \Phi_{n}^{\BY,\BZ}]$ into $U_1$ and $U_2$ as
 \begin{equation}
\label{exp_U1U2}
\begin{aligned}
& \E [(\jmath\mu^{\BY,\BZ}_{n}-u\nu^{\BY,\BZ}_{n} ) e^{ \alpha^{\BY,\BZ}_{n}}] =\E[\jmath  \mu^{\BY,\BZ}_{n}\Phi_{n}^{\BY,\BZ}]
-u\E [\nu^{\BY,\BZ}_{n}\Phi_{n}^{\BY,\BZ}]
=U_1+U_2.
\end{aligned}
\end{equation}
\subsubsection{The evaluation of $U_1$}
$U_1$ can be rewritten as~\cite[Eq. (4)]{zhang2021bias}
\begin{equation}
\begin{aligned}
\label{U1_step1}
U_1=(\jmath\sqrt{\frac{n}{M}}\int_{\sigma^2}^{\infty} \frac{N\E[\Phi_{n}^{\BY,\BZ}]}{z}
-\E[\Tr(\BQ(z))\Phi_{n}^{\BY,\BZ}]\mathrm{d}z)
-\frac{n}{\sqrt{Mn}}\E[\Tr(\BQ\BH\BA\BH^{H})\Phi_{n}^{\BY,\BZ}]
=U_{1,1}+U_{1,2}.
\end{aligned}
\end{equation}
In the following, we will evaluate $U_{1,1}$ and $U_{1,2}$, respectively. We first evaluate $\E[\Tr(\BQ(z))\Phi_{n}^{\BY,\BZ}]$ in $U_{1,1}$, which can be converted to the evaluation of $\E[\Tr(\BQ(z)\BH\BH^{H})\Phi_{n}^{\BY,\BZ}]$ by the resolvent identity $\BI_{N}=z\BQ(z)+\BQ(z)\BH\BH^{H}$. By the integration by parts formula in~(\ref{int_part}), we can obtain
\begin{align*}
\label{fir_app}
&\E[\Tr(\BQ(z)\BH\BH^{H})\Phi_{n}^{\BY,\BZ}]
=\E[ Y_{j,i}^{*} [\BZ^H \BQ(z)\BH ]_{j,i}\Phi_{n}^{\BY,\BZ}]
=\E\Bigl[ [\Tr(\BQ(z)\BZ\BZ^{H }) -\frac{  \Tr(\BQ(z)\BZ\BZ^{H })}{M}\Tr (\BQ(z)\BH\BH^{H})
\\
&+\frac{\jmath u \sqrt{n}}{\sqrt{M^3}}  \Tr(\BQ\BZ   \BZ^H \BQ(z)\BH\BH^{H} )
- \frac{\jmath u n}{\sqrt{M^3 n}}\Tr( \BZ  \BZ^H \BQ(z)\BH \BA\BH^{H}\BQ)
+\frac{\jmath u  n}{\sqrt{M^3n}}\Tr (\BZ  \BZ^H \BQ(z)\BH\BH^{H}\BQ\BH \BA\BH^{H}\BQ)
\\
&
-\sum_{i,j}\frac{u^2}{2M}\frac{\partial \nu_{n}^{\BY,\BZ}}{\partial Y_{j,i}}[\BZ^H \BQ(z)\BH ]_{j,i}+\sum_{i,j}\frac{\jmath u^3}{3M}\frac{\partial \beta_{n}^{\BY,\BZ}}{\partial Y_{j,i}}[\BZ^H \BQ(z)\BH ]_{j,i}]\Phi_{n}^{\BY,\BZ}\Bigr]
\overset{(a)}{=} 
 \frac{\E[\Tr(\BQ(z)\BZ\BZ^{H})]}{M} \E[\Tr(\BQ(z)\BH\BH^{H})]\E[\Phi_{n}^{\BY,\BZ}]
\\
&
+
\E [\Tr(\BQ(z)\BZ\BZ^{H})\Phi_{n}^{\BY,\BZ}](1-\frac{\E[\Tr(\BQ(z)\BH\BH^{H})]}{M})
- \frac{\E[\Tr(\BQ(z)\BZ\BZ^{H})]}{M} \E[ \Tr(\BQ(z)\BH\BH^{H})\Phi_{n}^{\BY,\BZ}]
\\
&
+
\frac{\jmath u \sqrt{n}}{\sqrt{M^3}}  \Tr(\BQ\BZ   \BZ^H \BQ(z)\BH\BH^{H})\Phi_{n}^{\BY,\BZ}
+\E[ \varepsilon_1]\E[\Phi_{n}^{\BY,\BZ}]+\E[\varepsilon_2]\E[\Phi_{n}^{\BY,\BZ}]+\E[\varepsilon_3]\E[\Phi_{n}^{\BY,\BZ}]+\E[\varepsilon_4]\E[\Phi_{n}^{\BY,\BZ}]
\\
&
+\BO\Biggl(\frac{\mathcal{P}_1(\frac{1}{z})}{zM} +\frac{u\mathcal{P}_2(\frac{1}{z})\sqrt{\Tr(\BA^2)}}{M^\frac{3}{2}}
+\frac{u^2\mathcal{P}_3(\frac{1}{z})(1+\sqrt{\frac{\Tr(\BA^2)}{M}})}{M^{2}}+\frac{u^3\mathcal{P}_4(\frac{1}{z}) (1+\sqrt{\frac{\Tr(\BA^2)}{M}})}{M^{3}}
  \Biggr)  \numberthis
\\
&
\overset{(b)}{=} M\omega^2_{z}\overline{\omega}^2_{z} \E[\Phi_{n}^{\BY,\BZ}]  +   \overline{\omega}^2_{z} \E  [\Tr(\BQ(z)\BZ\BZ^{H})\Phi_{n}^{\BY,\BZ}]
+\frac{\jmath u\overline{\omega}_z \sqrt{n}}{\sqrt{M^3}} \E[ \Tr(\BQ\BZ   \BZ^H \BQ(z)\BH\BH^{H}) \Phi_{n}^{\BY,\BZ}]
\\
&
+ \overline{\omega}_{z}\E[ \varepsilon_1]\E[\Phi_{n}^{\BY,\BZ}]+ \overline{\omega}_{z}\E[\varepsilon_2]\E[\Phi_{n}^{\BY,\BZ}]+ \overline{\omega}_{z}\E[\varepsilon_3]\E[\Phi_{n}^{\BY,\BZ}]
+ \overline{\omega}_{z}\E[\varepsilon_4] \E[ \Phi_{n}^{\BY,\BZ}]
\\
&
+\BO\Biggl(\frac{\mathcal{P}(\frac{1}{z})}{z}\Bigl(\frac{1}{M} +\frac{u\sqrt{\Tr(\BA^2)}}{M^\frac{3}{2}}
+\frac{u^2(1+\sqrt{\frac{\Tr(\BA^2)}{M}})}{M^{2}}+\frac{u^3 (1+\sqrt{\frac{\Tr(\BA^2)}{M}})}{M^{3}}
  \Bigr)\Biggr),
\end{align*}
where $\mathcal{P}_i(\cdot)$, $i=1,2,3,4$ and $\mathcal{P}(\cdot)$ are polynomials with positive coefficients and
\begin{align*}
\varepsilon_1&=- \frac{\jmath u n}{\sqrt{M^3 n}}\Tr (\BZ  \BZ^H \BQ(z)\BH \BA\BH^{H}\BQ), 
\\
\varepsilon_2&=\frac{\jmath u  n}{\sqrt{M^3n}}\Tr( \BZ  \BZ^H \BQ(z)\BH\BH^{H}\BQ\BH \BA\BH^{H}\BQ),
\\
\varepsilon_3&=-\sum_{i,j}\frac{u^2}{2M}\frac{\partial \nu_{n}^{\BY,\BZ}}{\partial Y_{j,i}}[\BZ^H \BQ(z)\BH ]_{j,i}, \numberthis
\\
\varepsilon_4&=\sum_{i,j}\frac{\jmath u^3}{3M}\frac{\partial \beta_{n}^{\BY,\BZ}}{\partial Y_{j,i}}[\BZ^H \BQ(z)\BH ]_{j,i}.
\end{align*}
Step $(a)$ in~(\ref{fir_app}) follows from
\begin{equation}
\label{cov_expand}
\begin{aligned}
\E [ab \Phi_{n}^{\BY,\BZ}]& =\E [a] \E[ b \Phi_{n}^{\BY,\BZ}]+ \E[ \underline{a}b \Phi_{n}^{\BY,\BZ}]
\\
&=\E [a] \E[ b \Phi_{n}^{\BY,\BZ}]+\E[ b] \E[ a\Phi_{n}^{\BY,\BZ}] -\E[ a]\E[ b ]\E[ \Phi_{n}^{\BY,\BZ} ]+\E[\underline{a}{\Phi_{n}^{\BY,\BZ}}\underline{b}]
\\
& = \E [a] \E[ b \Phi_{n}^{\BY,\BZ}]+\E[ b] \E[ a\Phi_{n}^{\BY,\BZ}] -\E [a]\E [b] \E[ \Phi_{n}^{\BY,\BZ}]+\BO(\frac{\mathcal{P}(\frac{1}{z})}{Mz^2}),
\end{aligned}
\end{equation}
where $a= \frac{\Tr(\BQ(z)\BZ\BZ^{H})}{M}$, $b=\Tr(\BQ(z)\BH\BH^{H})$. The term $\BO(\frac{\mathcal{P}(\frac{1}{z})}{Mz^2})$ follows from~(\ref{var_QZZ}) and~(\ref{var_QHH}) of Lemma~\ref{var_con} in Appendix~\ref{app_mathtool} with
\begin{equation}
|\E[\underline{a}{\Phi_{n}^{\BY,\BZ}}\underline{b}]|
\le \Var^{\frac{1}{2}}(a)\Var^{\frac{1}{2}}(b)=\BO\Bigl(\frac{\mathcal{P}(\frac{1}{z})}{Mz^2}\Bigr).
\end{equation}
Step $(b)$ in~(\ref{fir_app}) follows from~(\ref{appro_del}) in Lemma~\ref{appro_lem}. Similarly, by Lemma~\ref{var_con}, we can obtain
\begin{equation}
\begin{aligned}
\E [\varepsilon_1\Phi_{n}^{\BY,\BZ}]=\E [\varepsilon_1]\E[\Phi_{n}^{\BY,\BZ}]+\cov(\varepsilon_1,\Phi_{n}^{\BY,\BZ}),
\end{aligned}
\end{equation}
and $\cov(\varepsilon_1,\Phi_{n}^{\BY,\BZ})\le \Var^{\frac{1}{2}}(\varepsilon_1)\Var^{\frac{1}{2}}(\Phi_{n}^{\BY,\BZ})= \BO(\frac{u\sqrt{\Tr(\BA^2)}\mathcal{P}_2(\frac{1}{z})}{zM^{\frac{3}{2}}}) $. By the integration by parts formula, we have
\begin{equation}
\begin{aligned}
\label{ZZQHHA}
&\E\varepsilon_1=-\frac{\jmath u n}{\sqrt{M^3 n}}\Tr( \BZ  \BZ^H \BQ(z)\BH \BA\BH^{H}\BQ)
=-\frac{\jmath u n}{\sqrt{M^3 n}}\Bigl(\frac{\E[\Tr(\BA\E\Tr\BZ  \BZ^H \BQ(z)\BH \BA\BH^{H}\BQ)]}{M}  
\\
&
-\frac{\E[\Tr(\BQ\BZ\BZ^{H}\BQ(z)\Tr\BQ(z)\BH\BA\BH^{H})]}{M}
-\frac{\E[\Tr(\BQ\BZ\BZ^{H})\Tr(\BZ  \BZ^H \BQ(z)\BH \BA\BH^{H}\BQ)]}{M}\Bigr)
\\
&=E_1+E_2+E_3.
\end{aligned}
\end{equation}
Noticing that $\Tr(\BA)=\Tr(\BI_{M}-\frac{\BX\BX^{H}}{n})=0$, we have $E_1=0$. By~(\ref{var_QHCH}) and $\Tr(\BA^2)=\BO(M^2)$, we have
\begin{equation}
\begin{aligned}
&|E_2|=\frac{ u n}{\sqrt{M^3n}}  \Bigl|\frac{\E[\Tr(\BQ\BZ\BZ^{H}\BQ(z)\E\Tr\BQ(z)\BH\BA\BH^{H})]}{M}
+\cov\Bigl(\Tr(\BQ\BZ\BZ^{H}\BQ(z)),\frac{\Tr(\BQ(z)\BH\BA\BH^{H})}{M}\Bigr)\Bigr|
\\
&
\le 
\frac{K u n}{\sqrt{M^3n}}\frac{\sqrt{\Tr(\BA^2)}}{z^2M^{\frac{3}{2}}}=\BO\Bigl(\frac{u\sqrt{\Tr(\BA^2)}}{z^2n^\frac{5}{2}}\Bigr).
\end{aligned}
\end{equation}
and similarly we can obtain 
\begin{equation}
\begin{aligned}
|E_3|
\le 
\frac{ u n}{\sqrt{M^3n}}\frac{\sqrt{\Tr(\BA^2)}}{z^2M}=\BO\Bigl(\frac{u\sqrt{\Tr(\BA^2)}}{z^2n^{\frac{5}{2}}}\Bigr).
\end{aligned}
\end{equation}
A similar analysis can be performed to obtain $\E [\varepsilon_2]\E[\Phi_{n}^{\BY,\BZ}]=\BO(\frac{u\sqrt{\Tr(\BA^2)}}{z^2n^\frac{5}{2}})$ and $\E [ \varepsilon_3]\E [\Phi_{n}^{\BY,\BZ}]=\E [ \varepsilon_3]\E[\Phi_{n}^{\BY,\BZ}]+  \BO\Bigl(\frac{\mathcal{P}_3(\frac{1}{z})(1+\sqrt{\frac{\Tr(\BA^2)}{M}}) }{M^2 z^2}\Bigr)$ in the last line of~(\ref{fir_app}). Now we turn to evaluate $\E \varepsilon_3$, which is performed as
\begin{equation}
\label{mu_nu_QQ}
\begin{aligned}
&\E [\varepsilon_3]=-\E[\sum_{i,j}\frac{u^2}{2M}\frac{\partial \nu_{n}^{\BY,\BZ}}{\partial Y_{j,i}}[\BZ^H \BQ(z)\BH ]_{j,i}]
=-\frac{ u^2 n}{Mn}[\frac{1}{M} \E[\Tr(\BH\BH^{H} \BQ\BH\BH^{H}\BQ\BZ\BZ^H \BQ(z))]
\\
&
-\frac{1}{M}\E[\Tr(\BH\BH^{H}( \BQ\BH\BH^{H})^2\BQ\BZ\BZ^H \BQ(z))]
+\frac{\sigma^2}{2M}\E[\Tr(\BX\BX^{H}\BH^{H}\BQ^2\BZ\BZ^H \BQ(z)\BH)]
\\
&
-\frac{\sigma^2\E[\Tr(\BH\BH^{H}\BQ^2\BH\BX\BX^{H}\BH^{H}\BQ\BZ\BZ^H \BQ(z) )]}{M}
-\frac{\sigma^2\E[\Tr(\BH\BH^{H}\BQ\BH\BX\BX^{H}\BH^{H}\BQ^2\BZ\BZ^H \BQ(z))] }{M}
\\
&=-\frac{ u^2 n}{Mn}[E_{3,1}+E_{3,2}+E_{3,3}+E_{3,4}+E_{3,5}].
\end{aligned}
\end{equation}
By the matrix inequality $\Tr\BC\BB \le \|\BC \|\Tr\BB$, where $\BB$ is semi-definite positive, we have
\begin{equation}
\begin{aligned}
|\frac{n u^2 E_{3,1}}{Mn} |\le  \frac{n u^2 }{Mn} \E \| \BZ\|^{6}\|\BY \|^4  
\le  \frac{2n u^2 }{Mn} \frac{KN}{M\sigma^{4}z }=\BO(\frac{u^2 }{Mz}).
\end{aligned}
\end{equation}
Similarly, we can show that $E_{3,i},i=2,...,3$ are all $\BO(\frac{1}{Mz})$ terms. We can evaluate $ E_{3,4}$ as
\begin{equation}
\begin{aligned}
\label{E_33_eva}
 \frac{2n u^2 }{Mn} E_{3,4}&= -\frac{2n u^2 }{Mn}\frac{\sigma^2\E[\Tr(\BH\BH^{H}\BQ^2\BH\BX\BX^{H}\BQ\BZ\BZ^H \BQ(z) )]}{Mn}
= \frac{2n u^2 }{Mn}[\frac{\sigma^2\E[\Tr(\BH\BH^{H}\BQ^2\BH\BA\BH^{H}\BQ\BZ\BZ^H \BQ(z))] }{M}
\\
&-\frac{\sigma^2\E[\Tr(\BH\BH^{H}\BQ^2\BH\BH^{H}\BQ\BZ\BZ^H \BQ(z))] }{M}]
=\BO\Biggl(\frac{\mathcal{P}(\frac{1}{z})u^2(1+\frac{\sqrt{\Tr(\BA^2)}}{M^{\frac{5}{2}}})}{zM}\Biggr).
\end{aligned}
\end{equation}
Therefore, $\E [\varepsilon_3]=\BO\Bigl(\frac{u^2\mathcal{P}(\frac{1}{z})(1+\frac{\sqrt{\Tr\BA^2}}{M^{\frac{5}{2}}})}{Mz}\Bigr)=\BO(\frac{u^2\mathcal{P}(\frac{1}{z})}{zM})$.
According to the derivation in~(\ref{mu_nu_QQ}) to~(\ref{E_33_eva}), we can also show that
\begin{equation}
\label{eps_4_eva}
 \E [\varepsilon_4]=\E[\sum_{i,j}\frac{\jmath u^3}{3M}\frac{\partial \beta_{n}^{\BY,\BZ}}{\partial Y_{j,i}}[\BZ^H \BQ(z)\BH ]_{j,i}]=\BO\Biggl(\frac{u^3\mathcal{P}(\frac{1}{z})(1+\frac{\sqrt{\Tr(\BA^2)}}{M^{\frac{5}{2}}})}{M^2z}\Biggr)=\BO\Bigl(\frac{u^3\mathcal{P}(\frac{1}{z})}{zM^2}\Bigr). 
 \end{equation}
According to~(\ref{fir_app}) to~(\ref{eps_4_eva}), we have
\begin{equation}
\begin{aligned}
\E[\Tr(\BQ(z)\BH\BH^{H}\Phi_{n}^{\BY,\BZ})]
&=M\omega^2_{z}\overline{\omega}^2_{z} \E[\Phi_{n}^{\BY,\BZ}]  +   \overline{\omega}^2_{z} \E[  \Tr(\BQ(z)\BZ\BZ^{H})\Phi_{n}^{\BY,\BZ}]
+\frac{\jmath u\overline{\omega}_z \sqrt{n}}{\sqrt{M^3}} \E[ \Tr(\BQ\BZ   \BZ^H \BQ(z)\BH\BH^{H} )\Phi_{n}^{\BY,\BZ}]
\\
&+\BO\Bigl(\frac{\mathcal{P}(\frac{1}{z})}{z}f(\BA,u)
 \Bigr ).
\end{aligned}
 \end{equation}
where $f(\BA,u)$ is given by
\begin{equation}
f(\BA,u)=\Biggl(\frac{1}{M} +\frac{u\sqrt{\Tr(\BA^2)}}{M^\frac{3}{2}}
+\frac{u^2(M+\sqrt{\frac{\Tr(\BA^2)}{M}})}{M^{2}}+\frac{u^3(M+\sqrt{\frac{\Tr(\BA^2)}{M}})}{M^{3}}\Biggr).
\end{equation}

Now we turn to evaluate $\E [\Tr(\BQ(z)\BZ\BZ^{H} )\Phi_{n}^{\BY,\BZ}]$ in the second last line of~(\ref{fir_app}). By the integration by parts formula, we have
\begin{equation}
\label{QZZP1}
\begin{aligned}
&
\E [\Tr(\BQ(z)\BZ\BZ^{H}) \Phi_{n}^{\BY,\BZ}]=\sum_{i,j} \E [Z_{j,i}^{*} [\BQ(z)\BZ]_{j,i} \Phi_{n}^{\BY,\BZ}]
\\
&
=\E \Bigl[[\Tr(\BQ(z)) -\frac{1}{L}\Tr(\BQ(z)\Tr\BQ(z)\BH\BH^{H})
+\frac{\jmath u \sqrt{n}}{\sqrt{M}L} \Tr(\BQ(z)\BH\BH^{H}\BQ)
- \frac{\jmath u n}{\sqrt{M n}L}\Tr (\BQ(z)\BH \BA\BH^{H}\BQ)
\\
&
+\frac{\jmath u  n}{\sqrt{Mn}L}\Tr (\BQ(z)\BH\BH^{H}\BQ\BH \BA\BH^{H}\BQ)
-\frac{u^2\partial \nu_{n}^{\BY,\BZ}}{\partial Z_{j,i}}\frac{[\BQ(z)\BZ ]_{j,i}}{2L}+\frac{\jmath u^3 \partial \beta_{n}^{\BY,\BZ}}{\partial Z_{j,i}}\frac{[\BQ(z)\BZ ]_{j,i}}{3L}
]\Phi_{n}^{\BY,\BZ}\Bigr].
\end{aligned}
\end{equation}
By the analysis of~(\ref{ZZQHHA}), we can show that $\frac{1}{L}\E[\Tr (\BQ(z)\BH \BA\BH^{H}\BQ)]$ and $\frac{1}{L}\E[\Tr( \BQ(z)\BH\BH^{H}\BQ\BH \BA\BH^{H}\BQ)]$ are $\BO(\frac{\sqrt{\Tr(\BA^2)}}{z^2M^{\frac{5}{2}}})$. By following the same procedure from~(\ref{ZZQHHA}) to~(\ref{mu_nu_QQ}), we can show that the last two terms in~(\ref{QZZP1}) are $\BO(\frac{1}{Mz})$ and $\BO(\frac{u^3\mathcal{P}(\frac{1}{z})}{zM^2})$, respectively. Therefore, we have 
\begin{equation}
\label{QZZPhi}
\begin{aligned}
&
\E [\Tr(\BQ(z)\BZ\BZ^{H} )\Phi_{n}^{\BY,\BZ}]
=\E[ \Tr(\BQ(z)) \Phi_{n}^{\BY,\BZ}] (1-\frac{M}{L}\omega_z \overline{\omega}_z )
 -\delta_z\E [\Tr(\BQ(z)\BH\BH^{H}) \Phi_{n}^{\BY,\BZ}]
+M\delta_z \omega_z \overline{\omega}_z \E[\Phi_{n}^{\BY,\BZ}]
\\
&
+\frac{\jmath u \sqrt{n}}{\sqrt{M}L}\E[ \Tr(\BQ(z)\BH\BH^{H}\BQ) \Phi_{n}^{\BY,\BZ}]+\BO(f(\BA,u,z)),
\end{aligned}
\end{equation}
where $f(\BA,u,z)=f(\BA,u)\frac{\mathcal{P}(\frac{1}{z})}{z}$. By substituting~(\ref{fir_app}) into~(\ref{QZZPhi}) to replace $\E \Tr\BQ(z)\BH\BH^{H} \Phi_{n}^{\BY,\BZ}$, we can solve $\E \Tr\BQ(z)\BZ\BZ^{H} \Phi_{n}^{\BY,\BZ}$ to obtain
\begin{equation}
\label{QZZPHIu}
\begin{aligned}
&\E [\Tr(\BQ(z)\BZ\BZ^{H} )\Phi_{n}^{\BY,\BZ}]=\Bigl(\frac{M\delta_z\omega_z\overline{\omega}_z-M\delta_z\omega_z^2\overline{\omega}_z^2}{1+\delta_z\overline{\omega}_z^2} \Bigr)\E [ \Phi_{n}^{\BY,\BZ} ]
+
\frac{1-\frac{M\omega_z\overline{\omega}_z}{L}}{1+\delta_z\overline{\omega}_z^2}\E[ \Tr(\BQ(z)) \Phi_{n}^{\BY,\BZ} ]
\\
&
-\frac{\jmath u \sqrt{n}}{\sqrt{M^3}} \frac{\delta_z\overline{\omega}_z}{1+\delta_z\overline{\omega}_z^2}\E [\Tr(\BQ\BZ   \BZ^H \BQ(z)\BH\BH^{H}) \Phi_{n}^{\BY,\BZ} ]
+\frac{\jmath u \sqrt{n}}{\sqrt{M}L(1+\delta_z\overline{\omega}_z^2)}\E[ \Tr(\BQ(z)\BH\BH^{H}\BQ) \Phi_{n}^{\BY,\BZ} ]
+
\BO(f(\BA,u))
  ).
\end{aligned}
\end{equation}
Noticing that $1-\frac{M\omega_z\overline{\omega}_z}{L}=\frac{M\omega_z}{L\delta_z}$ and plugging~(\ref{QZZPHIu}) into~(\ref{fir_app}), we can obtain
\begin{equation}
\begin{aligned}
\label{Qphi}
\E[\Tr(\BQ(z)) \Phi_{n}^{\BY,\BZ}]& =
\frac{1}{z+\frac{M\omega_z\overline{\omega}_z^2}{L\delta_z(1+\delta_z\overline{\omega}_z^2)}}
\E\Bigl[ [
(N-M\omega_z^2\overline{\omega}_z^2
-\frac{M\delta_z\omega_z\overline{\omega}_z^4}{1+\delta_z\overline{\omega}_z^2} )
-\frac{\jmath u\overline{\omega}_z}{M(1+\delta_z\overline{\omega}_z^2)} \Tr(\BQ\BZ   \BZ^H \BQ(z)\BH\BH^{H})
\\
&-\frac{\jmath u\overline{\omega}_z^2}{L(1+\delta_z\overline{\omega}_z^2)} \Tr(\BQ(z)\BH\BH^{H}\BQ)
] \Phi_{n}^{\BY,\BZ}\Bigr]
+\BO(f(\BA,u,z))
\end{aligned}
\end{equation}

By~(\ref{exp_QQHH}),~(\ref{exp_QzQHH}),~(\ref{Qphi}) and the equation $zL\delta_z +M\omega_z\overline{\omega}_z=N $, $\E\Tr\BQ(z) \Phi_{n}^{\BY,\BZ}$ can be rewritten as 
\begin{align*}
&\E[\Tr(\BQ(z)) \Phi_{n}^{\BY,\BZ}]=L\delta_z \E[\Phi_{n}^{\BY,\BZ}]+ \frac{\jmath u}{z+\frac{M\omega_z\overline{\omega}_z^2}{L\delta_z(1+\delta_z\overline{\omega}_z^2)}}
[-\frac{\overline{\omega}_z}{M(1+\delta_z\overline{\omega}_z^2)}\E[ \Tr(\BQ\BZ   \BZ^H \BQ(z)\BH\BH^{H})]
\\
&
-\frac{\overline{\omega}_z^2}{L(1+\delta_z\overline{\omega}_z^2)} \E[\Tr(\BQ(z)\BH\BH^{H}\BQ)]
]\E[ \Phi_{n}^{\BY,\BZ}]+\BO\Bigl(f(\BA,u,z)+\frac{\mathcal{P}(\frac{1}{z})u}{zM}\Bigr)
\\
&
=[M\overline{C}(z)' +\frac{N}{z}]\E[\Phi_{n}^{\BY,\BZ}] + \frac{\jmath u\sqrt{n} }{\sqrt{M}} [ 
-\frac{\delta\overline{\omega}\overline{\omega}_z'}{1+\delta\overline{\omega}\overline{\omega}_z}
-\frac{M\overline{\omega}\omega_z  \overline{\omega}_z' }{L\Delta_{\sigma^2}(z)\delta_z(1+\delta\overline{\omega}\overline{\omega}_z )^2}
\\
&
-\frac{M \overline{\omega}\overline{\omega}_z }{\Delta_{\sigma^2}(z)L(1+\delta \overline{\omega}\overline{\omega}_z)}\frac{[\delta_z-(\delta_z+\delta_z^2\overline{\omega}_z^2)]\omega_z'}{\delta_z^2}
 ]\Phi_{n}^{\BY,\BZ}+ \BO\Bigl(f(\BA,u,z)+\frac{\mathcal{P}(\frac{1}{z})u}{zM}\Bigr).
 \\
 &=[M\overline{C}(z)' +\frac{N}{z}]\E\Phi_{n}^{\BY,\BZ} +\frac{\jmath u\sqrt{n} }{\sqrt{M}} [- \log(1+\delta\overline{\omega}\overline{\omega}_z)'
 -\frac{M\overline{\omega}\omega_z   }{\Delta_{\sigma^2}(z)L\delta_z}   \left(\frac{\overline{\omega}_z}{1+\delta\overline{\omega}\overline{\omega}_z }\right)' \numberthis \label{QzPHI}
\\
&
-\frac{M \overline{\omega}\overline{\omega}_z }{\Delta_{\sigma^2}(z)L(1+\delta \overline{\omega}\overline{\omega}_z)}(\frac{\omega_z}{\delta_z})'
 ]\Phi_{n}^{\BY,\BZ}+\BO\Bigl(f(\BA,u,z)+\frac{\mathcal{P}(\frac{1}{z})u}{zM}\Bigr)
 \\
 &
 =[M\overline{C}(z)' +\frac{N}{z}]\E[\Phi_{n}^{\BY,\BZ}] +\frac{\jmath u\sqrt{n} }{\sqrt{M}} [ -\log(1+\delta\overline{\omega}\overline{\omega}_z)'-[\log(\Delta_{\sigma^2}(z))]']\Phi_{n}^{\BY,\BZ}
+ \BO\Bigl(f(\BA,u,z)+\frac{\mathcal{P}(\frac{1}{z})u}{zM}\Bigr).
\end{align*}
The term $\overline{C}'(z)$ in the last line of~(\ref{QzPHI}) is given in~(\ref{delta_deri}). According to~(\ref{delta_def}), we can obtain that $\delta_z \xrightarrow[]{z\rightarrow \infty} 0$, $\omega_z \xrightarrow[]{z\rightarrow \infty} 0$, and $\overline{\omega}_z \xrightarrow[]{z\rightarrow \infty} 1$. Therefore, we have
\begin{equation}
\begin{aligned}
\overline{C}(\infty)&=0,
\\
\Delta_{\sigma^2}(\infty)&=\sigma^2+ \frac{M\omega\overline{\omega}}{L}=\frac{N}{L\delta}.
\end{aligned}
\end{equation}
By~(\ref{U1_step1}) and~(\ref{QzPHI}), we have
\begin{equation}
\label{U11_eva1}
\begin{aligned}
U_{1,1} &=\jmath \sqrt{\frac{n}{M}}
\int_{\sigma^2}^{\infty} \frac{N\E[\Phi_{n}^{\BY,\BZ}]}{z}-\E[\Tr(\BQ(z))\Phi_{n}^{\BY,\BZ}]\mathrm{d}z
\\
&
=\jmath \sqrt{nM}\overline{C}(\sigma^2)\E[\Phi_{n}^{\BY,\BZ}]
+\frac{nu}{M}[
\log(1+\delta\overline{\omega}\overline{\omega}_z)|_{\sigma^2}^{\infty}
+\log(\Delta_{\sigma^2}(z))|_{\sigma^2}^{\infty}]\E[\Phi_{n}^{\BY,\BZ}]
\\
&
=\jmath \sqrt{nM}\overline{C}(\sigma^2)\E[\Phi_{n}^{\BY,\BZ}]-\frac{un}{M}[-\log(\frac{1+\delta\overline{\omega}^2}{1+\delta\overline{\omega}})
-\log(\Delta_{\sigma^2})+\log(\frac{N}{L\delta})]\E[\Phi_{n}^{\BY,\BZ}]+\BO(f(\BA,u,\sigma^2)+\frac{\mathcal{P}(\frac{1}{\sigma^2})u}{\sigma^2 M})
\\
&=\jmath \sqrt{nM}\overline{C}(\sigma^2)\E[\Phi_{n}^{\BY,\BZ}]-\frac{un}{M}(-\log(\Xi))\E[\Phi_{n}^{\BY,\BZ}]+
\BO(f(\BA,u,\sigma^2)+\frac{u}{ M})
.
\end{aligned}
\end{equation}
Now we turn to evaluate $U_{1,2}$ in~(\ref{U1_step1}). By the integration by parts formula, we can obtain
\begin{equation}
\begin{aligned}
\label{U12app}
&\E[\Tr(\BQ\BH\BA\BH^{H})\Phi_{n}^{\BY,\BZ}]
=\sum_{i,j}\E[[\BY]_{j,i}^{*}[\BZ^{H}\BQ\BH\BA]_{j,i}\Phi_{n}^{\BY,\BZ}]
=\E\Bigl[ [\frac{1}{M}\Tr(\BQ\BZ\BZ^{H})\Tr(\BA)
-\frac{1}{M}\Tr(\BQ\BZ\BZ^{H})\Tr(\BQ\BH\BA\BH^{H})
\\
&
+\frac{\jmath u \sqrt{n}}{M^{\frac{3}{2}}}\Tr(\BQ\BZ\BZ^{H}\BQ\BH\BA\BH^{H})
-\frac{\jmath u n}{\sqrt{Mn}M}\Tr(\BQ\BZ\BZ^{H}\BQ\BH\BA\BH^{H}\BQ\BH\BA\BH^{H})
+\frac{\jmath u n}{\sqrt{Mn}}\frac{\Tr(\BZ\BZ^{H}\BQ\BH\BA^2\BH^{H}\BQ)}{M}]\Phi_{n}^{\BY,\BZ}\Bigr]
\\
&
-\E\Bigl[\frac{u^2 \partial \nu_{n}^{\BY,\BZ}}{\partial Y_{j,i}}\frac{[\BZ^{H}\BQ\BH\BA]_{j,i}\Phi_{n}^{\BY,\BZ}}{2M}+\E\frac{\jmath u^3 \partial \beta_{n}^{\BY,\BZ}}{\partial Y_{j,i}}\frac{[\BZ^{H}\BQ\BH\BA]_{j,i}\Phi_{n}^{\BY,\BZ}}{3M}\Bigr].
\end{aligned}
\end{equation}
According to~Lemma~\ref{var_con} and $\Tr\BA^2=\BO(M^2)$, we can obtain
\begin{equation}
\label{var_QHHA}
\Var(\frac{1}{M}\Tr(\BQ\BH\BA\BH^{H}))=\BO(\frac{\Tr(\BA^2)}{M^3})=\BO(\frac{1}{M}),
\end{equation}
and 
\begin{equation}
\label{var_QZZQHHA}
\Var(\frac{1}{M}\Tr(\BQ\BZ\BZ^{H}\BQ\BH\BA\BH^{H}))=\BO(\frac{\Tr(\BA^2)}{M^3})=\BO(\frac{1}{M}),
\end{equation}
to derive 
\begin{equation}
\label{exp_QHHA}
\begin{aligned}
&\frac{1}{M}\E[\Tr(\BQ\BH\BA\BH^{H})]=\frac{1}{M^2}\E[\Tr(\BQ\BZ\BZ^{H})]\Tr(\BA)
-\frac{1}{M^2}\E[\Tr(\BQ\BZ\BZ^{H})\Tr(\BQ\BH\BA\BH^{H})]
\\
&=-\frac{\cov(\frac{1}{M}\Tr(\BQ\BH\BA\BH^{H}),\frac{1}{M}\Tr(\BQ\BZ\BZ^{H}))}{1+\frac{\E\Tr\BQ\BZ\BZ^{H}}{M}}
=\BO(\frac{\sqrt{\Tr(\BA^2)}}{M^{\frac{5}{2}}})=\BO(\frac{1}{M^{\frac{3}{2}}}).
\end{aligned}
\end{equation}
Similarly, we have
\begin{equation}
\frac{1}{M}\E[\Tr(\BQ\BZ\BZ^{H}\BQ\BH\BA\BH^{H})]=\BO(\frac{\sqrt{\Tr(\BA^2)}}{M^{\frac{5}{2}}}),
\end{equation}
such that the third term in~(\ref{U12app}) is $\BO(\frac{\sqrt{\Tr\BA^2}}{M^{\frac{3}{2}}})$. Therefore, we can obtain
\begin{equation}
|\frac{1}{M}\E[\underline{a}{\Phi_{n}^{\BY,\BZ}}\underline{b}]|
\le\frac{1}{M} \Var^{\frac{1}{2}}(a)\Var^{\frac{1}{2}}(b)=\BO(\frac{\sqrt{\Tr(\BA^2)}}{M^{\frac{3}{2}}}),
\end{equation}
where $a=\Tr(\BQ\BH\BA\BH^{H})$ and $b=\Tr(\BQ\BZ\BZ^{H})$. Considering the technique used in~(\ref{cov_expand}), we can evaluate the second term in~(\ref{U12app}) as
\begin{equation}
\label{second_term_U12}
\begin{aligned}
&-\frac{1}{M}\E[\Tr(\BQ\BZ\BZ^{H})\Tr(\BQ\BH\BA\BH^{H})\Phi_{n}^{\BY,\BZ}]
=-\frac{\E[\Tr(\BQ\BH\BA\BH^{H})]}{M}\E[\underline{\Tr(\BQ\BZ\BZ^{H})}\Phi_{n}^{\BY,\BZ}]
-\omega \E[\Tr(\BQ\BH\BA\BH^{H})\Phi_{n}^{\BY,\BZ}]+\BO(\frac{\sqrt{\Tr(\BA^2)}}{M^{\frac{3}{2}}}),
\end{aligned}
\end{equation}
where $\frac{\E[\Tr(\BQ\BH\BA\BH^{H})]}{M}\E[\underline{\Tr(\BQ\BZ\BZ^{H})}\Phi_{n}^{\BY,\BZ}]=\BO(\frac{1}{M^{\frac{3}{2}}})$. By Poincar{\`e}-Nash inequality, we have
\begin{equation}
\begin{aligned}
\Biggl|\E\Biggl[\frac{u^2\partial \nu_{n}^{\BY,\BZ}}{\partial Y_{j,i}}\frac{[\BZ^{H}\BQ\BH\BA\Phi_{n}^{\BY,\BZ}]_{j,i}}{M}\Biggr]\Biggr|
=\BO\Bigl(\frac{u^2}{M}(\frac{\Tr(\BA^2)}{M}+1)\Bigr),~\Var^{\frac{1}{2}}(\frac{u^2\partial \nu_{n}^{\BY,\BZ}}{\partial Y_{j,i}}\frac{[\BZ^{H}\BQ\BH\BA\Phi_{n}^{\BY,\BZ}]_{j,i}}{M})=\BO\Bigl(\frac{u^2(\frac{\Tr(\BA^2)}{M^{\frac{1}{2}}}+1)}{M^{2}}\Bigr),
\\
\Biggl|\E\Biggl[ \frac{u^3\partial \beta_{n}^{\BY,\BZ}}{\partial Y_{j,i}}\frac{[\BZ^{H}\BQ\BH\BA\Phi_{n}^{\BY,\BZ}]_{j,i}}{M}\Biggr]\Biggr|
=\BO\Bigl(\frac{u^3}{M^2}(\frac{\Tr(\BA^2)}{M}+1)\Bigr),~\Var^{\frac{1}{2}}(\frac{u^3\partial \beta_{n}^{\BY,\BZ}}{\partial Y_{j,i}}\frac{[\BZ^{H}\BQ\BH\BA\Phi_{n}^{\BY,\BZ}]_{j,i}}{M})=\BO\Bigl(\frac{u^3(\frac{\Tr(\BA^2)}{M^{\frac{1}{2}}}+1)}{M^{3}}\Bigr),
\end{aligned}
\end{equation}
By~(\ref{var_QZZ}),~(\ref{var_QHCH}) in Lemma~\ref{var_con} and $\Tr\BA^4\le (\Tr\BA^2)^2 $, we can bound the variance
\begin{equation}
\label{AA_con_var}
\begin{aligned}
&\Var\Bigl(\frac{\Tr\BQ\BZ\BZ^{H}\BQ\BH\BA\BH^{H}\BQ\BH\BA\BH^{H}}{M}\Bigr)
=\BO\Bigl(\frac{\Tr(\BA^{4})}{M^3}\Bigr)=\BO\Bigl(\frac{(\Tr(\BA^{2}))^2}{M^3}\Bigr),
\\
&\Var\Bigl(\frac{\Tr(\BZ\BZ^{H}\BQ\BH\BA^2\BH^{H}\BQ)}{M}\Bigr)
=\BO\Bigl(\frac{\Tr(\BA^{4})}{M^3}\Bigr)=\BO\Bigl(\frac{(\Tr(\BA^{2}))^2}{M^3}\Bigr).
\end{aligned}
\end{equation}
Therefore, by noticing $\Tr(\BA)=0$, plugging~(\ref{second_term_U12}) into~(\ref{U12app}), moving $-\omega \E[\Tr(\BQ\BH\BA\BH^{H})\Phi_{n}^{\BY,\BZ}]$ to the LHS of~(\ref{U12app}), and solving $\E[\Tr(\BQ\BH\BA\BH^{H})\Phi_{n}^{\BY,\BZ}]$,~(\ref{U12app}) can be rewritten as
\begin{equation}
\label{QHAHPhiYZ}
\begin{aligned}
\E[\Tr(\BQ\BH\BA\BH^{H}\Phi_{n}^{\BY,\BZ})]
&=
-\frac{\jmath u \overline{\omega}n}{\sqrt{Mn}}\frac{\E[\Tr(\BQ\BZ\BZ^{H}\BQ\BH\BA\BH^{H}\BQ\BH\BA\BH^{H})]}{M}\E[\Phi_{n}^{\BY,\BZ}]
+\frac{\jmath u \overline{\omega}n}{\sqrt{Mn}}\frac{\E[\Tr(\BZ\BZ^{H}\BQ\BH\BA^2\BH^{H}\BQ)]}{M}\E[\Phi_{n}^{\BY,\BZ}]
\\
&
=-\frac{\jmath u \overline{\omega}n}{\sqrt{Mn}}(X_1-X_2)\E[\Phi_{n}^{\BY,\BZ}]
+\BO\Bigl(\frac{\sqrt{\Tr(\BA^2)}}{M^{\frac{3}{2}}}+   \frac{u\Tr(\BA^2)}{M^{\frac{3}{2}}} + \frac{u^2(1+ \frac{\Tr(\BA^2)}{M} )}{M}+\frac{u^3(1+ \frac{\Tr(\BA^2)}{M} )}{M^2}\Bigr).
\end{aligned}
\end{equation}
Now we will compute $X_1-X_2$. By the integration by parts formula and the variance control in~(\ref{AA_con_var}), we have
\begin{equation}
\label{X1_exp1}
\begin{aligned}
&X_1=
\frac{1}{M}\sum_{i,j} \E[ Y^{*}_{j,i}[\BZ^{H}\BQ\BH\BA\BH^{H}\BQ\BZ\BZ^{H}\BQ\BH\BA]_{j,i}]
\\
&=\E\Bigl[ \frac{\Tr(\BQ\BZ\BZ^{H})}{M}\frac{\Tr(\BQ\BZ\BZ^{H}\BQ\BH\BA^2\BH^{H})}{M}
+\frac{\Tr(\BA)}{M}\frac{\Tr(\BQ\BZ\BZ^{H}\BQ\BH\BA\BH^{H}\BQ\BZ\BZ^{H})}{M}
\\
&
-\frac{\Tr(\BQ\BZ\BZ^{H})}{M}\frac{\Tr(\BQ\BH\BA\BH^{H}\BQ\BZ\BZ^{H}\BQ\BH\BA\BH^{H})}{M}
-\frac{\Tr(\BQ\BZ\BZ^{H}\BQ\BH\BA\BH^{H})}{M}\frac{\Tr(\BQ\BZ\BZ^{H}\BQ\BH\BA\BH^{H})}{M}
\\
&
-\frac{\Tr(\BQ\BZ\BZ^{H}\BQ\BZ\BZ^{H}\BQ\BH\BA\BH^{H})}{M}\frac{\Tr(\BQ\BH\BA\BH^{H})}{M}
\Bigr]
\\
&
=A_{1}+A_{2}+A_{3}+A_{4}+A_{5}.
\end{aligned}
\end{equation}
By the variance control in~(\ref{var_QHCH}) ($m=0$) and~(\ref{AA_con_var}), we have
\begin{equation}
\label{eva_A1}
A_1=\frac{\omega \E[\Tr(\BQ\BZ\BZ^{H}\BQ\BH\BA^2\BH^{H})]}{M}+\BO\Bigl(\frac{\Tr(\BA^2)}{M^{\frac{5}{2}}}\Bigr).
\end{equation}
By noticing that $\Tr(\BA)=0$, we can obtain that $A_{2}=0$. $A_{3}$ can be evaluated as
\begin{equation}
A_3=-\omega X_1
+\BO\Bigl(\frac{\Tr(\BA^2)}{M^{\frac{5}{2}}}\Bigr),
\end{equation}
according to the bound of the variance in~(\ref{AA_con_var}) and the evaluation of~(\ref{appro_del}) in Lemma~\ref{appro_lem}. By $\E a^2=(\E a)^2+\Var(a) $ and~(\ref{var_QZZQHHA}), $A_{4}$ can be evaluated as
\begin{equation}
A_{4}=-(\E [a])^2+\BO(\frac{\Tr(\BA^2)}{M^{3}})=\BO\Bigl(\frac{\Tr(\BA^2)}{M^{3}}\Bigr),
\end{equation}
where $a=\frac{\Tr(\BQ\BZ\BZ^{H}\BQ\BH\BA\BH^{H})}{M}$ and the last equality follows from $\E [a] =\BO(\frac{\Tr(\BA^2)}{M^{5}})$, which can be derived by a similar approach as~(\ref{exp_QHHA}). Define $c=\frac{\Tr(\BQ\BZ\BZ^{H}\BQ\BZ\BZ^{H}\BQ\BH\BA\BH^{H})}{M}$ and $d=\frac{\Tr(\BQ\BH\BA\BH^{H})}{M}$. By~(\ref{var_QZZQHBH}) in Lemma~\ref{var_con}, we have $\Var(c)=\BO(\frac{\Tr(\BA^2)}{M^3})=\BO(\frac{1}{M})$ and $\Var(d)=\BO(\frac{\Tr(\BA^2)}{M^3})=\BO(\frac{1}{M})$. Therefore, $A_{5}$ can be handled by
\begin{equation}
A_{5}=-\E[ c] \E[ d ]+ \cov(c,d)
=   \BO(\frac{1}{M^3})+\BO\Bigl(\frac{\Tr(\BA^2)}{M^3}\Bigr),
\end{equation}
where $\E [c]$ and $\E [d]$ can be evaluated by the same approach as~(\ref{exp_QHHA}). By moving $\omega X_1$ from $A_3$ to the LHS of~(\ref{X1_exp1}) and using the evaluation of $A_1$ in~(\ref{eva_A1}), we can solve $X_{1}$ as
\begin{equation}
\begin{aligned}
\label{expX1QHAAH}
X_1
=\overline{\omega}A_1 +\BO\Bigl(\frac{\Tr(\BA^2)}{M^{\frac{5}{2}}}\Bigr)
 =\frac{\omega\overline{\omega}\E\Tr(\BQ\BZ\BZ^{H}\BQ\BH\BA^2\BH^{H})}{M}+\BO\Bigl(\frac{\Tr(\BA^2)}{M^{\frac{5}{2}}}\Bigr).
\end{aligned}
\end{equation}
Then we will evaluate $X_{2}$. By the integration by parts formula~(\ref{int_part}), we have
\begin{equation}
\label{QZZHHAA}
\begin{aligned}
X_{2}
&=\frac{1}{M} \sum_{i,j}\E[Y^{*}_{j,i}[\BZ^{H}\BQ\BZ\BZ^{H}\BQ\BH\BA^2]_{j,i}]
=\E \Bigl[ \frac{\Tr(\BA^2)}{M}\frac{\Tr(\BQ\BZ\BZ^{H}\BQ\BZ\BZ^{H})}{M}
\\
&
-\frac{\Tr(\BQ\BZ\BZ^{H})}{M}
\frac{\Tr(\BQ\BZ\BZ^{H}\BQ\BH\BA^2\BH^{H})}{M}
-\frac{\Tr(\BQ\BZ\BZ^{H}\BQ\BZ\BZ^{H})}{M}\frac{\Tr(\BQ\BH\BA^2\BH^{H})}{M}\Bigr]
\\
&\overset{(a)}{=}
(1-\omega\overline{\omega}) \frac{\Tr(\BA^2)}{M}\frac{\E[\Tr(\BQ\BZ\BZ^{H}\BQ\BZ\BZ^{H})]}{M}
-\frac{\omega\E[\Tr(\BQ\BZ\BZ^{H}\BQ\BH\BA^2\BH^{H})]}{M}+\BO\Bigl(\frac{\Tr(\BA^2)}{M^{\frac{5}{2}}}\Bigr)
\\
&{=}\frac{\overline{\omega}^2\Tr(\BA^2)\E[\Tr(\BQ\BZ\BZ^{H}\BQ\BZ\BZ^{H})]}{M^2}+\BO\Bigl(\frac{\Tr(\BA^2)}{M^{\frac{5}{2}}}\Bigr),
\end{aligned}
\end{equation}
where step $(a)$ in~(\ref{QZZHHAA}) is obtained by the variance control in~(\ref{AA_con_var}) and the evaluations in~(\ref{appro_del}). Based on~(\ref{expX1QHAAH}) and~(\ref{QZZHHAA}), we have
\begin{equation}
\label{U12_fir_eva}
\begin{aligned}
&-\frac{n}{\sqrt{Mn}}(\overline{\omega}X_1-\overline{\omega}X_2)
\\
&
=\frac{n^2\overline{\omega}^4\Tr(\BA^2)}{M^2n}\frac{\E[\Tr(\BQ\BZ\BZ^{H}\BQ\BZ\BZ^{H})]}{M}+\BO(\frac{\Tr(\BA^2)}{M^{\frac{5}{2}}})
\\
&\overset{(a)}{=}\frac{n^2\overline{\omega}^4\Tr(\BA^2)}{M^2 n}\Bigl(\frac{1+\delta\overline{\omega}}{1+\delta\overline{\omega}^2}\frac{M\omega^2}{L}+\frac{M\omega\E[\Tr(\BQ^2\BZ\BZ^{H})]}{L\delta(1+\delta \overline{\omega}^2)M}\Bigr)+\BO\Bigl(\frac{\Tr(\BA^2)}{M^{\frac{5}{2}}}\Bigr)
\\
&=\frac{n^2\overline{\omega}^4\Tr(\BA^2)}{M^2 n}\Bigl(\frac{1+\delta\overline{\omega}}{1+\delta\overline{\omega}^2}\frac{M\omega^2}{L}-\frac{M\omega\omega'}{L\delta(1+\delta \overline{\omega}^2)}\Bigr)+\BO\Bigl(\frac{\Tr(\BA^2)}{M^{\frac{5}{2}}}\Bigr),
\end{aligned}
\end{equation}
where step $(a)$ in~(\ref{U12_fir_eva}) follows from~(\ref{QZZQZZ_exp}). By~(\ref{QHAHPhiYZ}), we can obtain that 
\begin{equation}
\label{U12evaf}
\begin{aligned}
&U_{1,2}=\frac{\rho\Tr(\BA^2)\overline{\omega}^4(\frac{1+\delta\overline{\omega}}{1+\delta\overline{\omega}^2}\frac{M\omega^2}{L}-\frac{\omega\omega'}{\delta(1+\delta \overline{\omega}^2)})}{M}\E [\Phi_{n}^{\BY,\BZ}]
\\
&+\BO\Bigl(\frac{\sqrt{\Tr(\BA^2})}{M^{\frac{3}{2}}}+   \frac{u\Tr(\BA^2)}{M^{\frac{3}{2}}} + \frac{u^2(1+ \frac{\Tr(\BA^2)}{M} )}{M}+\frac{u^3(1+ \frac{\Tr(\BA^2)}{M} )}{M^2}\Bigr).
\end{aligned}
\end{equation}
By substituting~(\ref{U11_eva1}) and~(\ref{U12_fir_eva}) into~(\ref{U1_step1}), we complete the evaluation of $U_1$ in~(\ref{exp_U1U2}).  

\subsubsection{The evaluation of $U_2$} By the resolvent identity, we have
\begin{equation}
\begin{aligned}
\label{U2eva}
&\E [\nu^{\BY,\BZ}_{n}\Phi_{n}^{\BY,\BZ}]
=\E\Bigl[ [ \frac{1}{M}\Tr(\BQ\BH\BH^{H}\BQ\BH\BH^{H})
+\frac{2\sigma^2n}{Mn}\Tr\Bigl(\BQ^2\BH\frac{\BX\BX^{H}}{n}\BH^{H}\Bigr)]\Phi_{n}^{\BY,\BZ} \Bigr]
\\
&\overset{(a)}{=}\E\Bigl[ [ \frac{N}{M} - \frac{\sigma^2\Tr(\BQ)}{M}-\frac{\sigma^2\Tr(\BQ^2\BH\BH^{H})}{M} 
+\frac{2\sigma^2n}{Mn}\Tr\Bigl(\BQ^2\BH\frac{\BX\BX^{H}}{n}\BH^{H}\Bigr)]\Phi_{n}^{\BY,\BZ}\Bigr]
\\
&=\E \Bigl[ \frac{N}{M} - \frac{\sigma^4}{M}\Tr(\BQ^2) -\frac{2\sigma^2}{M}\Tr(\BQ^2\BH\BA\BH^{H})\Bigr]\E[\Phi_{n}^{\BY,\BZ}]
= ( \frac{N}{M}+\frac{\sigma^4\delta'}{\kappa} )\E [\Phi_{n}^{\BY,\BZ}] + \BO\Bigl(\frac{\sqrt{\Tr(\BA^2)}}{M^{\frac{3}{2}}}\Bigr),
\end{aligned}
\end{equation}
from which we can obtain the evaluation of $U_{2}$. By far, we have completed the evaluation of~(\ref{exp_U1U2}). By plugging~(\ref{U11_eva1}),~(\ref{U12evaf}), and~(\ref{U2eva}) into~(\ref{exp_U1U2}), we can obtain
\begin{equation}
\label{phiYZeva}
\E [(\jmath\mu^{\BY,\BZ}_{n}-u\nu^{\BY,\BZ}_{n} ) \Phi_{n}^{\BY,\BZ}]=(\jmath \sqrt{nM}\overline{C}(\sigma^2) -u V_n ) \E[\Phi_{n}^{\BY,\BZ}]+\varepsilon(\BA,u),
\end{equation}
where 
\begin{equation}
\varepsilon(\BA,u)=\BO\Bigl(f(\BA,u)+\frac{u}{ M}+\frac{(1+u)\sqrt{\Tr(\BA^2)}}{M^{\frac{3}{2}}}+   \frac{u\Tr(\BA^2)}{M^{\frac{3}{2}}} + \frac{u^2(1+ \frac{\Tr(\BA^2)}{M} )}{M}+\frac{u^3(1+ \frac{\Tr(\BA^2)}{M} )}{M^2}\Bigr)
\end{equation}

\subsection{Step 3: Convergence of $\E \Bigl[ e^{\jmath\frac{\gamma_{n}^{\BW,\BY,\BZ}-\overline{\gamma}_{n}}{\sqrt{V_n}}}\Bigr]$}
 By the evaluations in Appendix~\ref{appro_charaf}, we can obtain the evaluation for in~(\ref{diff_eq1})-(\ref{diff_eq4}) given by
\begin{align}
&\frac{\partial  \Psi^{\BW,\BY,\BZ}(u)}{\partial u}= \E [(\jmath \mu^{\BY,\BZ}_{n}-u\nu^{\BY,\BZ}_{n} )\Phi_{n}^{\BY,\BZ}] + \BO\Bigl(\frac{u^2}{M}+\frac{u^3\mathcal{M}(u)}{M^2}+\frac{u^4 \mathcal{M}(u)}{M^3}\Bigr) \label{diff_eq1}
\\
&
=(\jmath \sqrt{nM}\overline{C}(\sigma^2) -u V_n )\E  [\Phi_{n}^{\BY,\BZ}] + \BO\Bigl(\frac{u^2}{M}+\frac{u^3\mathcal{M}(u)}{M^2}+\frac{u^4 \mathcal{M}(u)}{M^3}\Bigr)+\varepsilon(\BA,u)\label{diff_eq2}
\\
&=(\jmath \sqrt{nM}\overline{C}(\sigma^2) -u V_n )\Psi^{\BW,\BY,\BZ}(u)+ \BO\Bigl(\frac{u\mathcal{M}(u)}{M}+\frac{u V_n\mathcal{M}(u)}{M^2}+ \frac{u^2}{M}+\frac{u^3\mathcal{M}(u)}{M^2}+\frac{u^4 \mathcal{M}(u)}{M^3}\Bigr)+\varepsilon(\BA,u)
 \label{diff_eq3}
\\
&
=(\jmath \sqrt{nM}\overline{C}(\sigma^2) -u V_n )\Psi^{\BW,\BY,\BZ}(u)+\overline{\varepsilon}(\BA,u),\label{diff_eq4}
\end{align}
where~(\ref{diff_eq1}) and~(\ref{diff_eq2}) follow from~(\ref{chara_lemma}) and~(\ref{phiYZeva}), respectively.~(\ref{diff_eq3}) is obtained by~(\ref{eq_cha_cha}). Solving $\Psi^{\BW,\BY,\BZ}(u)$ from the differential equation in~(\ref{diff_eq4}), we can obtain
\begin{equation}
\label{pphi_eps}
\begin{aligned}
\Psi^{\BW,\BY,\BZ}(u) &=e^{\jmath u \sqrt{nM}\times \overline{C}(\sigma^2) -\frac{u^2 V_n}{2}  }
(1+\int_{0}^{u}e^{-\jmath v \sqrt{nM}\times \overline{C}(\sigma^2) +\frac{v^2 V_n}{2}  }\overline{\varepsilon}(\BA,v)\mathrm{d}v)
\\
&
=e^{\jmath u \sqrt{nM}\times \overline{C}(\sigma^2) -\frac{u^2 V_n}{2}  }+E(u,\BA).
\end{aligned}
\end{equation}
According to the following evaluation
\begin{equation}
\begin{aligned}
e^{-\frac{A u^2}{2}}\int_{0}^{u}\mathcal{M}(B,x) x^{\alpha} e^{\frac{A x^2}{2}}\mathrm{d} x  \le u^{\alpha-1} \mathcal{M}(B,u) e^{-\frac{A u^2}{2}}\int_{0}^{u}  x e^{\frac{A x^2}{2}} \mathrm{d} x
\\
=u^{\alpha-1} A^{-1}\mathcal{M}(B,u) (1-e^{-\frac{A u^2}{2}})=\BO( u^{\alpha-1}\mathcal{M}(B,u)\mathcal{M}(A^{-1},u) ),
\end{aligned}
\end{equation}
the error term $E(u,\BA)$ can be bounded by
\begin{equation}
\label{final_error_term}
\begin{aligned}
E(u,\BA)=\BO\Bigl(\frac{\mathcal{M}(V_n,u)}{M}+\frac{ V_n\mathcal{M}(V_n,u)}{M^2}+ \frac{u\mathcal{M}(V_n,u)}{M}+\frac{u^2 \mathcal{M}(V_n,u)}{M^2}+\frac{u^3\mathcal{M}(V_n,u)}{M^3}
\\
+
\frac{u}{M} +\frac{u\sqrt{\Tr(\BA^2)}}{M^\frac{3}{2}}
+\frac{u \mathcal{M}(V_n,u) \Bigl(M+\sqrt{\frac{\Tr(\BA^2)}{M}}\Bigr)}{M^{2}}+\frac{u^2 \mathcal{M}(V_n,u)\Bigl(M+\sqrt{\frac{\Tr(\BA^2)}{M}}\Bigr)}{M^{3}}
\\
+\frac{\mathcal{M}(V_n,u)}{ M}+\frac{u\sqrt{\Tr(\BA^2)}}{M^{\frac{3}{2}}} +\Bigl[\frac{\Tr(\BA^2)}{M^{\frac{3}{2}}} + \frac{u(1+\frac{\Tr(\BA^2)}{M})}{M}+\frac{u^2(1+\frac{\Tr(\BA^2)}{M})}{M^2}\Bigr]\mathcal{M}(V_n,u)\Bigr).
\end{aligned}
\end{equation}
Denoting $\overline{\gamma}_{n}= \sqrt{nM}\overline{C}(\sigma^2) $, the characteristic function of the normalized version of $\gamma_{n}$ can be written as
\begin{equation}
\label{conver_cha}
\begin{aligned}
\Psi^{\BW,\BY,\BZ}_{norm}(u) &=\E e^{\jmath u \frac{\gamma_{n}-\overline{\gamma}_n}{\sqrt{V_n}}}
=\Psi^{\BW,\BY,\BZ}\Bigl(\frac{u}{\sqrt{V_n}}\Bigr)e^{-\jmath u\frac{\overline{\gamma}_{n}}{\sqrt{V_n}}}
+E\Bigl(\frac{u}{\sqrt{V_n}},\BA\Bigr)
\overset{(a)}{=}e^{-\frac{u^2}{2}}+\BO\Bigl(\frac{1}{\sqrt{M}}\Bigr),
\end{aligned}
\end{equation} 
where step $(a)$ in~(\ref{conver_cha}) follows from
\begin{equation}
\begin{aligned}
 \BO\Bigl(E(\frac{u}{\sqrt{V_n}},\BA)\Bigr) =\BO\Biggl(\frac{\mathcal{M}(u)\frac{\Tr(\BA^2)}{M}}{\sqrt{M}V_n}\Biggr)
=\BO\Biggl(\frac{\mathcal{M}(u)\frac{\Tr(\BA^2)}{M}}{\sqrt{M}(\frac{\tau\Tr(\BA^2)}{M}+K)}\Biggl)=\BO\Bigl(\frac{\mathcal{M}(u)}{\sqrt{M}}\Bigr),
\end{aligned}
\end{equation} 
and $\tau=\BO(1)$ is the coefficient of $\frac{\Tr(\BA^2)}{M}$ in~(\ref{clt_var}) and $K$ does not depend on $M$, $N$, $L$, and $n$. By~(\ref{conver_cha}) and L{\'e}vy’s continuity theorem~\cite{billingsley2017probability}, there holds true that
\begin{equation}
\begin{aligned}
\frac{\gamma_n-\overline{\gamma}_n}{\sqrt{V}_n} \xrightarrow[{N  \xrightarrow[]{\rho,\eta, \kappa}\infty}]{\mathcal{D}} \mathcal{N}(0,1).
\end{aligned}
\end{equation}

\subsection{Proof of~(\ref{prob_con_rate})}
To establish the convergence rate for the CDF of the MID, we will use Esseen inequality~\cite[p538]{feller1991introduction}, which says that
 there exists $C>0$ for any $T>0$ such that 
 \begin{equation}
 \label{berry_ineq}
\begin{aligned}
&\sup\limits_{x\in \mathbb{R}}|\Prob\left(\sqrt{\frac{{Mn}}{V_n}}(I_{N,L,M}^{(n)}-\overline{C}(\sigma^2)\right) \le x  )|
\le C \int_{0}^{T} u^{-1} | \Psi^{\BW,\BY,\BZ}_{norm}(u) - e^{-\frac{u^2}{2}}| \mathrm{d} u+ \frac{C}{T} 
\\
& \le K (\int_{0}^{T} u^{-1}  E\left(\frac{u}{\sqrt{V_n}},\BA) +\frac{1}{T}\right).
\end{aligned}
\end{equation}
Notice that the dominating term in~(\ref{final_error_term}) is $\BO(\frac{\mathcal{M}(u)+u}{\sqrt{M}})$ and for $T>1$,
\begin{equation}
\begin{aligned}
& \int_{0}^{T}u^{-1} (\mathcal{M}(1,u)+u) \mathrm{d} u
=\int_{0}^{1}(1+u) \mathrm{d} u
+\int_{1}^{T}u^{-1} (1+u) \mathrm{d} u = \BO(\log(T)+T)=\BO(T).
\end{aligned}
\end{equation}
By taking $T=n^{\frac{1}{4}}$ in~(\ref{berry_ineq}), we can obtain
 \begin{equation}
\sup\limits_{x\in \mathbb{R}}|\Prob\left(\sqrt{\frac{{Mn}}{V_n}}(I_{N,L,M}^{(n)}-\overline{C}(\sigma^2)) \le x  \right)|=\BO(n^{-\frac{1}{4}}),
\end{equation}
which concludes the proof of Theorem~\ref{clt_the}.

\section{Proof of Theorem~\ref{the_oaep}}
\label{bound_proof}
\begin{proof} According to Lemma~\ref{bnd_err}, the upper and lower bounds can be obtained by investigating the distribution of $I_{N,L,M}^{(n)}$ and $I_{N,L,M}^{(n+1)}$, respectively, which are analyzed as follows.

\textit{Lower bound}: By Theorem~\ref{clt_the}, we can obtain that 
\begin{equation}
\begin{aligned}
\Prob\Biggl(\frac{\sqrt{Mn}(I_{N,L,M}^{(n+1)}-\overline{C}(\sigma^2))}{\sqrt{\hat{V}_{n+1}}}\le z\Biggr)\xrightarrow[]{n  \xrightarrow[]{\rho,\eta, \kappa}\infty} \Phi(z),
\end{aligned}
\end{equation}
where $\hat{V}_{n+1}=\frac{n}{n+1}(V_{-}-\frac{\log(\Xi)}{M} + \frac{n+1}{M}\frac{\Theta\Tr(\BA_{n+1}^2)}{M})$, with $\Theta=\kappa\overline{\omega}^4
(\frac{\omega^2(1+\delta\overline{\omega})}{1+\delta\overline{\omega}^2}-\frac{\omega\omega'}{\delta(1+\delta \overline{\omega}^2)})$. 
Since $-\frac{\log(\Xi)}{M}  \xrightarrow[]{M  \xrightarrow[]{\rho,\eta, \kappa}\infty}  0$, by Slutsky's lemma~\cite{dasgupta2008asymptotic}, we have
\begin{equation}
\begin{aligned}
\Prob\Biggl(\frac{\sqrt{Mn}(I_{N,L,M}^{(n+1)}-\overline{C}(\sigma^2))}{\sqrt{{V}_{n,+}}}\le z\Biggr)\xrightarrow[]{n  \xrightarrow[]{\rho,\eta, \kappa}\infty} \Phi(z),
\end{aligned}
\end{equation}
where ${V}_{n,+}=V_{-}+\frac{\rho\Theta\Tr(\BA_{n+1}^2)}{M}$. Therefore, we can obtain
\begin{equation}
\begin{aligned}
&\Prob(\sqrt{Mn}(I_{N,L,M}^{(n+1)}-\overline{C}(\sigma^2))\le r-\zeta)
=
\Prob\Biggl(\frac{\sqrt{Mn}(I_{N,L,M}^{(n+1)}-\overline{C}(\sigma^2))}{\sqrt{{V}_{n,+}}}\le \frac{r-\zeta}{\sqrt{{{V}_{n,+}}}}\Biggr)
\\
&\overset{(a)}{\ge} 
\begin{cases}
\Prob\Bigl(\frac{\sqrt{Mn}(I_{N,L,M}^{(n+1)}-\overline{C}(\sigma^2))}{\sqrt{V_{n,+}}}\le \frac{r-\zeta}{\sqrt{V_{-}}} \Bigr)~,r\le 0
\\
\Prob\Bigl(\frac{\sqrt{Mn}(I_{N,L,M}^{(n+1)}-\overline{C}(\sigma^2))}{\sqrt{V_{n,+}}}\Bigr)\le 0 ~,r>0
\end{cases}
=
\begin{cases}
\Phi(\frac{r-\zeta}{\sqrt{V_{-}}})+c_{n}, ~~r\le 0
\\
0.5+c_{n},~~r>0
\end{cases}
\end{aligned}
\end{equation}
for the sequence $c_{n} \downarrow 0$. The inequality in step $(a)$ holds true since $V_{n,+}\ge V_{-}>0$ and the case for $r>0$ follows from the fact that
$\frac{\Tr(\BA_{n+1}^2)}{M}=\BO(M)$ such that $V_{n,+} \rightarrow \infty$. The lower bound can be obtained by taking the limit $ \downarrow 0$.

\textit{Upper bound}: By the upper bound in~(\ref{upp_bound}), we provide an exact implementation of $\BX^{(n)}$, which is constructed by the normalized Gaussian codebook, i.e.,
\begin{equation}
\label{code_cons}
\BX^{(n)}=\frac{\BG}{\sqrt{\frac{\Tr(\BG\BG^{H})}{Mn}}},
\end{equation}
where $\BG\in \mathbb{C}^{M\times n}$ is a Gaussian random matrix with zero-mean and unit-variance entries. This indicates that
\begin{equation}
\label{trace_A2_as}
\begin{aligned}
\frac{\Tr(\BA^2_{n})}{M}= 1-\frac{2\Tr(\BX^{(n)}(\BX^{(n)})^{H})}{Mn}
+\frac{\Tr((\BX^{(n)}(\BX^{(n)})^{H})^2)}{Mn^2} \xrightarrow[]{a.s.} \rho^{-1}.
\end{aligned}
\end{equation}
In this case, the variance in~(\ref{clt_var}) becomes $V_{+}$.

\subsection{Proof of~(\ref{prob_con_rate})}
To establish the convergence rate of the upper and lower bounds, we will use the Esseen inequality in~(\ref{berry_ineq}). The lower bound can be regarded as $\frac{\Tr(\BA^2)}{M}=0$, which $\BX$ can be constructed by orthogonal basis when $n > M$. In this case, the error term in~(\ref{final_error_term}) can be simplified as
\begin{equation}
\begin{aligned}
E(u,\BA) &=\BO\Bigl(\frac{\mathcal{M}(V_n,u)}{M}+\frac{ V_n\mathcal{M}(V_n,u)}{M^2}+ \frac{u\mathcal{M}(V_n,u)}{M}+\frac{u^2 \mathcal{M}(V_n,u)}{M^2}+\frac{u^3\mathcal{M}(V_n,u)}{M^3}
\\
&
+
\frac{u}{M} +\frac{u \mathcal{M}(V_n,u) }{M}+\frac{u^2 \mathcal{M}(V_n,u)}{M^{2}}+
\frac{\mathcal{M}(V_n,u)}{ M}+(\frac{u}{M}+\frac{u^2}{M^2})\mathcal{M}(V_n,u)\Bigr).
\end{aligned}
\end{equation}
 with the dominating term in the error term is $\BO(\frac{u}{M})$. By taking $T=n^{\frac{1}{2}}$ in~(\ref{berry_ineq}), we have
 \begin{equation}
\sup\limits_{x\in \mathbb{R}}\Bigl|\Prob\Bigl(\sqrt{\frac{{Mn}}{V_n}}(I_{N,L,M}^{(n)}-\overline{C}(\sigma^2)) \le x  \Bigr)-\Phi(x) \Bigr|=\BO(n^{-\frac{1}{2}}).
 \end{equation}
 For the upper bound,  by the construction of $\BX$ in~(\ref{code_cons}), we have $\frac{\|\BG\BG^{H}\|}{n}$ is bounded almost surely such that 
\begin{equation}
 \| \Tr(\BA^4) \|\le \| \Tr (\BI_{M})   \| +4 n^{-1}\|  \Tr (\BX\BX^{H})   \| +6 n^{-2} \| \Tr ((\BX\BX^{H})^2 )  \|+4n^{-3} \| \Tr( (\BX\BX^{H})^3 )  \| + n^{-4}\| \Tr( (\BX\BX^{H})^4 )  \| 
 =\BO(M).
 \end{equation}
 With the evaluation above, the variance control in~(\ref{AA_con_var}) can be further improved to be $\BO(\frac{\Tr\BA^4}{M^3})=\BO(M^{-2})$ such that the error term in~(\ref{final_error_term}) can be simplified as
\begin{equation}
\label{final_error_term2}
\begin{aligned}
E(u,\BA)&=\BO\Bigl(\frac{\mathcal{M}(V_n,u)}{M}+\frac{ V_n\mathcal{M}(V_n,u)}{M^2}+ \frac{u\mathcal{M}(V_n,u)}{M}+\frac{u^2 \mathcal{M}(V_n,u)}{M^2}+\frac{u^3\mathcal{M}(V_n,u)}{M^3}
\\
&+
\frac{u}{M} 
+\frac{u \mathcal{M}(V_n,u)}{M}+\frac{u^2 \mathcal{M}(V_n,u)}{M^{2}}
\frac{\mathcal{M}(V_n,u)}{ M}
+[\frac{1}{M} 
+ \frac{u}{M}
+\frac{u^2}{M^2}]\mathcal{M}(V_n,u)\Bigr),
\end{aligned}
\end{equation}
since $\frac{\Tr\BA^2_{n}}{M}\xrightarrow[\rho]{a.s.} \rho^{-1}$ in~(\ref{trace_A2_as}). In this case, by taking $T=n^{\frac{1}{2}}$ in~(\ref{berry_ineq}), we have
\begin{equation}
\sup\limits_{x\in \mathbb{R}}\Bigl|\Prob\Bigl(\sqrt{\frac{{Mn}}{V_n}}(I_{N,L,M}^{(n)}-\overline{C}(\sigma^2)) \le x  \Bigr)-\Phi(x) \Bigr|=\BO(n^{-\frac{1}{2}}).
 \end{equation}
 This concludes~(\ref{lower_exp}) and~(\ref{upper_exp}).
\end{proof}

\section{Proof of Proposition~\ref{equ_pro}}
\label{equ_pro_proof}
\begin{proof}
Letting $z=\sigma^2$ and $M=N$ in the first line of~(\ref{delta_def}), we have
$\delta=\frac{1}{z}(\frac{N}{L}-\frac{M\omega}{L(1+\omega)} )=\frac{\kappa\overline{\omega}}{z}$. Also, we have the following results
\begin{equation}
\begin{aligned}
\Delta_{\sigma^2}
=z+ \frac{z\omega\overline{\omega}}{1+\frac{\kappa}{z} \overline{\omega}^3}
=\frac{z+\kappa\overline{\omega}^3+z\omega\overline{\omega} }{1+\frac{\kappa}{z} \overline{\omega}^3}
=\frac{z(1+2z\omega^3+2z\omega^2)}{1+z\omega^2(1+\omega)},
 \end{aligned}
  \end{equation}
   \begin{equation}
  \begin{aligned}   
  &\frac{\omega^2(1+\delta\overline{\omega})}{1+\delta\overline{\omega}^2}=\frac{ \omega^2(z(1+\omega)^3+\kappa (1+\omega) ) }{z(1+\omega)^3+\kappa}
=\frac{ \omega^2(z(1+\omega)^2+\omega+1+\kappa ) }{z(1+\omega)^2+(1-\kappa)\omega+1+\kappa}.
    \end{aligned}
    \end{equation}  
  Then, by~(\ref{delta_deri}), we can obtain
  \begin{equation}
  \begin{aligned}
& \delta'=-\frac{\delta}{\Delta_{\sigma^2}}
=-\frac{\kappa(1+z\omega^2(1+\omega))}{z^2(1+\omega)(1+2z\omega^3+2z\omega^2)},
  \\
      &\omega'=\frac{\omega \delta'}{\delta(1+\delta \overline{\omega}^2)} 
  =\frac{z^2\omega(1+\omega)^4 \delta' }{\kappa(z(1+\omega)^3+\kappa)}
  =\frac{[1+(1-\kappa)\omega]^2 \delta'}{\kappa(z\omega(1+\omega)^3+\kappa\omega)}
  =\frac{[1+(1-\kappa)\omega]^2 \delta'}{\kappa[(1-\kappa)\omega^2+2\omega+1]}.
  \end{aligned}
    \end{equation}
Let $X(\omega)=\frac{z^2\delta'}{\kappa}$ and $Y(\omega)=\kappa\overline{\omega}^4
[\frac{\omega^2(1+\delta\overline{\omega})}{1+\delta\overline{\omega}^2}-\frac{\omega\omega'}{\delta(1+\delta \overline{\omega}^2)}]$, we can complete the proof by substituting the above results into~(\ref{var_upp_low}).   
\end{proof}

\section{Proof of Proposition~\ref{high_snr_pro}}
\label{high_snr_pro_proof}
\begin{proof} 
The sketch of the proof is given as follows. We can first obtain the high-SNR approximation for $\omega$ by analyzing the dominating terms of~(\ref{cubic_eq}) and the high SNR approximation for $\delta$ by~(\ref{delta_def}), which are then used to derive the high SNR approximation for $V_{-}$ and $V_{+}$ in~(\ref{var_upp_low}). Defining $\varrho=\sigma^{-2}$ to represent the SNR, we will handle each case as follows.


\textit{Case 1:$N>M$ and $L>M$}. We start by deriving the high SNR approximation for $\omega$ by analyzing the dominating terms in~(\ref{cubic_eq}). Observe that $(\eta-1)(\kappa-1)\varrho\omega^2$, $(\eta\kappa-2\eta+1)\varrho\omega$, and $-\eta\varrho$ are negative terms since $\eta>1$ and $\kappa<1$. The negative terms must be compensated by the positive terms $\omega^3$, $\omega^2$, and $\omega$ to hold the equality. We will use this fact to determine the order of $\varrho$ in $\omega$ and the coefficients of $\varrho$. If $\omega$ is $\BO(1)$, the $\BO(\varrho)$ terms can not be compensated. If $\omega$ has a higher order than $\varrho$, i.e., $\omega=\Theta(\varrho^{1+\varepsilon})$, the left hand side of~(\ref{cubic_eq}) will grow to infinity since $\omega^3=\Theta(\varrho^{3+3\varepsilon})$, $\varepsilon>0$, which can not be compensated by the negative terms. Therefore, $\omega$ is $\Theta(\varrho)$ such that $\omega^3$ is the highest-order positive term and must be cancelled by the highest-order negative term, i.e., $(\eta-1)(\kappa-1)\varrho\omega^2$. Therefore, we can obtain $\omega=-(\eta-1)(\kappa-1)\varrho+\BO(1)$. $\delta=(\eta\kappa-\kappa)\varrho+\BO(1)$ can be obtained by approximating $\omega$ using~(\ref{delta_def}). By~(\ref{delta_deri}), we further have $\delta'=-(\eta\kappa-\kappa)\varrho^2+\BO(\varrho)$, and $\omega'=-(\eta-1)(1-\kappa)\varrho^2+\BO(\varrho)$, such that 
$V_{-}=-\rho\log((1-\kappa)(1-\frac{1}{\eta}))+1+\BO(\varrho^{-1})$ and $V_{+}=V_{-}+\BO(\varrho^{-1})$.

When $M>N$, $\omega$ should be $\BO(1)$, due to the constraint $(\eta-1)\omega+\eta>0$. The dominating term is $[(\eta-1)(\kappa-1)\omega^2+(\eta\kappa-2\eta+1)\omega-\eta ]\varrho$ and $\omega$ can be obtained by letting the coefficient of $\varrho$ be zero such that $\omega\in \{\frac{1}{\kappa-1},-\frac{\eta}{\eta-1}\}$.

\textit{Case 2:$M>N$ and $L>N$}. The high SNR approximation for $\omega$ can be obtained by a similar analysis as~\textit{Case 1} and is omitted here. In this case, $\omega=\frac{1}{\kappa-1}+o(1)$ is not feasible since $(\eta-1)\omega+\eta=\frac{\eta\kappa-1}{\kappa-1}<0$. Therefore, we have the approximations $\omega=-\frac{\eta}{\eta-1}+\BO(\varrho^{-1})$, $\delta=(1-\eta)^{-1}(\frac{1}{\eta\kappa}-1)^{-1}+\BO(\varrho^{-1})$, $\omega'=\frac{\omega}{\delta(1+\delta \overline{\omega}^2)}+\BO(\varrho^{-1}), \delta'=-\frac{\eta}{(1-\eta)^3(1-\eta\kappa)}+\BO(\varrho^{-1})$ and 
\begin{equation}
\begin{aligned}
&\overline{\omega}^4
\Bigl(\frac{\kappa\omega^2(1+\delta\overline{\omega})}{1+\delta\overline{\omega}^2}-\frac{\kappa\omega\omega'}{\delta(1+\delta \overline{\omega}^2)}\Bigr)
\\
&
=(1-\eta)^4[\frac{\kappa\eta^2}{(1-\eta)^2(1-\eta^2\kappa)}+\frac{\eta(1-\eta\kappa)}{(1-\eta)^3(1-\eta^2\kappa)}]+\BO(\varrho^{-1})
\\
&=(1-\eta)\eta+\BO(\varrho^{-1}),
\end{aligned}
\end{equation}
such that
$V_{-}=-\rho\log((1-\eta)(1-{\eta\kappa}))+\eta+\BO(\varrho^{-1})$ and $V_{+}=V_{-}+\eta(1-\eta)+\BO(\varrho^{-1})$.


\textit{Case 3:$N> M$ and $M>L$}. In this case, we have the approximations $\omega=\frac{1}{\kappa-1}+o(1)$, $\delta=(\eta\kappa-1)\varrho+\BO(1)$, $\delta'=-(\eta\kappa-1)\varrho^2+\BO(\varrho)$, and $\omega'=\BO(1)$, such that
$V_{-}=-\rho\log((1-\frac{1}{\kappa})(1-\frac{1}{\eta\kappa}))+\frac{1}{\kappa}+\BO(\varrho^{-1})$ and $V_{+}=V_{-}+\frac{(\kappa-1)}{\kappa^2}+\BO(\varrho^{-1})$.

\textit{Case 4:$M>N$ and $N>L$}. In this case, by the analysis before~\textit{Case 2}, we have $\omega=\BO(1)$. If $\omega=-\frac{\eta}{\eta-1}+o(1)$, $\delta=(1-\eta)^{-1}(\frac{1}{\eta\kappa}-1)^{-1}+o(1)<0$ for large $\varrho$, which is not feasible. Therefore, $\omega=\frac{1}{\kappa-1}+\BO(\varrho^{-1})$ and the result for this case coincides with that of \textit{Case 3}.
\end{proof}
\section{Proof of Proposition~\ref{low_snr_app}}
\label{low_snr_app_proof}
\begin{proof}Letting $z=\sigma^2$, we have $\delta=\eta\kappa z^{-1}+\BO(z^{-2})$ and $\omega=\eta z^{-1}+\BO(z^{-2})$. By further analyzing the dominating terms in~(\ref{cubic_eq}), we can obtain
\begin{equation}
\begin{aligned}
\omega&=\eta z^{-1}-\eta(1+\eta\kappa)z^{-2}+\BO(z^{-3}),
\\
\delta&=\eta\kappa z^{-1}-\eta\kappa z^{-2}+\BO(z^{-3}),
\\
\delta'&=-\eta\kappa z^{-2}+2(1+\eta\kappa)\eta\kappa z^{-3}+\BO(z^{-4}).
\end{aligned}
\end{equation}
Thus, (\ref{v_low}) can be obtained by noticing that $-\log(\Xi)=\BO(z^{-4})$ and $\overline{C}(\sigma^2)$ can be approximated by
\begin{equation}
\overline{C}(\sigma^2)\approx (\eta-\frac{1}{\kappa})\sigma^{-2}+\eta z^{-1}-2\eta \sigma^{-2}+\kappa^{-1}(1+\eta\kappa)\sigma^{-2}
=\eta \sigma^{-2}+\BO(\sigma^{-4}).
\end{equation}
\end{proof}

\ifCLASSOPTIONcaptionsoff
  \newpage
\fi



%
\bibliographystyle{IEEEtran}
\bibliography{IEEEabrv,ref}
\end{document}